\documentclass[12pt]{article}
\usepackage[margin=1in]{geometry}
\usepackage{geometry}
\usepackage{verbatim}
\usepackage{mathtools}
\usepackage{amsmath}
\usepackage{amsthm}
\usepackage{amssymb}
\usepackage{graphicx}
\usepackage{setspace}

\makeatletter

\providecommand{\tabularnewline}{\\}

\theoremstyle{definition}
\newtheorem{defn}{\protect\definitionname}
\theoremstyle{plain}
\newtheorem{thm}{\protect\theoremname}
\theoremstyle{remark}
\newtheorem{claim}{\protect\claimname}

\usepackage{verbatim}

\usepackage{times}
\usepackage{subfigure}

\usepackage[font=scriptsize]{caption}

\theoremstyle{definition}
\theoremstyle{plain}

\providecommand{\definitionname}{Definition}
\providecommand{\remarkname}{Remark}
\providecommand{\theoremname}{Theorem}

\global\long\def\alp{\alpha}

\global\long\def\eps{\varepsilon}
\global\long\def\lam{\lambda}

\global\long\def\sig{\sigma}
\global\long\def\Sig{\Sigma}
\global\long\def\vrp{\varphi}

\global\long\def\Omg{\Omega}

\global\long\def\til#1{\tilde{#1}}

\global\long\def\what#1{\widehat{#1}}

\global\long\def\R{\mathbb{R}}

\global\long\def\zero{\boldsymbol{0}}

\global\long\def\fb{\boldsymbol{f}}

\global\long\def\sbb{\boldsymbol{s}}

\global\long\def\VV{\boldsymbol{V}}

\global\long\def\xx{\boldsymbol{x}}

\global\long\def\yy{\boldsymbol{y}}

\global\long\def\betaa{\boldsymbol{\beta}}
\global\long\def\epss{\boldsymbol{\eps}}

\global\long\def\ds{\text{{diag(\ensuremath{\sbb})}}}
\global\long\def\BC{\text{Barber and Cand\`es}}

\global\long\def\diag{\text{{diag}}}

\newcommand{\suppfig}{Supp.~Fig.~}
\newcommand{\suppfigs}{Supp.~Figs.~}
\newcommand{\suppsec}{Supp.~Sec.~}
\newcommand{\suppsecs}{Supp.~Secs.~}

\makeatother

\providecommand{\claimname}{Claim}
\providecommand{\definitionname}{Definition}
\providecommand{\theoremname}{Theorem}

\begin{document}
\title{Controlling the FDR in variable selection via multiple knockoffs}
\author{Kristen Emery and Uri Keich\\
 School of Mathematics and Statistics F07\\
 University of Sydney\\
}
\maketitle
\begin{abstract}
\BC~recently introduced a feature selection method called knockoff+
that controls the false discovery rate (FDR) among the selected features
in the classical linear regression problem. Knockoff+ uses the competition
between the original features and artificially created knockoff features
to control the FDR \cite{barber:controlling}. We generalize \BC'
knockoff construction to generate multiple knockoffs and use those
in conjunction with a recently developed general framework for multiple
competition-based FDR control \cite{emery:multiple2}.

We prove that using our initial multiple-knockoff construction the
combined procedure rigorously controls the FDR in the finite sample setting.
Because this construction has a somewhat limited utility we introduce a heuristic
we call ``batching'' which significantly improves the power of our multiple-knockoff
procedures.

Finally, we combine the batched knockoffs with a new context-dependent resampling
scheme that replaces the generic resampling scheme used in the general multiple-competition setup.
We show using simulations that the resulting ``multi-knockoff-select'' procedure  empirically controls
the FDR in the finite setting of the variable selection problem while
often delivering substantially more power than knockoff+.
\end{abstract}

\section{Introduction}

When using the classical linear regression model we posit that the
observed response vector $\yy\in\R^{n}$ satisfies
\begin{equation}
\yy=X\betaa+\epss,\label{eq:lin_model}
\end{equation}
where $X$ is the $n\times p$ known, real-valued, design matrix,
$\betaa\in\R^{p}$ is the unknown vector of coefficients, and $\epss\sim N(0,\sig^{2}I)$
is Gaussian noise. This model is ubiquitously utilized in many fields
of science when trying to explain observed response measurements using
a large number of potential explanatory features. A critical question
that scientists face when using the model is that of variable, or
model selection: which of the explanatory features (columns of $X$)
should be included in the model and which should not (e.g.,~\cite{James2013}).

Recently, G'Sell et al.~suggested using the notion of false discovery
rate (FDR) as a way of gauging and hence controlling the quality of
a selected set of variables~\cite{gsell:sequential}. Originally
introduced by Benjamini and Hochberg in the context of multiple hypotheses
testing~\cite{benjamini:controlling}, in our model selection context
FDR amounts to the expected proportion of the variables that were
erroneously added to the model.

Soon afterwards \BC\ introduced their knockoff+ procedure (KO+)
that rigorously controls the FDR in the finite variable selection
context~\cite{barber:controlling}. Briefly, knockoff+ relies on
introducing an $n\times p$ knockoff design matrix $\til X$, where
each column consists of a knockoff copy of the corresponding original
variable. These knockoff variables are constructed so that in terms
of the underlying regression problem the true null features (the ones
that are not included in the model) are in some sense indistinguishable
from their knockoff copies. The procedure then assigns to each null
hypothesis $H_{i}:\beta_{i}=0$ two test statistics $Z_{i},\til Z_{i}$
which correspond to the point $\lambda$ on the Lasso path \cite{tibshirani:regression}
at which feature $X_{i}$, respectively, its knockoff competition
$\tilde{X}_{i}$, first enters the model when regressing the response
$\yy$ on the augmented design matrix $\left[X\tilde{X}\right]$.\footnote{The knockoff+ procedure can utilize other statistics that satisfy
a certain exchangeability condition but the one presented here is the focus of \cite{barber:controlling}.}
The intuition here is that generally $Z_{i}>\til Z_{i}$ for true
model features, whereas for null features, $Z_{i}$ and $\til Z_{i}$
are identically distributed.

It is the competition between each $Z_{i}$ and its corresponding
$\til Z_{i}$ that allows \BC\ to define a selection procedure that
controls the FDR. Formally, this is done through their rigorous (Selective)
SeqStep+ procedure but in essence it is based on their ability to
estimate the FDR among the list of top $k$ original variable wins
in $T=\left\{ i\,:\,Z_{i}>\til Z_{i}\right\} $ using the number of
knockoff wins ($Z_{i}<\til Z_{i}$). Specifically, if we let $Z_{k}^{*}$
denote the score $Z_{i}$ of the $k$th largest feature in $T$, then
(ignoring possible ties) the FDR among the top $k$ features in $T$
is estimated as the ratio of (one plus) the number of knockoff wins
$\til Z_{i}>\max\{Z_{k}^{*},Z_{i}\}$ to $k$. The knockoff+ procedure
selects the largest subset of top $k=k(\alp)$ features in the set $T$ of all original feature
wins so that the above estimated FDR is still $\le\alp$.

Thus, at its core knockoff+ implements FDR control via competition
which applies in a much more general setting. Indeed, exactly this
kind of competition based FDR control has been widely used in computational
mass spectrometry for over a decade using the alternative terminology
of target vs.~decoy instead of original vs.~knockoff~ \cite{elias:target,cerqueira:mude,jeong:false,elias:target2}.

In their paper \BC\ suggest that creating multiple knockoffs for
each feature could potentially increase the power of knockoff+ ---
something that has since been done in other contexts of competition
based FDR control. Specifically, \cite{keich:progressive,keich:averaging}
utilize multiple decoys in the context of the spectrum identification
problem and Emery et al.~offer a more powerful approach to FDR control
in the context where for each observed score $Z_{i}$ we can generate
a small number of independent decoy scores $\til Z_{i}^{j}$ $j=1,\dots,d$
\cite{emery:multiple2}. Emery et al.~also point out that their
approach applies in a more general setting where the decoys satisfy
an extension of the ``null exchangeability'' of~\cite{barber:controlling}
that we will revisit below.

In attempting to create multiple knockoffs to which we can apply the
procedures of Emery et al.~we face several challenges. First, the
knockoff variables that \BC\ introduced do not allow an obvious
generalization to multiple knockoffs. Their paper discusses creating
a single deterministic knockoff for each variable, and while their
published code has a knockoff randomization option, the resulting
true null knockoffs are not independent of one another, nor do they
satisfy the aforementioned null exchangeability. Second, as we will
see below, the more intuitive approach for generalizing \BC' construction
to creating multiple knockoffs suffers from reduced power and limited
applicability. Here we explore remedying this loss of power and applicability
by introducing a heuristic that we refer to as ``batched knockoffs.''
The idea behind batching is that while we need to create all the $k$
knockoffs of each feature at the same time, we might not need to create
the knockoffs for all features at the same time.

Applying Emery et al.'s FDR controlling procedures to the batched
knockoffs we find empirically that the combined procedures seem to
maintain control of the FDR in the variable selection problem. Moreover,
their overall recommended procedure, LBM, often enjoys a non-negligible
power advantage over knockoff+. A critical component of LBM is its
resampling approach to determining the values of two tuning parameters
that are then used in conjunction with their mirandom mapping
to define their selection procedure. The resampling strategy of LBM
is constrained to fit the general context of independent or exchangeable
knockoffs/decoys but in the specific context of linear regression
we can do better. Indeed, we propose an alternative resampling strategy
that makes use of the underlying linear regression model to select
the same tuning parameters. We provide empirical evidence that using
this so-called model-aware resampling yields a more powerful procedure
that seemingly still controls the FDR even when we use it to optimize
not just the tuning parameters but the number of knockoffs as well.

\section{Constructing multiple knockoffs (I)\label{sec:multikosI}}

\BC' knockoff construction ensures that the correlation (technically,
inner product) between any two distinct original features remains
unchanged if we replace one or both of those with their knockoff copies.
Thus, in terms of the Lasso, each null variable ($\beta_{j}=0$) is
statistically indistinguishable from its knockoff. At the same time,
their construction tries to minimize the correlation between each
feature and its knockoff so that true variables ($\beta_{j}\ne0$)
would not be too similar to their knockoffs, lest the procedure's
power would be compromised.

Specifically, given the Gram matrix $\Sigma=X^{T}X$\footnote{We adopt the same convention of \cite{barber:controlling} that the
columns of $X$ are normalized so $\diag\left(\Sig\right)\equiv1$.} \BC\ define their set of knockoff features $\tilde{X}$ through
requiring that $\tilde{X}^{T}\til X=\Sigma$ and $X^{T}\til X=\Sigma_{0}$,
where $\Sigma_{0}\coloneqq\Sigma-\ds$, and $\sbb$ is a non-negative
vector that will be specified below. That is, the Gram matrix of the
$n\times2p$ dimensional augmented design matrix $\left[X\tilde{X}\right]$
satisfies 
\[
\left[X\tilde{X}\right]^{T}\left[X\tilde{X}\right]=\begin{bmatrix}\Sigma & \Sigma_{0}\\
\Sigma_{0} & \Sigma
\end{bmatrix}\eqqcolon G.
\]
\BC\ show that these latter equations can be solved if and only
if the vector $\sbb$ is chosen so that the above defined $G$ is
a non-negative definite matrix ($G\succeq0$).

Considering a constant vector $\sbb\equiv s_{0}$, we can minimize
($1-s_{0}$), the correlation between each feature and its knockoff,
by maximizing $s_{0}$ subject to the constraint that $G\succeq0$.
\BC' equi-correlated construction shows this maximization can be
achieved if we choose $s_{0}=2\lambda_{\min}\left(\Sig\right)\wedge1$,
where $\lambda_{\min}\left(\Sig\right)$ is the minimal eigenvalue
of $\Sig$. They then explicitly define a set of knockoff variables
that satisfies the above equations (Equation (2.2) in \cite{barber:controlling}):
\begin{equation}
\tilde{X}=X\left(I-\Sigma^{-1}\ds\right)+\til UC,\label{eq:BC_tilX}
\end{equation}
where $\til U\in\R^{n\times p}$ is an orthonormal matrix whose column
space is orthogonal to that of $X$, and $C^{T}C=2\ds-\ds\Sig^{-1}\ds$.

We next generalize this construction to create $d$ knockoffs per
feature by first finding an augmented $\left(d+1\right)p\times\left(d+1\right)p$-dimensional
Gram matrix $G$ and then finding an $n\times dp$-dimensional solution
$\til X$ for the equation $\left[X\tilde{X}\right]^{T}\left[X\tilde{X}\right]=G$.
Throughout this section we assume $n\ge(d+1)p$ (generalizing \BC'
assumption that $n\ge2p$). We will relax this assumption in Section
\ref{sec:Batched-KOs}.

\subsection{Creating a Gram matrix}

We first demonstrate our construction using $d=2$ knockoffs per feature.
The original matrix G suggests the following $3p\times3p$-dimensional
augmented Gram matrix:

\[
G\coloneqq\begin{bmatrix}\Sigma & \Sigma_{0} & \Sigma_{0}\\
\Sigma_{0} & \Sigma & \Sigma_{0}\\
\Sigma_{0} & \Sigma_{0} & \Sigma
\end{bmatrix},
\]
where $\Sig=X^{T}X$ and $\Sigma_{0}\coloneqq\Sigma-\ds$ as before.
The idea is that now the knockoff matrix will be $\til X=\left[\til X^{1}\,\til X^{2}\right]$,
where each $\til X^{i}$ corresponds to one complete set of knockoff
variables, so that each $\til X^{i}$ behaves exactly as a single
set of \BC' knockoffs. In addition, the correlations between the
two sets of knockoffs are the same as between each one of them and
the original design matrix $X$.

More generally, we define the $\left(d+1\right)p\times\left(d+1\right)p$-dimensional
augmented Gram matrix as a $\left(d+1\right)\times\left(d+1\right)$
block matrix, where each block is a $p\times p$ sub-matrix $B_{ij}$,
where $B_{ii}=\Sig$, and for $i\ne j$ $B_{ij}=\Sig_{0}$. This corresponds
to a knockoff matrix $\til X=\left[\til X^{1}\,\til X^{2}\dots\til X^{d}\right]$
that is made of $d$ blocks/copies $\til X^{i}$, $i=1,\dots,d$,
with the same correlation structure as discussed for the $d=2$ case
above.

We will next show how to construct $\til X$ so that $G$ is indeed
the Gram matrix of the augmented $n\times\left(d+1\right)p$ design
matrix $\left[X\tilde{X}\right]$. However, we can only do that if
$G\succeq0$, which in turn depends on $\sbb$. Again, we consider
the equi-correlated case of $\sbb\equiv s_{0}$, but we can no longer
use the same $s_{0}=2\lambda_{\min}\left(\Sig\right)\wedge1$ that
works for the $d=1$ case. That said, we empirically found that setting
\begin{equation}
s_{0}=\frac{d+1}{d}\lambda_{\min}\left(\Sig\right)\wedge1\label{eq:s0_equi_general}
\end{equation}
yields the optimal result in the general case. That is, with this
critical value, $G\succeq0$, and if $s_{0}<1$ then $G$ is also
rank deficient so $s_{0}$ cannot be any larger than this value. Notably,
this critical value, which generalizes \BC' expression for $d=1$,
decreases with $d$ --- a point we will return to below.

\subsection{Creating the knockoff variables with the given Gram matrix}

The original procedure \eqref{eq:BC_tilX} of deriving $\tilde{X}$
from $X$ is not clearly generalizable to our setting, so instead
we offer the following alternative procedure.

We first find $X_{0}$, a $\left(d+1\right)p\times\left(d+1\right)p$-dimensional
symmetric root of $G$ so that $X_{0}X_{0}=G$. Technically, we accomplish
this by starting with a singular value decomposition (SVD) of $G$:
$G=USV^{T}$, where $S$ is a diagonal matrix and $U,V$ are orthogonal
matrices. Since $G$ is symmetric, the SVD is in fact a spectral decomposition
of $G$: $G=USU^{T}$, so we can define $X_{0}\coloneqq US^{1/2}U^{T}$.

Note that the Gram matrix of the first $p$ columns of $X_{0}$ is
the corresponding $p\times p$ leading sub-matrix of $G$, which is
$\Sig$. Hence, assuming as we do that $n\ge(d+1)p$, there exists
an orthogonal map $\til U:\R^{(d+1)p}\mapsto\R^{n}$ that maps the
first $p$ columns of $X_{0}$ to $X$. Specifically, we can find
such a map by first doing a QR decomposition of $X_{0}$: 
\[
X_{0}=Q_{0}R_{0},
\]
where $Q_{0}$ is a $\left(d+1\right)p\times\left(d+1\right)p$ orthogonal
matrix, and $R_{0}$ is an upper triangular matrix of the same dimension.
We next find a thin QR decomposition \cite{golub:matrix} of 
\begin{equation}
[XA]=QR,\label{eq:QR_on_X}
\end{equation}
where $A$ is an arbitrary $n\times dp$ matrix, $Q$ is an $n\times\left(d+1\right)p$
matrix with orthonormal columns, and $R$ is a $\left(d+1\right)p\times\left(d+1\right)p$
upper triangular matrix. Subject to a sign normalization we discuss
below, the map $\til U$ we seek can be defined by the matrix 
\begin{equation}
\til U\coloneqq QQ_{0}^{T}.\label{eq:orthogonal_map}
\end{equation}

Defining 
\begin{equation}
X_{1}\coloneqq\til UX_{0}=QQ_{0}^{T}X_{0}=QR_{0},\label{eq:X1_from_X0}
\end{equation}
we note that $X_{1}$ is an $n\times\left(d+1\right)p$ matrix and
\[
X_{1}^{T}X_{1}=X_{0}^{T}Q_{0}Q^{T}QQ_{0}^{T}X_{0}=X_{0}^{T}Q_{0}Q_{0}^{T}X_{0}=G.
\]
Moreover, because the Gram matrices of the columns of $X$ and of
the first $p$ columns of $X_{0}$ are the same, and because the QR
decomposition is essentially just the Gram-Schmidt procedure, it follows
that the $p\times p$ leading minor of $R_{0}$ ($R_{0}(i,j)$ for
$i,j\le p$) agrees with $R$ up to row signs, which we can readily
match by adjusting the signs of the corresponding columns of $Q$.

Thus, without loss of generality, the first $p$ columns of $X_{1}$
coincide with the original design matrix $X$, and the next $dp$
columns are our knockoff variables. In other words, $X_{1}$ is the
augmented design matrix, where for each feature $i\in\left\{ 1,\dots,p\right\} $
the $i$th column of $X_{1}$ corresponds to the original $n$ variables,
and columns $i+jp$ for $j=1,\dots,d$ are its $d$ knockoff copies.

\subsection{The knockoff scores and conditional null exchangeability}

We can now describe the (first version of) our procedure for constructing
multiple-knockoff scores. Assuming $n\ge(d+1)p$, the procedure constructs
the $n\times\left(d+1\right)p$ augmented design matrix $\left[X\tilde{X}\right]$
as described above. Following the knockoff+ procedure, it then applies
the Lasso procedure (to $\yy$ and $\left[X\tilde{X}\right]$) to
generate the set of scores $\left\{ \til Z_{i}^{0}\coloneqq Z_{i},\til Z_{i}^{1},\dots,\til Z_{i}^{d}\right\} $
for each feature $i\in\left\{ 1,\dots,p\right\} $. Specifically,
each value is the point $\lambda$ on the Lasso path at which the
corresponding variable, the original $X_{i}$ or its $d$ knockoffs
$\til X_{i}^{j}$, $j=1,\dots,d$, first enters the model.

We next show that our procedure creates knockoff scores that satisfy
the null exchangeability condition of Emery et al.~ and hence applying
their meta-procedure with any pre-determined values of the tuning
parameters $(c,\lam)$ and the mirandom map $\vrp_{md}$ controls
the FDR in the finite variable selection setting~\cite{emery:multiple2}.\footnote{Note that the number of hypotheses here is $m=p$, the number of features
considered.}
\begin{defn}
\label{def:permissible_perm}Let $\Pi_{d+1}$ denote the set of all
permutations on $\left\{ 1,2,\dots,d+1\right\} $ and let $N\subset\left\{ 1,2,\dots,p\right\} $
be the indices of the true null features. A sequence of permutations
$\Pi=(\pi_{1},\dots,\pi_{p})$ with $\pi_{i}\in\Pi_{d+1}$ is a \emph{null-only
sequence} if $\pi_{i}=Id$ (the identity permutation) for all $i\notin N$.
\end{defn}
\begin{thm}
\label{Thm:main}Suppose $\yy$ is generated according to the linear
model (\ref{eq:lin_model}) with a given $n\times p$ design matrix
$X$ with $n\ge(d+1)p$. Let $\VV_{i}\coloneqq\left(\til Z_{i}^{0},\til Z_{i}^{1},\dots,\til Z_{i}^{d}\right)$,
where $\til Z_{i}^{0}\coloneqq Z_{i}$ is the $i$th original feature
score and $\til Z_{i}^{1},\dots,\til Z_{i}^{d}$ are its corresponding
$d$ knockoff scores defined above. For $\pi\in\Pi_{d+1}$ let $\VV_{i}\circ\pi\coloneqq\left(\til Z_{i}^{\pi(1)-1},\dots,\til Z_{i}^{\pi(d+1)-1}\right)$,
i.e., the permutation $\pi$ is applied to the indices of the vector
$\VV_{i}$ rearranging the order of its entries. Then for any null-only
sequence of permutations $\Pi=(\pi_{1},\dots,\pi_{p})$, the joint
distribution of $\VV_{1}\circ\pi_{1},\dots,\VV_{p}\circ\pi_{p}$ is
invariant of $\pi_{1},\dots,\pi_{p}$.
\end{thm}
Note that (a) the conclusion of the theorem is exactly the conditional
null exchangeability of Emery et al.~and (b) that the joint distribution
is the one induced by the Gaussian noise $\epss$ in our linear model
(the design matrix $X$ is fixed).
\begin{proof}
The proof of the theorem uses claims analogous to Lemmas 1, 2 and
3 of \cite{barber:controlling}. Denote by $\hat{X}=\left[X\tilde{X}\right]$
the above $n\times(d+1)p$ augmented design matrix, so that $G=\hat{X}^{T}\hat{X}$,
and by $\hat{X}(i)$ its $i^{th}$ column, so for $i\in\{1,\dots,p\}$
the columns $\hat{X}(i),\hat{X}(i+p),\dots,\hat{X}(i+dp)$ correspond
to the $i^{th}$ feature and its $d$ knockoffs.

For a null-only sequence of permutations $\Pi=(\pi_{1},\dots,\pi_{p})$
let $\hat{X}\circ\Pi$ denote the $n\times(d+1)p$ matrix whose $i^{th}$
column for any $i=i_{0}+i_{1}\cdot p$, where $i_{0}\in\{1,\dots,p\}$
and $i_{1}\in\{0,1,\dots,d\}$, is given by 
\[
(\hat{X}\circ\Pi)(i)\coloneqq\hat{X}(i_{0}+\pi'_{i_{0}}(i_{1})\cdot p),
\]
where $\pi'_{i_{0}}(i_{1})=\pi_{i_{0}}(i_{1}+1)-1$ (note that $i_{0}=(i-1)\mod p+1$
and $i_{1}=(i-i_{0})/p$). In words, the permutation $\pi_{i_{0}}$
is applied to reorder the columns $i_{0},i_{0}+p,\dots i_{0}+dp$
of $\hat{X}$ so their new order is $\pi_{i_{0}}(1),\dots,\pi_{i_{0}}(d+1).$

The first of our claims generalizes Lemma 2 of \BC: applying as above
any sequence of permutations (not necessarily null-only) $\Pi=(\pi_{1},\dots,\pi_{p})$
to the columns of the augmented design matrix does not change the
correlations between its columns.
\begin{claim}
\label{thm3_claim1} $(\hat{X}\circ\Pi)^{T}(\hat{X}\circ\Pi)=\hat{X}^{T}\hat{X}=G$.
\end{claim}
\begin{proof}
Let $i=i_{0}+i_{1}p$ and $j=j_{0}+j_{1}p$, where, as above, $i_{0},j_{0}\in\{1,\dots,p\}$
and $i_{1},j_{1}\in\{0,1,\dots,d\}$. Then, with $\Sig=(\sig_{ij})$,
and $\delta_{i,j}$ the Kronecker delta we have 
\begin{align*}
\hat{X}(i)^{T}\hat{X}(j) & =G_{i,j}=G_{i_{0}+i_{1}p,j_{0}+j_{1}p}=\\
& =\sigma_{i_{0},j_{0}}-\delta_{i_{0},j_{0}}(1-\delta_{i_{1},j_{1}})s_{0}=G_{i_{0}+\pi'_{i_{0}}(i_{1})\cdot p,j_{0}+\pi'_{j_{0}}(j_{1})\cdot p}\\
& =\hat{X}(i_{0}+\pi'_{i_{0}}(i_{1})\cdot p)^{T}\hat{X}(j_{0}+\pi'_{j_{0}}(j_{1})\cdot p)=\left[(\hat{X}\circ\Pi)(i)\right]^{T}\left[(\hat{X}\circ\Pi)(j)\right].
\end{align*}
\end{proof}
The next claim generalizes Lemma 3 of \BC: applying a null-only sequence
of permutations $\Pi$ to the columns of $\hat{X}$ has no effect
on the distribution of $\hat{X}^{T}\yy$.
\begin{claim}
\label{thm3_claim2} $(\hat{X}\circ\Pi)^{T}\yy\overset{d}{=}\hat{X}^{T}\yy$.
\end{claim}
\begin{proof}
As noted by \BC, $\yy=X\betaa+\epss\sim N(X\betaa,\sigma^{2}I),$
and therefore $\hat{X}^{T}\yy\sim N(\hat{X}^{T}X\betaa,\sigma^{2}\hat{X}^{T}\hat{X}),$
and $(\hat{X}\circ\Pi)^{T}\yy\sim N((\hat{X}\circ\Pi)^{T}X\betaa,\sigma^{2}(\hat{X}\circ\Pi)^{T}(\hat{X}\circ\Pi)).$

By Claim \ref{thm3_claim1}, $\hat{X}^{T}\hat{X}=(\hat{X}\circ\Pi)^{T}(\hat{X}\circ\Pi),$
therefore it suffices to show that for $i=1,\dots,p$, 
\begin{equation}
(\hat{X}^{T}X)(i)\cdot\betaa_{i}=((\hat{X}\circ\Pi)^{T}X)(i)\cdot\betaa_{i}.\label{claim_2_eq_star}
\end{equation}

This, again, follows along the lines of \BC: first, clearly \eqref{claim_2_eq_star}
holds for $i$ for which $\betaa_{i}=0$. For $\betaa_{i}\ne0$ we
need to show that the $i$th columns of $\hat{X}^{T}X$ and of $(\hat{X}\circ\Pi)^{T}X$
are identical. Consider the $j$th entry of that column where $j=j_{0}+j_{1}p$,
with $j_{0}\in\{1,\dots,p\}$ and $j_{1}\in\{0,\dots d\}$. Then, 
\begin{enumerate}
\item If $\betaa_{j_{0}}\neq0$ then as $\Pi$ is a null-only sequence of
permutations, $\pi_{j_{0}}=Id$ and therefore 
\[
\hat{X}(j)^{T}X(i)=\hat{X}(j_{0}+j_{1}p)^{T}X(i)=\hat{X}(j_{0}+\pi'_{j_{0}}(j_{1})\cdot p)^{T}X(i)=\left[(\hat{X}\circ\Pi)(j)\right]^{T}X(i),
\]
so \eqref{claim_2_eq_star} holds.
\item Else, $\betaa_{j_{0}}=0$ so $j_{0}\neq i$ and therefore 
\[
\hat{X}(j)^{T}X(i)=\hat{X}(j_{0}+j_{1}p)^{T}X(i)=\sigma_{j_{0},i}=\hat{X}(j_{0}+\pi'_{j_{0}}(j_{1})\cdot p)^{T}X(i)=\left[(\hat{X}\circ\Pi)(j)\right]^{T}X(i),
\]
and again \eqref{claim_2_eq_star} holds.
\end{enumerate}
\end{proof}
We finally generalize Lemma 1 of \BC. Recall that $V_{i}=(\tilde{Z}_{i}^{0}=Z_{i},\tilde{Z}_{i}^{1},\dots,\tilde{Z}_{i}^{k})$
and $V_{i}\circ\pi_{i}=(\tilde{Z}_{i}^{\pi_{i}(1)-1},\dots,\tilde{Z}_{i}^{\pi_{i}(d+1)-1})$. 
\begin{claim}
\label{thm3_claim3} For any null-only sequence of permutations $\Pi$,
$(V_{1},\dots V_{p})\overset{d}{=}(V_{1}\circ\pi_{1},\dots,V_{p}\circ\pi_{p})$.
\end{claim}
\begin{proof}
As explained by \BC, $\{V_{i}\}$ depend only on $\hat{X}^{T}\hat{X}$
and $\hat{X}^{T}\yy$. By Claim \ref{thm3_claim1}, $(\hat{X}\circ\Pi)^{T}(\hat{X}\circ\Pi)=\hat{X}^{T}\hat{X}=G$
and by Claim \ref{thm3_claim2} $\hat{X}^{T}\yy\overset{d}{=}(\hat{X}\circ\Pi)^{T}\yy$.

The result now follows by observing that applying the Lasso to $(\hat{X}\circ\Pi,\yy)$
would produce the vectors $V_{i}\circ\pi_{i}:\hat{X}\hat{\betaa}=(\hat{X}\circ\Pi)(\hat{\betaa}\circ\Pi)$. 
\end{proof}
The last claim completes the proof showing that the joint distributions
of $(V_{1},\dots V_{p})$ and $(V_{1}\circ\pi_{1},\dots,V_{p}\circ\pi_{p})$
are the same.
\end{proof}
As defined, our construction is only applicable when $n\ge\left(d+1\right)p$,
which greatly limits its utility. We can relax this restriction
by using an analog of \BC' extension of their method to the case
where $p\le n<2p$. Namely, as long as $n-p$ is reasonably large
we can estimate $\sig^{2}$, the variance of the noise in \eqref{eq:lin_model},
extend the design matrix $X$ with $(d+1)p-n$ rows of $0$s and extend
the response $\yy$ with $(d+1)p-n$ independent draws from the $N(0,\what{\sig^{2}})$
distribution~\cite{Barber2015}. One problem with this extension
is that the guarantee of the last theorem no longer applies, although
in practice as long as $n-p$ is not very small this did not seem
to be a major issue.

However, the more significant problem we face, regardless of whether
or not an extension is required, is that according to \eqref{eq:s0_equi_general}
$s_{0}$ is decreasing with $d$. Recalling \BC' argument that a
smaller $s_{0}$ leads to a loss of power (because of the increased
correlation between a real variable and its knockoff copies), we see
that as we increase the number of knockoff copies, we reduce the power
associated with each individual copy. In practice, the overall effect
is therefore mixed where the introduction of additional knockoffs
can often reduce power rather than increase it, as we will see later
on. In order to address this problem we next introduce our so-called
batching heuristic.

\section{Batched partial sets of knockoffs or multiple knockoffs (II)\label{sec:Batched-KOs}}

Our batching heuristic consists of partitioning the original set of
features, or their indices $P\coloneqq\left\{ 1,\dots,p\right\} $,
into a disjoint union $P=\cup_{j}I_{j}$ and separately creating the
knockoffs for each subset of features $I_{j}$. This allows us to
reduce the size of the matrix $G$ so that $s_{0}$ can be made larger.
Specifically, we simultaneously create $d$ knockoff variables for
each of the original features $X_{i}$ for $i\in I$, where $I\subset P$.
The $d\cdot|I|$ knockoffs created in this batch will need to have
exactly the same correlations among themselves, as well as with \emph{all}
the original variables, as they have when we create knockoffs for
all the variables at the same time.

In order to do this, we essentially repeat the above procedure for
simultaneously creating the knockoffs for all features but omitting
all uninvolved knockoff features, that is, columns $\tilde{X}_{i+jp}$
with $i\in P\setminus I$. Specifically, we define the augmented design
covariance matrix $G^{I}$ as a $\left(p+d\left|I\right|\right)\times\left(p+d\left|I\right|\right)$
dimensional block matrix made again of $\left(d+1\right)\times\left(d+1\right)$
blocks $B_{ij}$ (of varying sizes), which are defined here for $i,j\in\left\{ 1,\dots,d+1\right\} $
as: 
\begin{align*}
B_{1j} & =\begin{cases}
\Sig & j=1\\
\Sig_{0}^{PI} & j>1
\end{cases} & B_{i1} & =\begin{cases}
\Sig & i=1\\
\Sig_{0}^{IP} & i>1
\end{cases} & B_{ij} & =\begin{cases}
\Sig^{II} & i=j>1\\
\Sig_{0}^{II} & i\ne j\,,\,i\wedge j>1
\end{cases},
\end{align*}
where $A^{IJ}$ is the restriction of the matrix A to the rows specified
by the set $I$ and the columns specified by the set $J$. For example,
if $I=\left\{ 1,2\right\} $ and $d=3$ then 
\[
G^{I}\coloneqq\begin{bmatrix}\Sigma & \Sigma_{0}^{PI} & \Sigma_{0}^{PI} & \Sigma_{0}^{PI}\\
\Sigma_{0}^{IP} & \Sigma^{II} & \Sigma_{0}^{II} & \Sigma_{0}^{II}\\
\Sigma_{0}^{IP} & \Sigma_{0}^{II} & \Sigma^{II} & \Sigma_{0}^{II}\\
\Sigma_{0}^{IP} & \Sigma_{0}^{II} & \Sigma_{0}^{II} & \Sigma^{II}
\end{bmatrix},
\]
where $\Sigma^{II}=\begin{bmatrix}\sig_{11} & \sig_{12}\\
\sig_{21} & \sig_{22}
\end{bmatrix}$, $\Sigma_{0}^{II}=\begin{bmatrix}\sig_{11}-s_{0} & \sig_{12}\\
\sig_{21} & \sig_{22}-s_{0}
\end{bmatrix}$, $\Sigma_{0}^{IP}=\begin{bmatrix}\sig_{11}-s_{0} & \sig_{12} & \dots & \sig_{1p}\\
\sig_{21} & \sig_{22}-s_{0} & \dots & \sig_{2p}
\end{bmatrix}$, and $\Sig_{0}^{PI}=\left(\Sig_{0}^{IP}\right)^{T}$.

We want to construct the $n\times d\left|I\right|$ dimensional knockoff
matrix $\tilde{X}^{I}$, so that the correlation (Gram) matrix of
the partially-augmented design matrix $\left[X\tilde{X}^{I}\right]$
is $G^{I}$. Again, this can be done if we can find $s_{0}$ such
that, with $\Sig_{0}=\Sig-s_{0}\cdot I$, $G^{I}\succeq0$. Of course,
with our new partial knockoff scheme the $s_{0}$ we chose for the
full matrix $G$ in \eqref{eq:s0_equi_general} is no longer optimal.
Indeed, this was our motivation for looking at the partial knockoff
scheme to begin with. Instead, we use a numerical procedure to find
the value $s_{0}$ for which the minimal eigenvalue of $G^{I}=G^{I}\left(s_{0}\right)$
is zero (or $s_{0}=1$ and $G^{I}\succeq0$).

We proceed with constructing the matrix of knockoff variables $\tilde{X}^{I}$
using mostly the same procedure described above to generate the complete
set of knockoff features with a couple of notable differences relating
to the definition of the orthogonal transformation $\til U$ that
maps $X_{0}$ to $X_{1}$~\eqref{eq:X1_from_X0}. When generating
the full set of multiple knockoffs $\til X$ the map $\til U$ is
defined by \eqref{eq:orthogonal_map}, where $Q$ is obtained by applying
the QR factorization to an arbitrary extension $A$ of $X$~\eqref{eq:QR_on_X}.
We found that our batched knockoffs benefit from the following more
elaborate construction of $Q$ that aims at reducing some unwarranted
correlations between the knockoff variables from different batches.

First, possibly using the same extension procedure mentioned above,
we verify that $n\ge\left(d+1\right)p$ (again assuming that initially
$n>p$ and $n-p$ is not too small). We then apply the same thin QR
factorization as in~\eqref{eq:QR_on_X} to create the $n\times\left(d+1\right)p$
matrix $Q_{b}$ with orthonormal columns. Then, when constructing
the batch of knockoffs $\tilde{X}^{I}$ we define the batch-specific
map $\til U$ using a batch specific $Q\coloneqq Q^{I}$, where $Q^{I}$
consists of the first $p$ columns of $Q_{b}$ as well as its $d\left|I\right|$
columns corresponding to the knockoffs associated with batch $I$.
The result is that each batch of knockoffs can be expressed as a linear
combination of the original features and vectors in a batch specific
subspace, where these subspaces are orthogonal to one another, as
well as to the original features subspace. The rest of the procedure
is unchanged.

We stress that batching is a heuristic: in general the resulting knockoffs
do not satisfy the conditional null exchangeability property. In particular
we found that if the number of batches is too large, for example when
each feature defines its own batch, the conditional null exchangeability
could be violated in such a way that our competition based FDR control
can fail (see Section \ref{subsec:Too-many-batches} below for such
an example with $d=1$).

To address this problem we first require that the sets $I_{j}$ are
not too small (in practice we used an average of at least 4 or 5 features
per batch). In addition, to make use of the fact that knockoffs that
share the same batch are guaranteed to retain the same correlation
structure as the corresponding original features we used the following
clustering approach to create the partition that defines the batches:
defining the leaves as the columns of the original matrix $X$ we
first construct an agglomerative hierarchical cluster tree using the
averaged (Euclidean) distance between features as the distance metric
(UPGMA). Then, traversing the tree from its root we determine the
clusters, or our partition, based on the pre-specified number of batches.
Thus, the more correlated the original features are, the more likely
it is that the same correlation would be retained between its knockoffs.
In Section~\ref{subsec:Clustering-can-help} below we given an example
demonstrating the potential advantage clustering defined partition
can offer.

Regardless of how our partition $P=\cup_{j}I_{j}$ is defined, our
revised multi-knockoff construction procedure then applies the above
partial knockoff procedure, using each set of indices, $I_{j}$, at
a time, to create a $\left(p+d\left|I_{j}\right|\right)\times\left(p+d\left|I_{j}\right|\right)$
augmented design matrix $\left[X\tilde{X}^{I_{j}}\right]$. It then
applies the Lasso procedure to this design matrix (and $\yy$) to
obtain the set of scores $\left\{ \til Z_{i}^{0}\coloneqq Z_{i},\til Z_{i}^{1},\dots,\til Z_{i}^{d}\right\} $
for each feature $i\in I_{j}$ \emph{ignoring} the other values for
$i\ne I_{j}$.

\section{Controlling the FDR via multiple knockoffs}

\subsection{General methods to control the FDR using multiple competing scores\label{subsec:General-methods-to}}

Emery et al.~recently introduced several selection procedures that attempt to
control the FDR in a multiple competition setup like the one we have
here. Our methods are all based on a meta-procedure that assigns to
each hypothesis/feature a label $L_{i}\in\left\{ -1,0,1\right\} $
based on the competition between the original variable score $Z_{i}$
and its associated decoy/knockoff scores $\til Z_{i}^{1},\dots,\til Z_{i}^{d}$.
The label is determined by the rank $r_{i}$ of $Z_{i}$ in the combined
list of $d_{1}=d+1$ scores $\left\{ Z_{i},\til Z_{i}^{1},\dots,\til Z_{i}^{d}\right\} $
as well as by the tuning parameters ($c,\lam$). Specifically, $c=i_{c}/d_{1}$
determines the original-win threshold and $\lam=i_{\lam}/d_{1}$ determines
the decoy win threshold:

\[
L_{i}=\begin{cases}
1 & r_{i}\ge d_{1}-i_{c}+1\qquad\text{(original win)}\\
0 & r_{i}\in\left(d_{1}-i_{\lam},d_{1}-i_{c}+1\right)\qquad\text{(ignored hypothesis)}\\
-1 & r_{i}\le d_{1}-i_{\lam}\qquad\text{(decoy/knockoff win)}
\end{cases}.
\]

The selection procedures vary in how they define the tuning parameters ($c,\lam$)
but given the values of those parameters they all rely on the mirandom
map which determines the selected score $W_{i}\in\left\{ Z_{i},\til Z_{i}^{1},\dots,\til Z_{i}^{d}\right\} $
assigned to a feature corresponding to a knockoff win, or $L_{i}=-1$ (in the case of an original win,
$L_{i}=1$, $W_{i}\coloneqq Z_{i}$, and in the case of neither an original nor a knockoff win, $L_{i}=0$, $W_{i}$
is randomly assigned). With the feature scores and labels defined,
our procedures continue similarly to knockoff+: given the FDR threshold
$\alp$ they sort the selected scores $W_{i}$ and report $D(\alp,c,\lam)\coloneqq\left\{ i\,:\,i\le i_{\alp c\lam},L_{i}=1\right\} $,
the list of original feature wins among those top scores, where\footnote{See Section \eqref{sec:An-assessment-of} for an explanation of the
rationale behind \eqref{eq:reject_criterion-general-c-lam}.}

\begin{equation}
i_{\alp c\lam}\coloneqq\max\left\{ i\,:\,\frac{1+\#\left\{ j\le i\,:\,L_{j}=-1\right\} }{\#\left\{ j\le i\,:\,L_{j}=1\right\} \vee1}\cdot\frac{c}{1-\lam}\le\alp\right\} .\label{eq:reject_criterion-general-c-lam}
\end{equation}

Thus, applying any one of our procedures to the combined set of original
and knockoff scores yields a multiple-knockoff procedure that generalizes
\BC' original knockoff approach to controlling the FDR in the variable
selection problem. When $n\ge\left(d+1\right)p$ and our knockoffs
are constructed without batching as in Section \ref{sec:multikosI},
Theorem \ref{Thm:main} here and Theorem 2 of
\cite{emery:multiple2} guarantee that applying our procedure
with its tuning parameters $\left(c,\lam\right)$ predetermined controls
the FDR in the \emph{finite setting} just as \BC' original knockoffs
do. When $n<\left(d+1\right)p$, when we use batching to construct
our knockoffs, or when applying one of our data-driven methods, where
$c$ and $\lam$ are determined from the data, the resulting multiple
knockoff procedure is no longer guaranteed to control the FDR although
in practice the simulations below indicate the variants we consider
here do.

The specific procedures we consider here include the mirror ($c=\lam=1/2$)
and the max method ($c=\lam=1/\left(d+1\right)$) both of which rely
on predetermined values of $\left(c,\lam\right)$. As Emery et al.~pointed
out there is much to be gained from using data-driven approaches to
set the values of the tuning parameters we naturally considered LBM
as well. LBM is the overall recommended procedure for the general
multiple competition setup and it uses a resampling procedure to try
and optimize the values of $(c,\lam)$~\cite{emery:multiple2}.
That resampling strategy is constrained by the assumption that it
is forbiddingly expensive to generate additional decoys and hence
it makes do with the available decoys. In the context of the knockoffs
it is in fact impossible to create additional independent knockoffs
so in that sense LBM is suitable here. However, in this context, using
the underlying linear regression model, we can generate what we call
model-aware bootstrap (or simply model-bootstrap) samples. We next
describe this new resampling technique and how we use it in a new
selection procedure that we call ``multi-knockoff'' that seems much
more suitable for optimally setting $(c,\lam)$. Our last selection
method described below, ``multi-knockoff-select'', also relies on
our new resampling technique but it goes one step further than the
other procedures we consider by trying to determine the optimal number
of knockoff copies $d$.

\subsection{Model-aware resampling and parameter optimization (multi-knockoff)\label{subsec:Model-aware-resampling}}

Our model-aware resampling method adopts the same ``labeled resampling''
procedure of conjectured true/false null labels that was introduced
in our generic bootstrap approach that LBM relies on (Supplementary
Section 6.5 of \cite{emery:multiple2}). Here a conjectured false
null label corresponds to a variable that is conjectured to be included
in the model, and a conjectured true null label to a variable that
is not included in the model. The original algorithm then continued
to resample the indices in the usual bootstrap manner and then randomly
permuted the vector of original and decoy scores for each resampled
index corresponding to a conjectured true null label. Instead, our
new model-resampling scheme first regresses the response variable
on the conjectured included variables and then it uses the resulting
linear model to generate a new sample of the response variable. The
details of our model-aware resampling are provided next.
\begin{enumerate}
\item Determine $\lambda=\lambda_{0}$ from the empirical p-values / ranks
$r_{i}$ of the original variable scores $Z_{i}$ as described in
\suppsec 6.3 of \cite{emery:multiple2}. Note that we randomly break
all ties by first transforming all observed and knockoff scores into
ranks.
\item Run the first two steps of our meta-procedure (Section 3.2 of \cite{emery:multiple2})
with $\lambda=c=\lambda_{0}$ and the mirandom map $\vrp_{md}$ to
assign a score $W_{i}$ and a knockoff/original win label $L_{i}$
to each variable $i=1,\dots,p$. Those values of $W_{i}$ and $L_{i}$
are kept fixed when generating all subsequent bootstrap samples.
\item To generate each of the $m_{b}$ model-aware bootstrap samples, for
$l=1,\dots m_{b}$ do:
\begin{enumerate}
\item Run steps 3-7 of the algorithm described in \suppsec 6.5 of \cite{emery:multiple2}
to sample an indicator vector $\fb\in\left\{ 0,1\right\} ^{p}$ where
$\fb_{i}=1$ if the $i$th variable is conjectured to be part of the
model (false null) and $\fb_{i}=0$ if the $i$th variable is conjectured
to be missing from the model (true null).
\item With $J=J_{l}=\left\{ i\,:\,\fb_{i}=1\right\} $ let $X_{J}=X[:,J]$
be the submatrix of $X$ consisting of the columns specified by the
set $J$ and use standard least square regression to find the coefficient
vector $\betaa_{J}$ that minimizes the residual sum of squares $\left\Vert \yy-X_{J}\betaa_{J}\right\Vert _{2}^{2}$
\item Randomly draw a noise vector $\epss_{l}\in\R^{n}$ from the $N(\zero,I)$
distribution, where $\zero$ is the $n$-dimensional zero vector and
$I$ is the $n\times n$ identity matrix, and define $\yy_{l}=X\betaa_{J}+\hat{\sigma}\epss_{l}$,
where $\hat{\sigma}$ is the standard deviation estimated as in Section
2.1.2 of \cite{barber:controlling} from the residual sum of squares
in the original data. Note that $X$ here is the 0-extended matrix
if $n<\left(d+1\right)p$.
\item We next apply our multiple-knockoff generating procedure to $\yy_{l}$
and $X$ to generate the model-bootstrap sample of $\left\{ \left(Z_{l,i}=\til Z_{l,i}^{0},\til Z_{l,i}^{1},\dots,\til Z_{l,i}^{d}\right)\,:\,i=1,\dots,p\right\} $.
Note that the set of batched knockoff matrices $\tilde{X}^{I_{j}}$
needs to be created only once. Scores are transformed to ranks with
ties randomly broken.
\end{enumerate}
\item Return the set of $m_{b}$ model-bootstrap samples where each sample
is accompanied by the corresponding set $J_{l}$ of the true features.
\end{enumerate}
The model-bootstrap samples are used differently from the cruder bootstrap
samples that LBM relies on. Indeed, the model-aware resamples are
used to directly optimize the number of discoveries (a strategy that
generally fails to control the FDR when applied to the cruder samples).
Specifically, we apply the above general selection procedure (Section
\ref{subsec:General-methods-to}) for each pair of possible $\left(c,\lam\right)$
values with $1/\left(d+1\right)\le c\le\lam\le1/2$ and select the
pair that maximizes the average number of conjectured true discoveries.
After selecting these optimal values for $\left(c,\lam\right)$ our
so-called \textbf{multi-knockoff} procedure again proceeds along the
general outline of our selection methods which applies our meta-procedure
with the mirandom map defining the selected scores $W_{i}$. We provide
below empirical evidence that multi-knockoff is overall significantly
better than what we achieve relying on our previously published methods.

Finally, we can take this one step further and try to optimize the
power by choosing not only the optimal $\left(c,\lam\right)$ for
each fixed number of knockoffs $d$, but optimize over several considered
values of $d$. We do this using the same model-bootstrap samples
as described above. Specifically, we first determine for each considered
number of knockoffs $d$ its optimal setting of $\left(c,\lam\right)$,
that is the values of these parameters that maximize the average number
of conjectured true discoveries, and then we choose the number of knockoffs
that maximizes this average. We refer to this procedure as \textbf{multi-knockoff-select}
and below we offer some empirical evidence for its effectiveness.

\subsection{How many knockoffs to construct?\label{subsec:How-many-knockoffs}}

Note that when applying any of our mirandom-map-based procedures using,
say $d=3$, knockoffs we can in principle arbitrarily select that
number of knockoffs from a larger constructed set of, say $d=7$,
knockoffs per feature. However, recalling that increasing $d$ increases
the similarity between an original feature and each of its individual
knockoffs it is clear that to optimize the power of the competition-based
FDR controlling procedure one should construct as many knockoffs
as one will use. In particular, when considering multiple numbers of
knockoff copies $d$, say $d\in\left\{ 1,3,7\right\} $, we are actually
constructing three different sets of knockoffs, one for each of these
values of $d$ rather than creating $d=7$ knockoffs and selecting
one / three of those.

\section{Empirically assessing the multiple-knockoff procedures}

We performed extensive simulations to examine how our methods behave across a range of different experimental designs.
In particular, we investigated two things:
\begin{itemize}
	\item Whether the knockoffs created with batching still maintain the desired properties required for FDR control (Section \ref{sec:An-assessment-of}).
	\item The performance of our proposed selection procedures in terms of empirical FDR and power (Section \ref{subsec:Assesssing-the-methods}).
	Specifically, we give empirical evidence that our methods essentially control the FDR in the finite sample case, and we demonstrate
	that the proposed model-knockoff-select procedure is overall the most powerful among all the considered methods including knockoff+.
\end{itemize}

\subsection{Simulation setup: generating the datasets and defining the original
and knockoff scores\label{subsec:Simulation-setup}}

We largely adopted the simulation setup of \cite{barber:controlling},
where we repeatedly begin with drawing an $n\times p$ design matrix
$X$. The rows of $X$ are independently sampled from a multivariate
normal distribution with zero mean and one of the following two types
of covariance matrices. The first is the same T\"oeplitz covariance
matrix $\Theta_{\rho}$ in the original setup of \cite{barber:controlling} where for $\rho=0$
the covariance matrix is $\Theta_{0}\coloneqq I_{p}$, the $p$-dimensional
identity matrix, corresponding to no feature correlation, and for
$\rho>0$, $\left(\Theta_{\rho}\right)_{ij}=\rho^{|i-j|}$ which introduces
some feature correlation. We also introduced a second class of covariance
matrices $\Omg_{\rho}$ that are constant $\rho>0$ on the off-diagonal
terms and with a diagonal of 1s.

We next draw $K<p$ indices $i_{1},\dots,i_{K}\in\left\{ 1,\dots,p\right\} $
for which we set $\betaa_{i_{j}}\coloneqq\pm A$, where $A$ is a
fixed amplitude, and the signs are drawn independently and uniformly.
The rest of the values of the coefficient vector $\betaa$ were set
to 0 corresponding to a model with the $K$ features $i_{1},\dots,i_{K}$
(so the corresponding hypotheses $H_{i_1},\dots,H_{i_K}$ are false nulls). Finally, we draw
the noise vector $\epss$ as iid $N\left(0,1\right)$ variates and
we define the response vector $\yy$ through (\ref{eq:lin_model}).

For each such randomly generated pair of a design matrix $X$ and
a response vector $\yy$ we use $b$ batches to construct the set
of the original plus $d$ knockoff scores per feature $\big\{ \big(\til Z_{i}^{0}\coloneqq Z_{i},\til Z_{i}^{1},\dots,\til Z_{i}^{d}\big)\,:\,i=1,\dots,p\big\} $
as described in Section \eqref{sec:Batched-KOs}. Note that even when
we construct a single knockoff set ($d=1$) using a single batch ($b=1$)
it will in practice differ from the one generated by knockoff+ although
the two sets are essentially equivalent.

In Supplementary Section \ref{subsec:Simulation-setup-details} we
provide more details about the specific combination of parameter values
that we used in our simulations for generating the data (design matrix
and response variables) as well as for constructing the knockoffs
(number of knockoffs and batches).

\subsection{An assessment of the batched knockoffs\label{sec:An-assessment-of}}

While we will explicitly examine the FDR control of our competition-based
procedures below, we first examine our knockoffs from a different
perspective. As noted above, our procedures that use a pre-determined
value of $\left(c,\lam\right)$ will control the FDR provided our
knockoff scores satisfy the conditional null exchangeability. However,
this exchangeability is unlikely to apply in general for our batched
knockoffs and moreover it is not a necessary condition.

Emery et al.~argue that if conditional exchangeability holds then
sorting the mirandom-selected scores $W_{i}$ in decreasing order
and applying their general selection procedure with a predetermined $c=\lam$, for any true null feature $j$,
$P\left(L_{j}=1\right)=c$ and $P\left(L_{j}=-1\right)=1-c$
independently of all other features (Section 3.5 and \suppsec 6.9
of \cite{emery:multiple2}). Going back to the critical ratio \eqref{eq:reject_criterion-general-c-lam}
we see that our procedure's control of the FDR hinges on the expected
proportion of original $\left(L_{j}=1\right)$ vs.~knockoff wins
$\left(L_{j}=-1\right)$. Indeed, if there are $i_{0}$ true null
features among the top $i$ scores then the number of original wins
among those is a binomial $\left(i_{0},c\right)$ random variable
(RV), and the number of knockoff wins is the complementary binomial
$\left(i_{0},1-c\right)$. Therefore, when multiplied by the $c/\left(1-c\right)$
factor, the expected value of the numerator of \eqref{eq:reject_criterion-general-c-lam}
bounds $i_{0}\cdot c$ which is the expected number of true null features
among the original wins in the top $i$ scores.

In this section we therefore evaluate the quality of our knockoffs
from this perspective: considering only the true null features, are
the numbers of original score wins among the top $i_{0}$ null features
consistent with a sequence of binomial RVs defined as the cumulative
sum of iid Bernoulli$\left(c\right)$ RVs? A specific concern is when
that observed sequence of true null original wins significantly exceeds
the expected value of the latter, theoretical sequence, because it
would indicate a potential liberal bias in our FDR estimation.

Note that in the case of a single batched knockoff per feature ($d=1$)
we have a related point of reference which is to compare the same
percentage of original wins among the top true null features when
using our batched knockoffs with the corresponding percentage observed
when using \BC' knockoffs. The latter, of course, are guaranteed
to satisfy the conditional null exchangeability so any observed deviations
from the expected 50\% of original wins is due to random fluctuations.

\subsubsection{Too many batches can be problematic\label{subsec:Too-many-batches}}

We used the above mentioned reference point to show the potential
problem with having too many batches. Specifically, we generated 60K
datasets as described in Section~\ref{subsec:Simulation-setup},
each with $p=50$, $n=100$, a covariance matrix $\Theta_{\rho}=I_{p}$
($\rho=0$), $K=1$ feature included in the model and an amplitude
that was deliberately set very high at $A=10.0$. For each of the
60K datasets we used \BC' construction, as well as our batched construction
--- using the maximal possible number of 100 batches, so each batch
contained a single feature --- to generate the sets of original feature
scores $Z_{i}$ with their corresponding knockoff scores $\til Z_{i}$.

With $c=\lam=1/2$ and only one knockoff copy a feature counted as
an original win if $Z_{i}>\til Z_{i}$ (ties were randomly broken)
and the winning scores $W_{i}=\max\left\{ \til Z_{i},\til Z_{i}\right\} $
were sorted in decreasing order, again randomly breaking ties. We
then noted the percentage of target wins among the top $i_{0}$ scores
corresponding to the true null features as we varied $i_{0}$ from
$1$ to $49$ (the score of the single false null feature was not
considered here).

Recall that we evaluate our batched knockoffs against the assumption
that the sequence of proportions we observe is consistent with that
generated by a cumulative sum of iid Bernoulli$\left(c=1/2\right)$
RVs. Under that assumption we can get some idea of whether our batched
knockoffs are consistent with this model by plotting the 97.5\% and
2.5\% quantiles, as well as the mean, of the corresponding binomial
RVs (in practice we used the normal approximation to draw the quantiles).
Keep in mind that these plotted quantiles are only provided for reference:
they are asymptotically only valid pointwise, so even for data that
is consistent with the model the probability that the curve will wander
out of the band outlined by the quantiles is, of course, higher than
5\%.

Judging by panel A of \suppfig\ref{fig:tgt-win-pct} it seems that
in this example where each batch consists of a single feature the
resulting knockoffs exhibit a clear liberal bias: the percentage of
original wins among the top true null features significantly exceeds
our model-determined expected value of 1/2, as well as the variability
we observed in \BC' knockoffs. This bias further manifested itself
in compromised FDR control. For example, applying our batched-knockoff+
(Section \ref{subsec:Assesssing-the-methods}) we find that the empirical
FDR at $\alp=0.5$ is $0.5144$. This 3\% overshoot of the empirical
FDR might not seems that much, however our empirical FDR was computed
from 60K independent samples so statistically it is a very significant
deviation (8.8 standard deviations).%

\subsubsection{Clustering the features can help\label{subsec:Clustering-can-help}}

In practice we found that with an average of five or more features
per batch we avoid the significant bias observed in the example above.
As mentioned, we partition the features into their batches by clustering
them based on the similarities of the corresponding columns of the
design matrix. This clustering based partition typically delivers
only a modest improvement compared with an arbitrary uniform partition
but there are cases where the difference can be significant. To see
that we again consider the effect of batching on a single knockoff
only now our emphasis is on the difference in percentage of target
wins between these two types of partitions: uniform vs.~clustering.

Specifically, we generated two sets of 50K datasets each with $p=200$,
$n=800$, a covariance matrix $\Omg_{\rho}$ with $\rho=0.7$, $K=10$
features included in the model and an amplitude $A=2.8$. Each dataset's
features were partitioned into 40 batches but for the first 50K datasets
we randomly and uniformly assigned 5 features to each batch while
clustering was applied to define the batches of the subsequent 50K
datasets.

Comparing panels B and C of \suppfig\ref{fig:tgt-win-pct} we see
that while clustering based batching creates knockoffs for which the
target wins percentage is in line with our model (panel B, black curve),
the uniformly partitioned batches exhibit an undesirable significant
liberal bias at some point (panel C, black). In both cases we added
for reference the corresponding percentages we observe using \BC'
provably-reliable knockoffs.

\subsubsection{Model-wise the batched multiple knockoffs behave similarly to their
non-batched counterparts\label{subsec:Model-wise-the-batched}}

In light of the above examples, and unless otherwise stated, our batched
knockoffs were generated using clustering with an average of at least
five features per batch. In this section we look specifically at the
effect of batching on the agreement between the observed percentage
of target wins among the true nulls and our model.

We begin with an example that did not require extending $X$: we generated
two sets of 10K datasets, both with $p=200$, $n=800$, using an amplitude
$A=2.8$, $K=10$ features included in the model and a covariance
matrix $\Theta_{\rho}=I_{p}$ ($\rho=0$). We then compared the percentage
of target wins using $d=3$ non-batched knockoffs
with the same percentage when using $d=3$ knockoffs constructed using
40 batches. Panel D of \suppfig\ref{fig:tgt-win-pct} shows that
in this case our batched knockoffs behave similarly to the un-batched
ones. Notably, the latter are guaranteed to follow the model and indeed,
in both cases the percentage of target wins does not deviate significantly
from the theoretical $c=1/2$ (using $c=1/4$ yields qualitatively
similar results).

The next example required extending $X$ because we constructed $d=3$
knockoffs as before but now $p=200$ and $n=600$ so $n<(d+1)p$.
Again, we generated two sets of 10K datasets, one where the knockoffs
were created using 40 batches per dataset and the other using a single
batch per dataset. In this example all features were true null ($K=0$)
and the covariance matrix was $\Theta_{\rho}=I_{p}$. Panel E of \suppfig\ref{fig:tgt-win-pct}
shows that again our batched knockoffs behave similarly to the un-batched
ones, and in both cases the percentage of target wins does not deviate
significantly from the theoretical $c=1/4$ (using $c=1/2$ yields
qualitatively similar results). Note that because $\yy$ was extended
using an estimate of $\sig$ even the un-batched knockoffs are not
guaranteed to follow the model in this case but in practice it seems
they still do.

In our final example we look at a more significant extension of $X$
where we compared our knockoffs constructed in three different ways.
For each of the three we generated 10K datasets using our model with
$p=200$ and $n=600$, all features are true null ($K=0$) and a covariance
matrix $\Theta_{\rho}=I_{p}$. Panel F of \suppfig\ref{fig:tgt-win-pct}
shows that using 40 batches our $d=11$ knockoffs (black curve) demonstrate
a clear liberal bias with $c=2/12$. Interestingly, when using a single
batch to create the same number of $d=11$ knockoffs (red curve) we
observe an even larger liberal bias than the one exhibited by the
batched knockoffs (same value of $c=2/12$). This suggests that the
issue lies with the fairly extreme extension we used rather than
with the batching.\footnote{Note that we needed to extend the response $\yy$ from $n=600$ to $n=2400$
and that it is easy to find examples where any of the knockoff based
procedures considered here, \emph{including} \BC' knockoff+, fails
to control the FDR where one extends $X$ and $\yy$ when $n-p$ is
fairly small.} Indeed, constructing our third set of knockoffs using the known
$\sig=1$, rather than its estimate, to extend $\yy$ we note that
the liberal bias has all but disappeared (green curve). Note that
the three curves of Panel F were generated using the same $c=2/12$
but the results look qualitatively similar using other values of $c=i/12$
with $i\le6$. Regardless of the source of the above liberal bias
we will show below that in practice it is sufficiently mild that it
does not seem to obstruct our ability to control the FDR in the examples
we looked at.

\subsection{Assessing the knockoff selection procedures\label{subsec:Assesssing-the-methods}}

We next investigate and compare the performance of our selection procedures by applying them to randomly drawn datasets. Specifically we considered:
\begin{itemize}
\item \BC' knockoff+, that uses its own single knockoff construction, and
``batched-knockoff+'' which, like knockoff+, uses a single knockoff
but in this case the knockoff is constructed using our batching procedure
(so when the number of batches $b=1$ the two procedures are essentially
equivalent though they can differ substantially when $b>1$).
\item the recently proposed methods of mirror, max and LBM (Section \ref{subsec:General-methods-to}).
\item the new multi-knockoff and multi-knockoff-select that use a pre-specified
number of model-aware bootstrap samples, $m_{b}$ (Section \ref{subsec:Model-aware-resampling}).
\end{itemize}
In \suppsec\ref{subsec:Simulation-setup-details} we provide the
details of the settings that were used by these selection procedures
(e.g., number of bootstrap samples).

We evaluated the performance of each method by noting its empirical
FDR and power as we varied the FDR threshold. Specifically, for each
combination of parameter values we randomly drew (typically) 1K datasets
and for each considered FDR threshold $\alp\in\Phi$ \footnote{For computational efficiency we considered a selected list of FDR
thresholds specified in \suppsec\ref{subsec:The-set-Phi}.} we averaged the FDP in the reported list of discoveries to get the
empirical FDR, and we averaged the percentage of true features in
the same list to get the average power.

We used three types of plots to visually study the selection methods
we consider: power, power-difference and empirical FDR. Each plot
is typically made of multiple curves, where each curve corresponds
to a unique combination of parameter values. Specifically, each curve
summarizes the results obtained by applying, at each considered FDR
threshold, one or two of the methods to (typically) 1K datasets that
were randomly drawn with the same given combination of parameter values,
where:
\begin{itemize}
\item in a power plot ($y$-axis label indicates ``Power'') each curve depicts
a selection method's average power over the randomly drawn datasets.
\item in a power-difference plot ($y$-axis label indicates ``Power Difference'')
a curve represents the difference in average power between the first
and second methods, so negative values indicate the second method
is more powerful at the given FDR threshold.
\item in an empirical FDR plot ($y$-axis label indicates ``FDR'') the curve
yields the ratio between the empirical FDR (average of the FDP) to
the FDR threshold, so a value below 1 indicates a conservative bias
and a value above 1 indicates a liberal bias.
\end{itemize}

\subsubsection{Multiple-knockoff procedures that rigorously control the FDR in
the finite sample case}

Comparing the performance of knockoff+ with that of the multiple-knockoff
procedures when all are guaranteed to control the FDR we see mixed
results. Recall that such finite sample FDR control is guaranteed
when the data is generated according to our model, we construct
our $d\le n/p-1$ knockoffs using a single batch and we apply our
procedure with the mirandom map and pre-determined tuning parameters
(e.g., the mirror and the max methods). Indeed, Theorem \ref{Thm:main}
here and Theorem 2 of \cite{emery:multiple2} guarantee FDR control
in this setting.

Figure \ref{fig:main-panel} (A) shows that in some cases max delivers
significantly more power than knockoff+ while in others it can deliver
substantially less power. \suppfig\ref{fig:supp-rigor-FDR-control}
offers more insight by showing how the power of max and knockoff+
vary with the parameters of the data and the FDR threshold. Overall
max tends to do better for smaller FDR thresholds, sparser models
and a larger $d$ but the results are generally mixed.

\suppfig\ref{fig:supp-rigor-FDR-control-2} shows a summary of
the difference in power between max/mirror/batched-knockoff+ and knockoff+
(left column) as well as the empirical evidence of the corresponding
FDR control (right column). Note that (a) because we use a single
batch in this case, batched-knockoff+ is essentially equivalent to
knockoff+ and the variations in power between them are random, and
(b) mirror is much closer to knockoff+ here than max.

The guaranteed FDR control setup considered here is rather limited.
In practice we would like to apply our methods to the case where $p<n<(d+1)p$.
In addition, as we will see below, we can gain significant power by
learning $c$ and $\lam$ from the data, as well as by using batching
when creating the knockoffs. We empirically explore these extensions
next.

\subsubsection{Batching can significantly increase the power of the knockoff procedures}

Panel B of Figure \ref{fig:main-panel} as well as panels A-D of \suppfig\ref{fig:supp_varying_b_n600}
show examples where, as expected, the power of our procedures generally
increases with the number of batches because we are able to better
distinguish the original features from their knockoffs.

Similarly, panel C of Figure \ref{fig:main-panel} as well as \suppfig\ref{fig:supp_batching_power_1vs40}
show in the context of the various datasets that make the $n=800,p=200,d=3$
set (\suppsec\ref{subsec:n800_p200_d1_3_b1}) that increasing the
number of batches from 1 to 40 typically yields substantial power
gains. This holds for all three procedures we looked at so far: max,
mirror and batched-knockoff+, and for the wide range of parameter
combinations described in \suppsec\ref{subsec:n800_p200_d1_3_b1}.

As expected, batching offers a larger increase in power as $d$ and
$p$ increase. Some evidence of this can be seen in the left column
panels of \suppfig\ref{fig:batching_larger_sets}, which compare
the power of max, mirror and batched-knockoff+ to the power of knockoff+
using $b=1$ and $b=40$ batches: the gains using $b=40$ are significantly
larger when $p$ is increased from 200 to 1000 as well as when $d$
is increased from 3 to 11.

As mentioned in Sections \ref{subsec:Too-many-batches} and \ref{subsec:Model-wise-the-batched}
FDR control can be compromised when introducing batching, and particularly
when a significant extension of $X$ and $\yy$ is involved. Thus,
we should examine whether the significant power gains we see in our
examples when we introduce batching are not attained at the cost of
compromised FDR control. \suppfigs~\ref{fig:supp_varying_b_n600}
(right panels), \ref{fig:supp_varying_b_n3000} (right panels), \ref{fig:supp_batching_power_1vs40_FDR},
and \ref{fig:batching_larger_sets} confirm that the FDR seems to
be properly controlled in spite of the large power gains.

\begin{figure}
	\centering %
	\begin{tabular}{ll}
		A. Max ($d=2,3$) knockoffs vs.~knockoff+ & B. Varying the number of batches (max)\tabularnewline
		\includegraphics[width=3in]{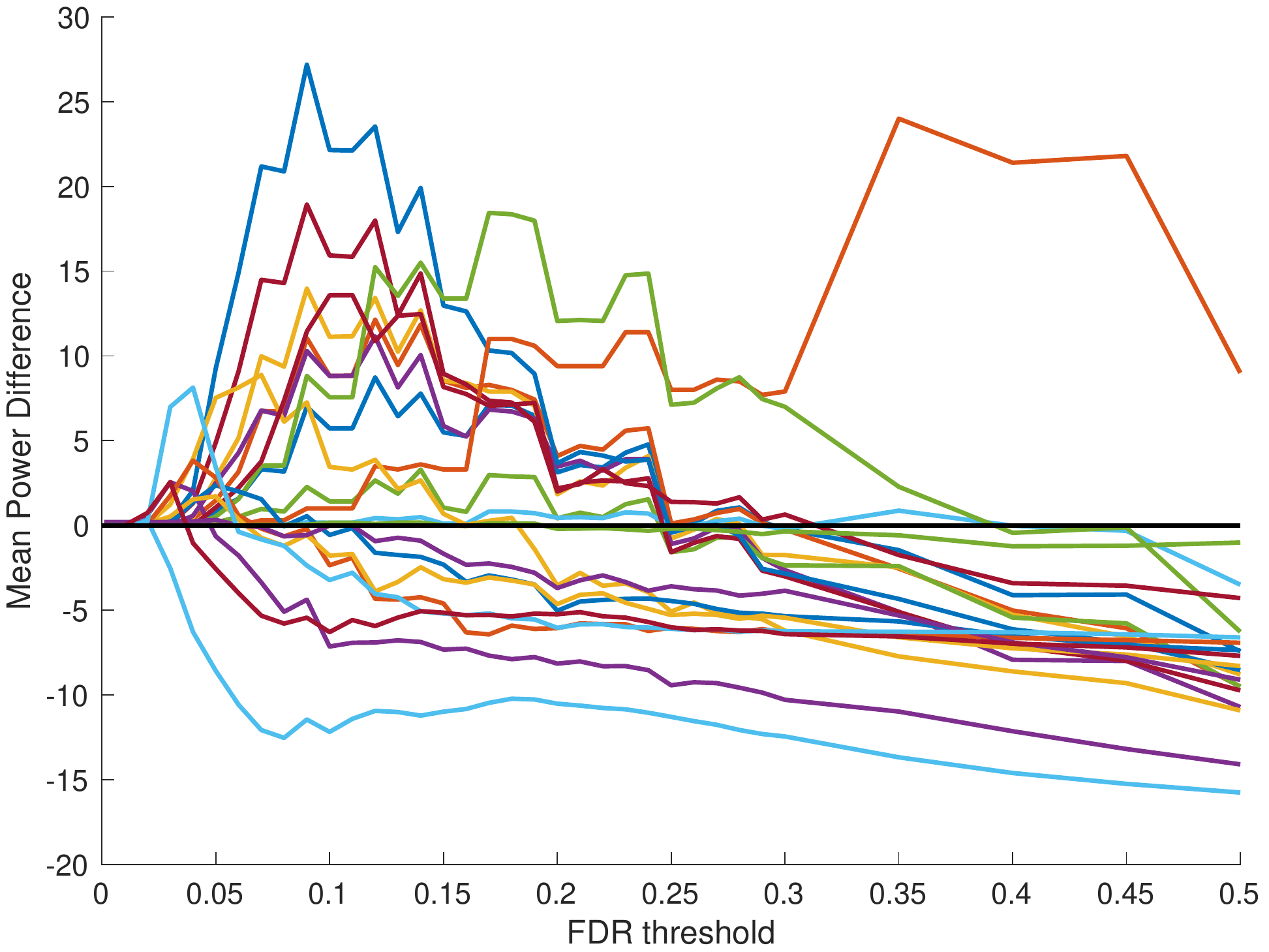} & \includegraphics[width=3in]{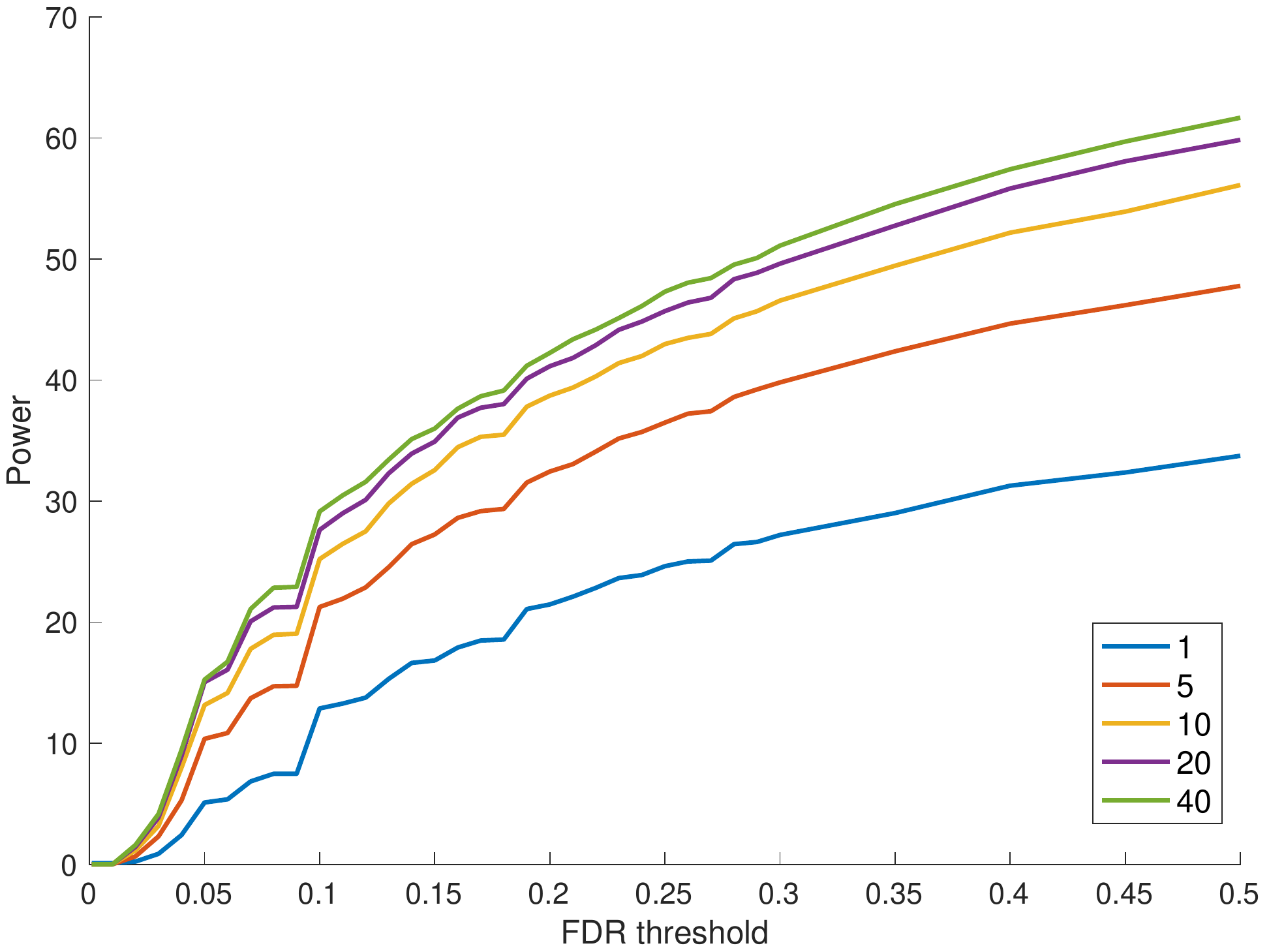}\tabularnewline
		C. Max using $b=1$ vs.~$b=40$ batches & D. LBM vs.~multi-knockoff (combined dataset)\tabularnewline
		\includegraphics[width=3in]{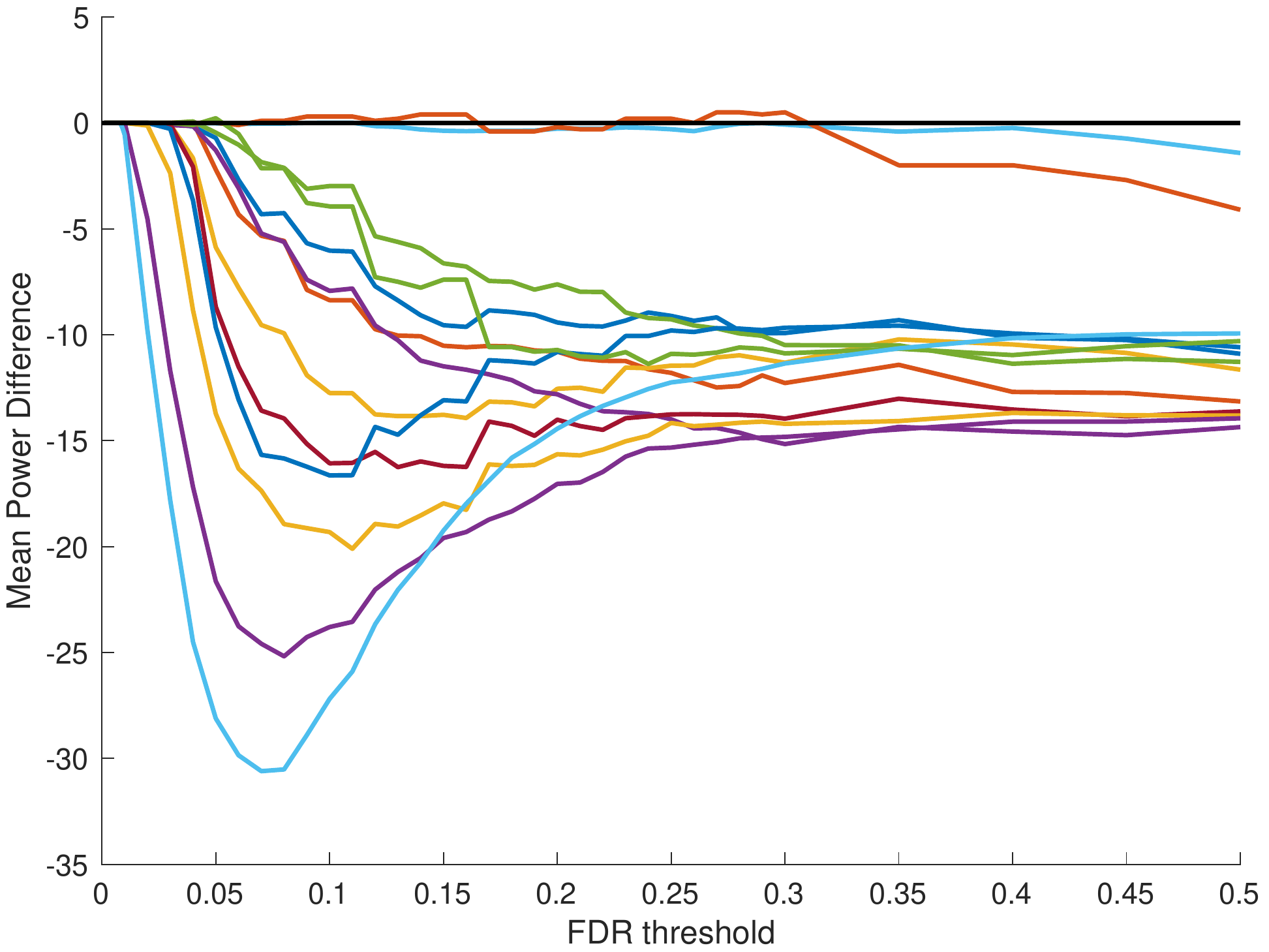} & \includegraphics[width=3in]{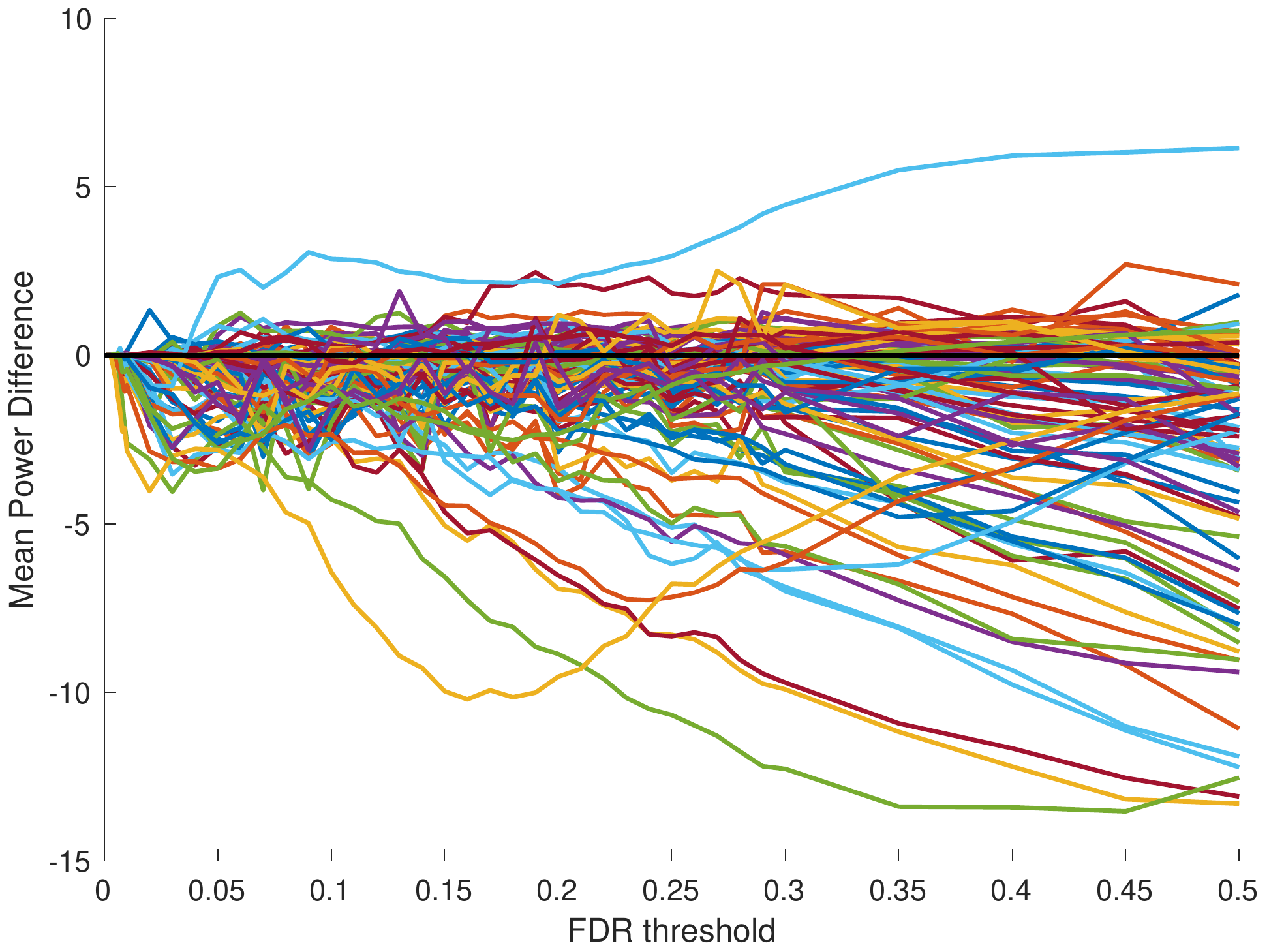}\tabularnewline
		E. Varying $d$ (multi-knockoff) & F. knockoff+ vs.~multi-knockoff-select (combined)\tabularnewline
		\includegraphics[width=3in]{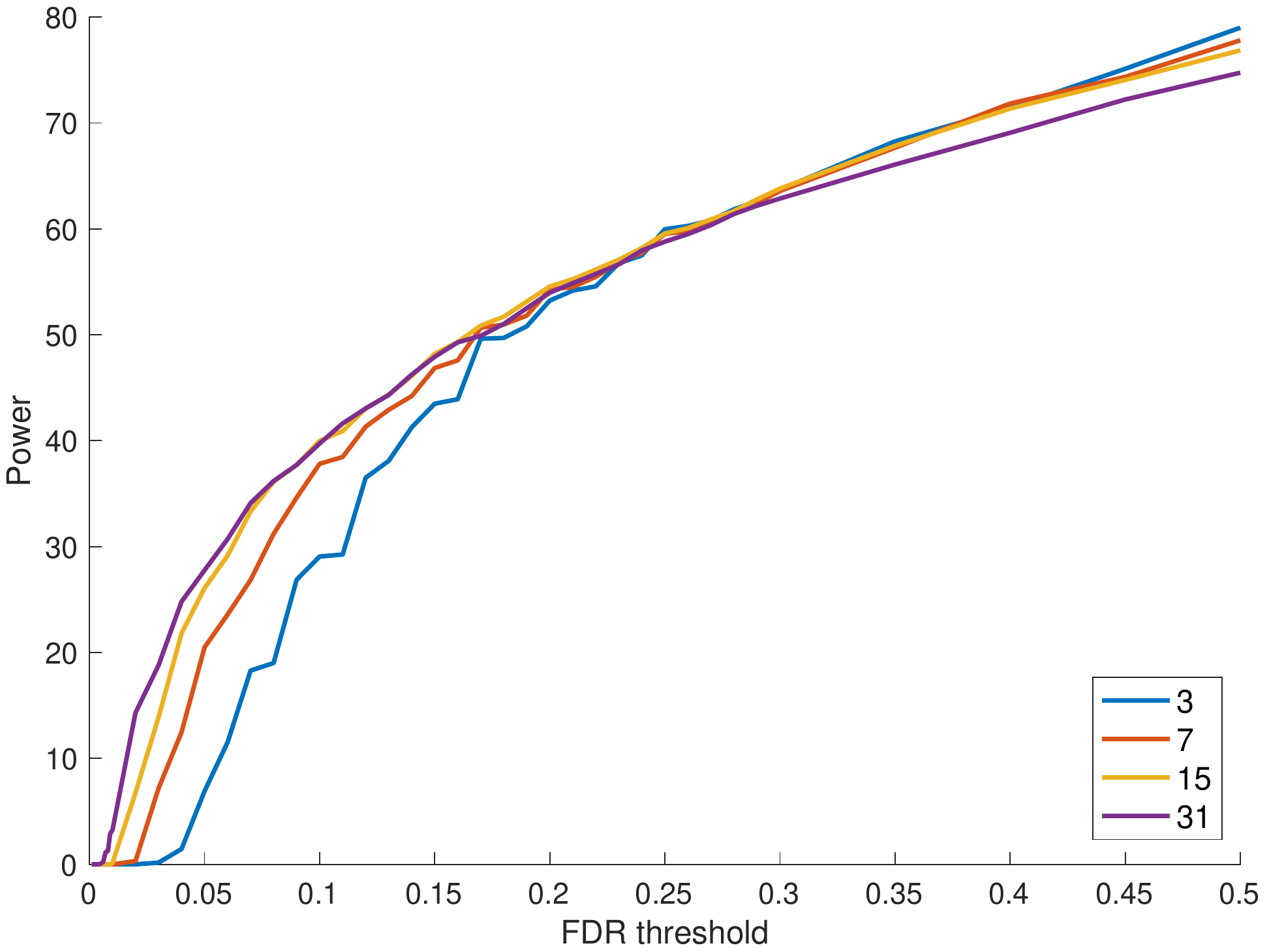} & \includegraphics[width=3in]{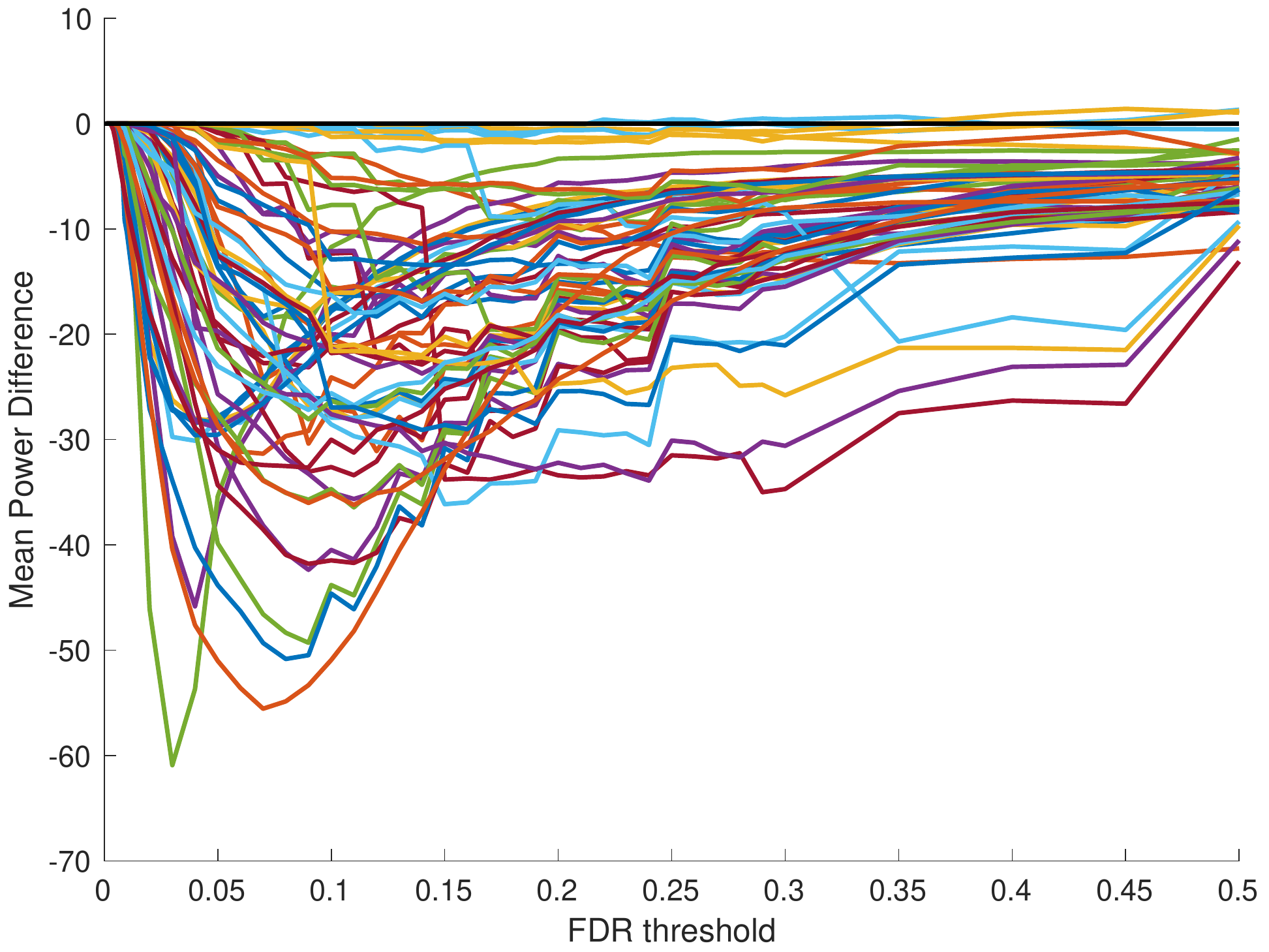}\tabularnewline
	\end{tabular}\caption{\textbf{Main figure.} (A) Power difference between knockoff+ and max
		using $d=2$ ($n=3000,p=1000$, \suppsec\ref{subsec:n3000_p1000_d1_2_b1})
		or $d=3$ ($n=800,p=200$, \suppsec\ref{subsec:n800_p200_d1_3_b1})
		knockoffs. (B) Power of max for varying
		number of batches $b\in\left[1,40\right]$ ($n=600,p=200,d=11,A=2.8$, \suppsec\ref{subsec:n600_p200_d11_varying_b}).
		(C) Power difference between the max method using a single vs.~40
		batches ($n=800,p=200,d=3$, \suppsec\ref{subsec:n800_p200_d1_3_b1})
		(D) Power difference: LBM vs.~multi-knockoff on the combined dataset
		(\suppsec\ref{subsec:Combined-dataset}). Negative values indicate
		multi-knockoff is more powerful. (E) Power of multi-knockoff for varying number of knockoffs $d\in\left\{ 1,3,7,15,31\right\}$
		($n=600$, $p=200$, $K=10$, $A=3.0$, $b=40$ dataset from the set described in \suppsec\ref{subsec:n600_p200_d1_3_7_15_31_b40}).
		(F) Power difference between knockoff+ and multi-knockoff-select on
		the combined set (\suppsec\ref{subsec:Combined-dataset}). \label{fig:main-panel}}
\end{figure}

\subsubsection{Empirically choosing the tuning parameters}

We first compare the performance of LBM, our general multiple-competition
selection procedure, with that of multi-knockoff which is designed
for this linear regression context. Specifically, we apply both methods
to all the datasets in our combined collection of experiments, which
spans the wide range of parameter values described in \suppsec~\ref{subsec:Combined-dataset}.
Panel D of Figure \ref{fig:main-panel} shows that the model-aware
multi-knockoff generally offers more power than the general-purpose
LBM does.\footnote{The one example where LBM is moderately better than multi-knockoff
(cyan colored) corresponds to a realistically borderline 80\% proportion
of features in the model: $K=160$ and $p=200$ ($n=600$, $d=11$).} More specifically, comparing panels A and B of \suppfig\ref{fig:mKO-vs-LBM}
we find that the advantage of multi-knockoff becomes evident when
the number of knockoffs is larger: for $d=3$ (panel B) we do not
see much of a difference, which is expected given that in this case
multi-knockoff considers only three possible combination of values
for ($c,\lam$).

When we rely on data-driven methods to set the values of $c$ and
$\lam$ we lose the theoretical guarantee of FDR control regardless
of whether or not we use batching and/or extension. Resorting to simulation
studies we find that in the same extensive set of experiments both
LBM and multi-knockoff seem to essentially control the FDR (panels
C and D of \suppfig\ref{fig:mKO-vs-LBM}), so the advantage of
multi-knockoff does not seem to come at the expense of controlling
the FDR.

With multi-knockoff's optimization of the tuning parameters ($c,\lam$)
apparently being better than LBM's we went ahead and also compared
the former's power against all the other methods we consider here.
Panels A-D of \suppfig\ref{fig:all-vs-multiKO} show that in each
case multi-knockoff is overall a better option: more often than not
it delivers more power than each of the other methods, and moreover,
when it is not optimal it is giving up only a small amount of power
(certainly for the more practical FDR thresholds of $\alp\le0.3$),
while often enjoying a substantial advantage in power when it is optimal.

\subsubsection{Choosing the optimal number of knockoffs}

So far we examined the performance of the methods when the number
of knockoff copies $d$ is given. However, it is not clear how to
choose an optimal value of $d$ as the setup here is quite different
to the iid decoys model that Emery et al.~looked at. In the latter
case, the larger $d$ is the more power the multiple-decoy procedure
will generally deliver, however in our linear regression context there
is a delicate balance between the increased power due to the increasing
number of competing knockoffs and the reduction in power due to increased
correlation between the knockoffs and the original features. Panels
E of Figure \ref{fig:main-panel} and A and B of \suppfig \ref{fig:multiple_nKO}
demonstrate this problem: the optimal number of knockoffs varies with
the method we use, the parameters of the problem, and the FDR threshold.
This was the motivation behind our new multi-knockoff-select that
tries to optimally select $d$ from the choices it is given, so how
well is it doing in practice?

Panels C and D of \suppfig\ref{fig:multiple_nKO} show that in
the case of the experiments described in \suppsec\ref{subsec:n600_p200_d1_3_7_15_31_b40}
multi-knockoff-select seems to consistently select a nearly optimal
$d$: in the studied cases its power for any $\alpha\le0.5$ was at
worst 5\% below the power of multi-knockoff applied with the optimal
$d$ and the power difference was even smaller for $\alp\le0.2$.
At the same time, for each fixed $d$ there are settings where multi-knockoff-select
delivers significantly more power than multi-knockoff. Importantly,
the overall performance of multi-knockoff-select on this set generated
using six different combinations of parameter values (\suppsec\ref{subsec:n600_p200_d1_3_7_15_31_b40})
was uniformly better than \BC's knockoff+ procedure for $\alp\le0.5$
and often by a significant power margin (panel E, \suppfig\ref{fig:multiple_nKO}).
At the same time, panel F of the same figure shows that this increase
in power was not the result of compromised FDR control.

Moving on to our more extensive set of experiments described in \suppsec\ref{subsec:Combined-dataset}
we find that multi-knockoff-select's flexibility of optimizing over
$d$ makes it our overall preferred procedure\footnote{When an experiment only looks at, say $d\in\left\{ 1,11\right\} $,
then multi-knockoff-select essentially decides whether to use multi-knockoff
with $d=11$ or batched-knockoff+.}. Indeed, \suppfig\ref{fig:all-vs-multiKO-select}
shows that compared with any of the other methods we consider here
multi-knockoff-select overall offers more power. In particular, panel
A of \suppfig \ref{fig:all-vs-multiKO-select} (for convenience
it is the same as panel F of Figure~\ref{fig:main-panel}) shows
that multi-knockoff-select essentially uniformly delivers more power
than knockoff+ and often significantly more. At the same time we again
find that this increase in power does not come at the expense of our
ability to control the FDR (panel A of \suppfig\ref{fig:supp_mKOsel_more}).

Finally, it is instructive to take a closer look at the main example
\BC\ considered of $n=3000$, $p=1000$, $K=30$, $A=3.5$, and 0
feature correlation $\Theta_{0}=I_{p}$. If we use $b=50$ batches
to construct $d\in\left\{ 1,3,7\right\} $ knockoffs then even with
only $m_{b}=4$ bootstrap multi-knockoff-select is a very computationally
demanding procedure (about 11 hours per run on a 3.2GHz macMini).
Fortunately, this significant computational effort is rewarded as
we can see when comparing the power of multi-knockoff-select to that
of knockoff+ (panel B of \suppfig\ref{fig:supp_mKOsel_more}),
and again FDR is well under control (panel C of same figure).

\section{Discussion}

When introducing their knockoff+ procedure \BC\ noted that using multiple
knockoff copies could increase the power of their approach. We recently
introduced a general approach to multiple competition-based FDR control
and here we show how the two concepts could be merged. We first generalize
the knockoff construction of \BC\ to generate $d$ knockoff copies
and prove that under certain conditions (no extension of $X$, no
batching and using pre-determined tuning parameters ($c,\lam$)), applying
our competition-based selection method to these multiple knockoffs
rigorously controls the FDR in the finite sample setting.

Our initial knockoff construction is limited both in terms of its
applicability ($n\ge(d+1)p$) and its utility (panel A of Figure~\ref{fig:main-panel}).
To address these issues we combine \BC' extension notion with our proposed batching heuristic, and
we empirically show that these revised knockoffs still allow
us to effectively control the FDR in the variable selection problem while delivering a substantial increase in power.

Our recommended general procedure for controlling the FDR using multiple
competition is constrained by the generic resampling technique it
uses. Here we show that using a resampling scheme that is specifically
adjusted to the linear regression context allows us to offer more
powerful selection methods. Indeed, our multi-knockoff-select procedure
is largely successful not only in setting a near-optimal value of
($c,\lam$), the tuning parameters of our general FDR controlling
procedure, but also in selecting an optimal number of knockoff copies.
The latter is a non-trivial optimization problem due to the inherent
conflict between the advantage that increasing $d$ offers in terms
of the competition and the reduced power for each knockoff ($s_{0}$
is decreasing as $d$ is increasing).

While there are alternative procedures for controlling the FDR in
the associated variable selection problem (e.g. \cite{miller:subset,miller:selection,gsell:sequential,meinshausen:stab_selection,liu:stability}),
\BC\ note that those, and for that matter their own knockoff (in
contrast with their knockoff+) procedure, generally only asymptotically
guarantee FDR control. They further demonstrate that among the procedures
that control the FDR in the finite setting of the variable selection
problem, their knockoff+ seems to be the most powerful one. Multi-knockoff-select
is more powerful than knockoff+, allowing us to identify more truly
associated features, while empirically we see that it maintains control
of the rate of falsely discovered features even in the finite setting.
It does however come at a substantial computational cost as well as
of using a mathematically unproved technique.

We concentrated on comparing our multiple-knockoff methods with knockoff+
because they naturally generalize that method but it is also instructive
to consider the more recent model-X knockoffs \cite{candes:panning}.
The model-X knockoffs are designed for a different variant of the
linear regression problem where the design matrix itself is also drawn
according to some known distribution. This assumption is still consistent
with the setup of our simulations so we compared the performance of
the model-X knockoffs with the other knockoff procedures in a couple
of examples (\suppsec\ref{subsec:The-model-X-dataset}). \suppfig\ref{fig:supp_modelX}
suggests that in the context of our simulations the model-X lasso
signed max (LSM) statistic was roughly on-par or slightly weaker than
the original knockoff+, and the model-X lasso coefficient difference
(LCD) statistic significantly lagged behind those two. In particular,
unless the FDR threshold was relatively high and the feature correlation
extremely high ($\rho=0.9$), all these single-knockoff methods offered
significantly less power than multi-knockoff-select.

There are a few directions we would like to explore following this
work. First, like the original knockoff+ our method is limited to
the case $n>p$. Thus, it would be particularly interesting to see
whether our approach can be combined with \BC' recent extension to
the $n\le p$ case based on their data-splitting technique coupled
with their introduction of a two step procedure: acquiring a partial
model with $n>p$ and performing the knockoff procedure on the partial
model \cite{barber:knockoff}. Second, in this work we focused on
generalizing the original knockoff construction to multiple knockoffs
so an obvious question is how much of this carries over to the model-X
knockoffs. Third, our current estimate of the noise level $\sigma$
is rather naive and, as we saw, using it the knockoff scores gradually
drift from the assumed model when using large extensions (Section
\ref{subsec:Too-many-batches}). It would therefore be interesting
to explore more sophisticated estimations of $\sig$ such as the one
in \cite{reid:study}.

\bibliographystyle{plain}
\bibliography{../multiple_refs,/Users/keich/bio_papers/noble-lab-references/refs,/Users/keich/References/journal_papers/stats/stat_papers}

\clearpage

\section{Supplementary Material}

\subsection{Simulation setup\label{subsec:Simulation-setup-details}}

Our general simulation setup is described in Section \ref{subsec:Simulation-setup}.
In the following sections we give further details about the parameter
settings we used in our experiments in generating the original design
matrices (and the response variables) and the knockoff features as
well as any optional settings of our selection methods. When generating
the data we varied the dimension of the design matrix $X$, $n\times p$,
the number of true features, $K$, the signal amplitude $A$, and
the feature correlation strength $\rho$ while keeping the variance
of the noise fixed at $\sig^{2}=1$ (cf. \eqref{eq:lin_model} and
Section \eqref{subsec:Simulation-setup}). We generally randomly sampled
1K sets of data for each of the parameter combinations we considered
and constructed a set of original plus $d$ knockoff scores for each
feature using the specified number of $b$ batches.

While knockoff+ and batched-knockoff+ were each applied only once
to the data --- each with its corresponding knockoff --- the multiple
knockoff procedures were applied separately for each considered value
of $d$ (with the knockoffs also separately constructed for each value
of $d$, cf. Section \ref{subsec:How-many-knockoffs}).

Following knockoff+ we also use the glmnet implementation of the Lasso
\cite{qian:glmnet}. We found that the set of values the regularization
parameter lambda is allowed to assume can have a non-negligible effect
on our analysis. This is not surprising given that the original and
the knockoff feature score corresponds to the largest value of lambda
for which the coefficient of that feature is non-zero. Therefore,
to make sure that the differences we observe between the methods are
not due to variations in the number of lambdas, we set each method
to use the same number of possible lambda values. Specifically, this
number was set to $5\cdot\left(1+d_{\max}\right)\cdot p$ (we experimented
a little with coefficients other than 5 but kept it at 5 throughout
the simulations described here), where $d_{\max}$ is the maximal
value of $d$ considered in that experiment. Importantly, it means
that knockoff+ also used that maximal number of lambdas even though
it is using a single knockoff.

In the same vein, we found that our knockoffs are better behaved if
all our batches use the same set of lambdas. Specifically, we use
the same exponentially spaced set of lambda values that knockoff+
uses only we set the maximal value to $\max_{i}\xx_{i}^{T}\yy/n$
where $\xx_{i}$ varies over the columns of all augmented design matrices
$\left[X\tilde{X}^{I_{j}}\right]$ (Section \ref{sec:Batched-KOs}).
Note that when using a single batch this maximal value coincides with
the one originally used in the knockoffs package.

Somewhat more surprising was the fact that permuting the columns of
the extended design matrix before applying glmnet also occasionally
had a substantial effect on the performance of the FDR controlling
methods, \emph{including} on knockoff+. Therefore, we uniformly randomly
permuted all extended design matrices.

Note that we applied LBM exactly as described in \suppsec 6.6 of \cite{emery:multiple2}
and that when multiple number of knockoffs are considered in the same run the data
is only resampled once using the largest considered number of knockoffs to create the model-aware resamples.

\subsubsection{$n=800$, $p=200$, $d\in\left\{ 1,3\right\} $, $b=1$\label{subsec:n800_p200_d1_3_b1}}

A set of experiments designed to study the performance of FDR control
with multiple knockoffs in the setting of guaranteed finite sample
control. Each drawn dataset was generated starting with randomly drawing
the $n\times p$ design matrix $X$ where $n=800$ and $p=200$ while
varying the following parameters (drawing 1K independent datasets
per each setting of the parameters):
\begin{itemize}
	\item the number of true features $K$ (model sparsity): $K=1,5,10,20,40,80$
	with $A=3.0$ and 0 feature correlation $\Theta_{0}=I_{p}$.%
	\item the signal amplitude $A=2.6,2.8,3.0,3.2,3.4$ with $K=10$ and 0 feature
	correlation $\Theta_{0}=I_{p}$.
	\item the feature correlation strength $\rho=0,0.3,0.6,0.9$ of the T\"oeplitz
	correlation matrix $\Theta_{\rho}$ with $K=10$, $A=3.0$.
\end{itemize}
When analyzing each of these datasets we constructed three sets of
knockoffs (cf. Section \ref{subsec:How-many-knockoffs}): the first
with a single knockoff per feature as part of running knockoff+',
the second using our procedure to construct $d=1$ knockoff (for use
in batched-knockoff+), and the third using our construction of $d=3$
knockoffs. Note that in this case we used a single batch so while
our construction of $d=1$ knockoff is in practice different from
knockoff+' it is mathematically equivalent to it. In particular, knockoff+
and batched-knockoff+ are essentially equivalent. %

\subsubsection{$n=3000$, $p=1000$, $d\in\left\{ 1,2\right\} $, $b=1$\label{subsec:n3000_p1000_d1_2_b1}}

Similar in design and intent to the last section except that we used
the same dimensions of the data as used in the introduction of knockoff+
\cite{barber:controlling} and varied:
\begin{itemize}
	\item the number of true features $K$ (model sparsity): $K=0,30,50$ with
	$A=3.5$ and 0 feature correlation $\Theta_{0}=I_{p}$.
	\item the signal amplitude $A=3.1,3.5,3.9$ with $K=30$ and 0 feature correlation
	$\Theta_{0}=I_{p}$.
	\item the feature correlation strength $\rho=0,0.3,0.6,0.9$ with $K=30$,
	$A=3.5$.
\end{itemize}
For each of these combinations of parameters we drew 1K datasets and
constructed two sets of knockoffs, each using a single batch, one
with $d=1$ (again equivalent to knockoff+' single knockoff) and another
with $d=2$ knockoffs per feature.

\subsubsection{$n=800$, $p=200$, $d\in\left\{ 1,3\right\} $, $b=40$\label{subsec:n800_p200_d1_3_b40}}

In this set of experiments we largely repeated the setup described
in Supplementary Section \eqref{subsec:n800_p200_d1_3_b1} above with
the major difference that we used 40 batches when constructing either
$d=1$ or $d=3$ knockoffs per feature rather than a single batch.
In particular, here batched-knockoff+ substantially differs from knockoff+.

Also, in addition to generating 1K datasets for each of the parameter
combinations described above we generated additional
1K datasets per parameter combination while varying:
\begin{itemize}
	\item the number of true features $K=1,5,10,20,40,80$ with $A=2.8$ and
	0 feature correlation $\Theta_{0}=I_{p}$.
	\item the feature correlation strength $\rho=0,0.3,0.6,0.9$ with $K=10$,
	$A=2.8$.
\end{itemize}
Applications of multi-knockoff and multi-knockoff-select used $m_{b}=32$
model-aware bootstrap samples.

\subsubsection{$n=600$, $p=200$, $d\in\left\{ 1,11\right\} $, $b\in\left\{ 1,5,10,20,40\right\} $\label{subsec:n600_p200_d11_varying_b}}

A set of experiments designed to demonstrate the effects that increasing
number of batches can have. The datasets were generated using a fixed
$600\times200$ dimension for the design matrix, with $A=2.8$ and
0 feature correlation $\Theta_{0}=I_{p}$. For each $b\in\left\{ 1,5,10,20,40\right\} $
we randomly generated 1K datasets, then using $b$ batches each time
we constructed two sets of knockoffs one with $d=1$ and another
with $d=11$ knockoffs per feature. Applications of multi-knockoff
and multi-knockoff-select used $m_{b}=32$ model-aware bootstrap samples.

\subsubsection{$n=3000$, $p=1000$, $d\in\left\{ 1,3\right\} $, $b\in\left\{ 1,5,10,30,50\right\} $\label{subsec:n3000_p1000_d3_varying_b}}

Similar in design and intent to the last section except that (a) we
used different parameter values, and (b) when creating these sets
of knockoffs we did not use clustering when defining the batches,
instead we arbitrarily divided the features into equal sized batches
(which of course is irrelevant for $b=1$).

Applications of multi-knockoff and multi-knockoff-select used $m_{b}=4$
model-aware bootstrap samples.

\subsubsection{``Eclectic batching example''\label{subsec:Eclectic-bacthing-example}}

This set gathered together some specific experiments for demonstrating
the effects of batching on larger examples. It consisted of three
pairs of experiments:
\begin{itemize}
	\item $n=800,p=200,d\in\left\{ 1,3\right\} ,b\in\left\{ 1,40\right\}$
	taken from \suppsecs \ref{subsec:n800_p200_d1_3_b1} and \ref{subsec:n800_p200_d1_3_b40}.
	\item $n=600,p=200,d\in\left\{ 1,11\right\} ,b\in\left\{ 1,40\right\}$ taken from \suppsec \ref{subsec:n600_p200_d11_varying_b}.
	\item $n=3000,p=1000,d\in\left\{ 1,3\right\} ,b\in\left\{ 1,30\right\}$ taken from \suppsecs \ref{subsec:n3000_p1000_d1_2_b1} and \ref{subsec:n3000_p1000_d3_b30}.
\end{itemize}

\subsubsection{$n=800$, $p=200$, $d\in\left\{ 1,3,7\right\} $, $b=40$\label{subsec:n800_p200_d1_3_7_b40}}

In this set of experiments we used the same parameter combinations
for generating the data as described in Supplementary Section \ref{subsec:n800_p200_d1_3_b1}
above. The difference again was in the analysis stage where here we
constructed for each drawn design matrix $X$ three sets of knockoffs
(in addition to the set generated by knockoff+). Each set was constructed
using 40 batches: one with $d=1$, another with $d=3$ and the third with $d=7$ knockoffs per feature.
We then applied knockoff+ and
batched-knockoff+ using their corresponding single knockoff set, and
we applied each of the multiple-knockoff procedures twice, once to
the $d=3$ set and once to the $d=7$ set. Applications of multi-knockoff
and multi-knockoff-select used $m_{b}=32$ model-aware bootstrap samples.

\subsubsection{$n=600$, $p=200$, $d\in\left\{ 1,11\right\} $, $b=40$\label{subsec:n600_p200_d1_11_b40}}

In this set of experiments we used a larger number of knockoffs with
$600\times20$0 design matrices. Again, we generated 1K datasets (design
matrix and response variables) for each of the following combinations
of parameters, varying:
\begin{itemize}
	\item the number of true features $K=0,1,5,10,20,40,80,160$ with $A=2.8$
	and 0 feature correlation $\Theta_{0}=I_{p}$.
	\item the signal amplitude $A=2.4,2.6,2.8,3.0,3.2$ with $K=10$ and 0 feature
	correlation $\Theta_{0}=I_{p}$.
	\item the feature correlation strength $\rho=0,0.3,0.5,0.7,0.9$ of the
	T\"oeplitz correlation matrix $\Theta_{\rho}$, as well as using
	$\rho=0.5$ for $\Omg_{\rho}$ (constant $\rho$ on the off-diagonal
	terms) with $K=10$, $A=2.8$.
\end{itemize}
Note that for two of the experiments we increased the number of runs
to 4K from the initial 1K by adding another 3K independent runs to
clarify whether the relatively high empirical FDR that was observed
in a couple of the settings was significant. In both cases $(K=5$
and $K=0$) the aggregated 4K runs did not show a substantial FDR
violation. For each drawn dataset we used $b=40$ batches to construct
two sets of knockoffs, one with a single knockoff per feature, and
another with $d=11$.

\subsubsection{$n=600$, $p=200$, $d\in\left\{ 1,3,7,15,31\right\} $, $b=40$\label{subsec:n600_p200_d1_3_7_15_31_b40}}

This set of experiments was specifically designed to compare the performance
using a varying number of knockoffs while analyzing the same data,
as well as to test the ability of multi-knockoff-select to select an optimal
number of knockoff. Each of the following combination of parameters
was used to generate 1K datasets and for each we constructed a knockoff
set for each value of $d\in\left\{ 1,3,7,15,31\right\} $ using $b=40$
batches per construction (in addition to the knockoffs generated by
knockoff+). To generate the data we varied:
\begin{itemize}
	\item the number of true features $K$ (model sparsity): $K=1,10,40$ with
	$A=3.0$ and 0 feature correlation $\Theta_{0}=I_{p}$.
	\item the signal amplitude $A=2.6,3.0,3.4$ with $K=10$ and 0 feature correlation
	$\Theta_{0}=I_{p}$.
	\item the feature correlation strength $\rho=0,0.5$ of $\Theta_{\rho}$
	with $K=10$, $A=3.0$.
\end{itemize}
Applications of multi-knockoff and multi-knockoff-select used $m_{b}=32$
model-aware bootstrap samples.

\subsubsection{$n=3000$, $p=1000$, $d\in\left\{ 1,3\right\} $, $b=30$\label{subsec:n3000_p1000_d3_b30}}

The data for this set of experiments was generated using the same
general parameter settings as those used in the simulation part of
\cite{barber:controlling}. Specifically, we drew $3000\times1000$
design matrices and response variables by varying:
\begin{itemize}
	\item the number of true features $K$ (model sparsity): $K=0,5,10,30,50,75,100,200$
	with $A=3.5$ and 0 feature correlation $\Theta_{0}=I_{p}$.
	\item the signal amplitude $A=2.7,3.1,3.5,3.9,4.3$ with $K=30$ and 0 feature
	correlation $\Theta_{0}=I_{p}$.
	\item the feature correlation strength $\rho=0,0.3,0.5,0.7,0.9$ of the
	T\"oeplitz correlation matrix $\Theta_{\rho}$, as well as using
	$\rho=0.5$ for $\Omg_{\rho}$ (constant $\rho$ on the off-diagonal
	terms) with $K=30$, $A=3.5$.
\end{itemize}
For each of these combinations of parameters we drew 1K datasets and
constructed two sets of knockoffs, one with $d=1$ and another with
$d=3$ knockoffs per feature, using $b=30$ batches each time.

Applications of multi-knockoff and multi-knockoff-select used $m_{b}=4$
model-aware bootstrap samples.

\subsubsection{Combined dataset\label{subsec:Combined-dataset}}

The ``combined dataset'' was created by merging together the $n=3000,p=1000$
sets (\suppsec\ref{subsec:n3000_p1000_d3_b30}), the $n=800,p=200$
sets (\suppsecs\ref{subsec:n800_p200_d1_3_b40} and \ref{subsec:n800_p200_d1_3_7_b40}
but keeping only one copy of each duplicated $d=3$ set), and the
$n=600,p=200$ sets (\suppsecs\ref{subsec:n600_p200_d1_11_b40}
and \ref{subsec:n600_p200_d1_3_7_15_31_b40}).

\subsubsection{The model-X dataset\label{subsec:The-model-X-dataset}}

We used this data to compare against the newer model-X knockoffs of
\cite{candes:panning}. The design matrix was $n=600$ by $p=200$,
the model had $K=10$ features with $A=2.8$ and we only varied the
feature correlation $\rho=0.0,0.5,0.9$ while creating $d=11$ knockoffs
using $b=40$ batches. Applications of multi-knockoff and multi-knockoff-select
used $m_{b}=32$ model-aware bootstrap samples. The model-X knockoffs were created as follows:
\begin{itemize}
	\item The knockoff features were created using the \texttt{create} function
	of the Matlab knockoffs package \cite{candes:panning} with the model
	defined as \texttt{gaussian} coupled with the mean and covariance
	estimated from the randomly drawn design matrix $X$ using the Matlab
	functions \texttt{mean} and \texttt{cov}. We used the ``equicorrelated''
	construction because our knockoff construction also uses the same
	constant $s_{0}$ construction.
	\item We generated the model-X Lasso signed max (LSM) statistic using a
	slightly modified version of the function \texttt{lassoLambdaSignedMax}
	from the knockoffs package that enabled us to to randomly permute
	the order of the columns of the extended design matrix (see \suppsec\ref{subsec:Simulation-setup-details}).
	We set the function's \texttt{nlambda} parameter to the same value
	that the other methods were using (see \suppsec\ref{subsec:Simulation-setup-details}).
	\item We generated the model-X Lasso coefficient difference (LCD) statistic
	using a similarly modified version of the \texttt{lassoCoefDiff} function
	from the knockoffs package that allowed us to randomly permute the
	extended design matrix columns but other than that we used all the
	default settings of the original function.
\end{itemize}

\subsubsection{The set $\Phi$ of FDR thresholds\label{subsec:The-set-Phi}}

For computational efficiency we evaluated the power and empirical
FDR of each of the considered procedures on a pre-determined set of
possible FDR thresholds. Specifically we used the set of FDR thresholds
$\Phi$: from 0.001 to 0.009 by jumps of 0.001, from 0.01 to 0.29
by jumps of 0.01, and from 0.3 to 0.95 by jumps of 0.05. Our figures
however only extend to an FDR threshold of 0.5 since in practice
FDR thresholds higher than 50\% are typically of little importance.

\subsection{Figures}


\begin{figure}
\centering %
\begin{tabular}{ll}
A. Too many batches ($d=1$) & B. Cluster-defined batches ($d=1$)\tabularnewline
\includegraphics[width=3in]{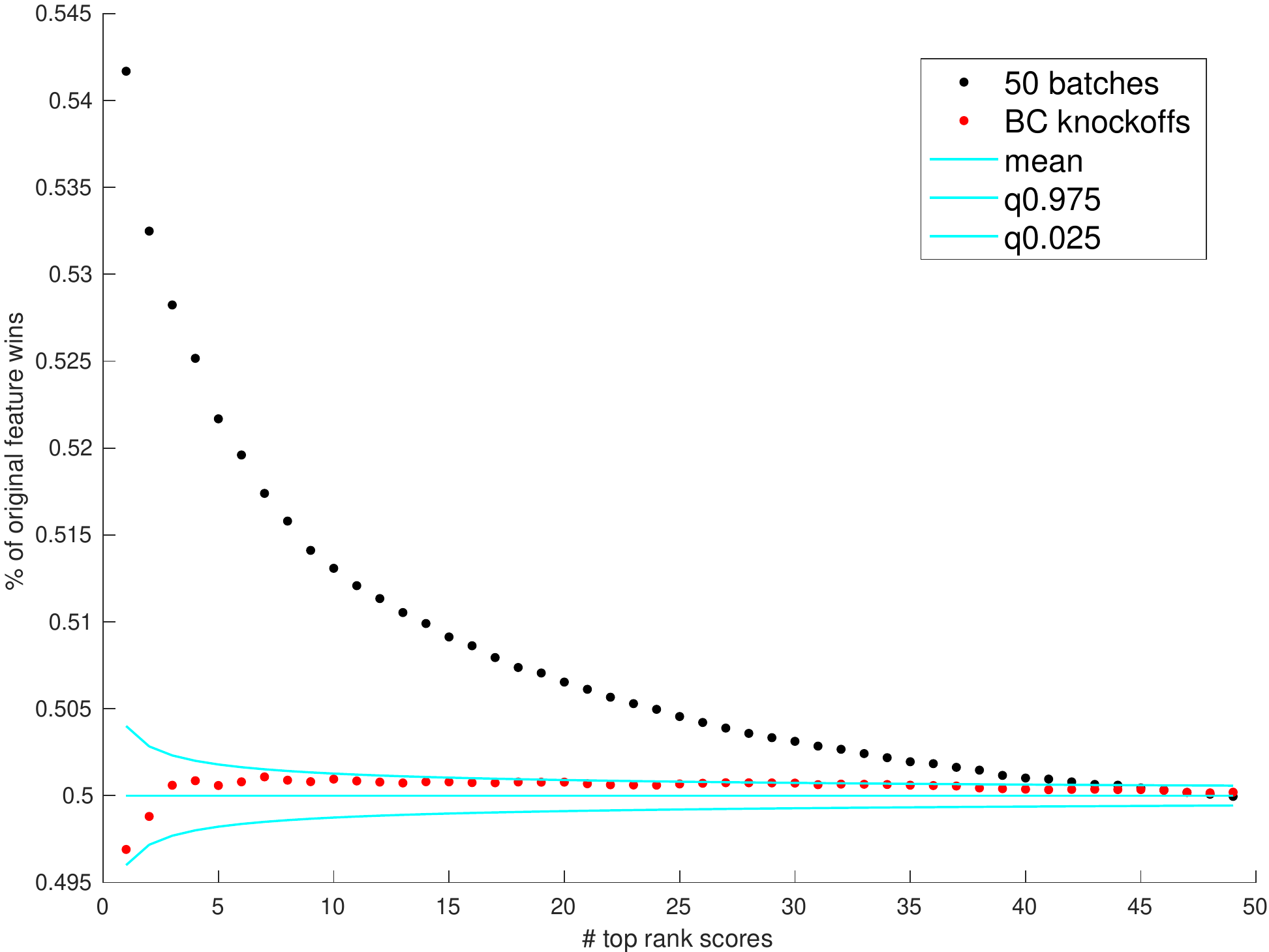} & \includegraphics[width=3in]{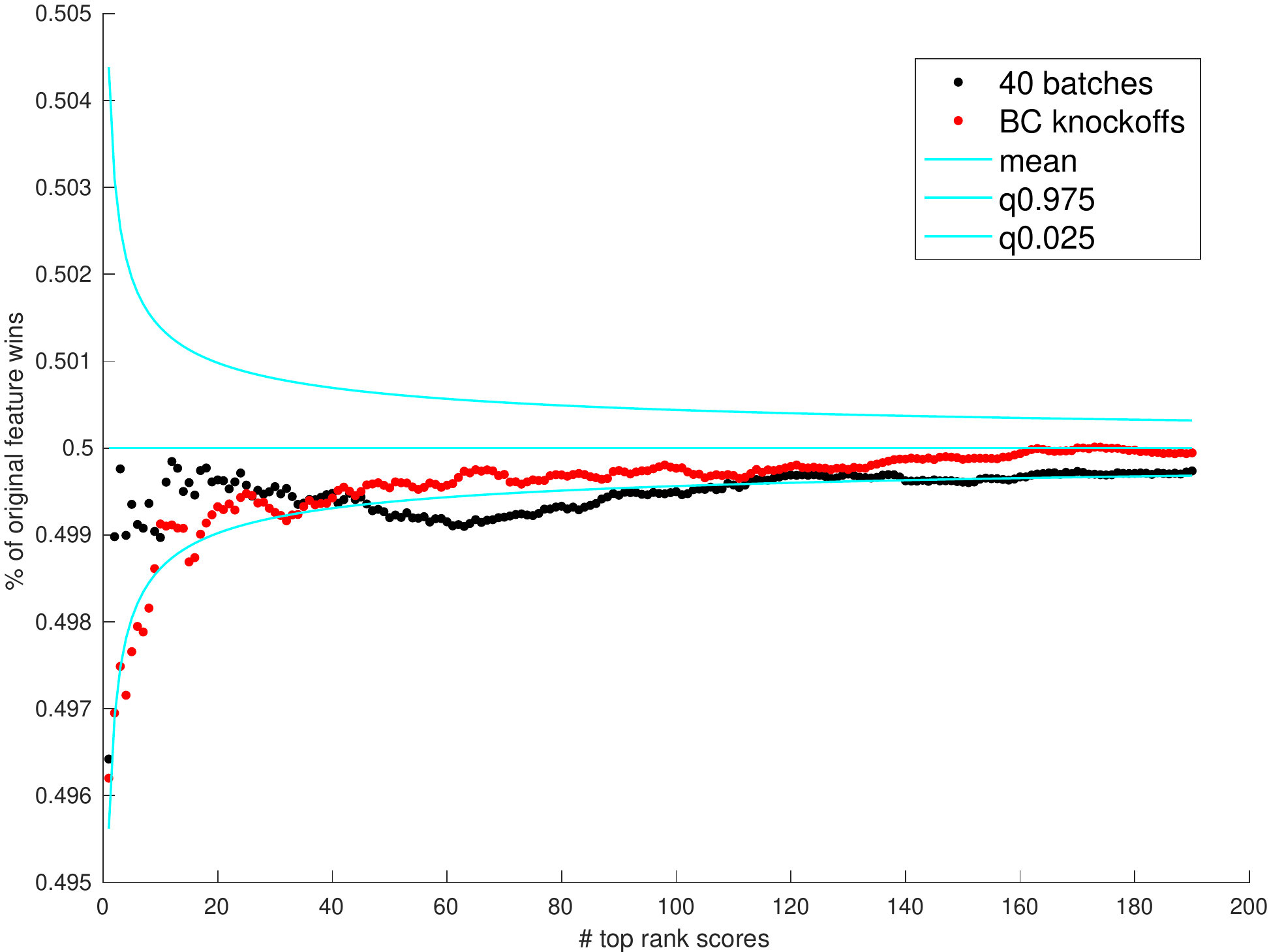}\tabularnewline
C. Uniform random batches ($d=1$) & D. $X$ is not extended ($d=3$)\tabularnewline
\includegraphics[width=3in]{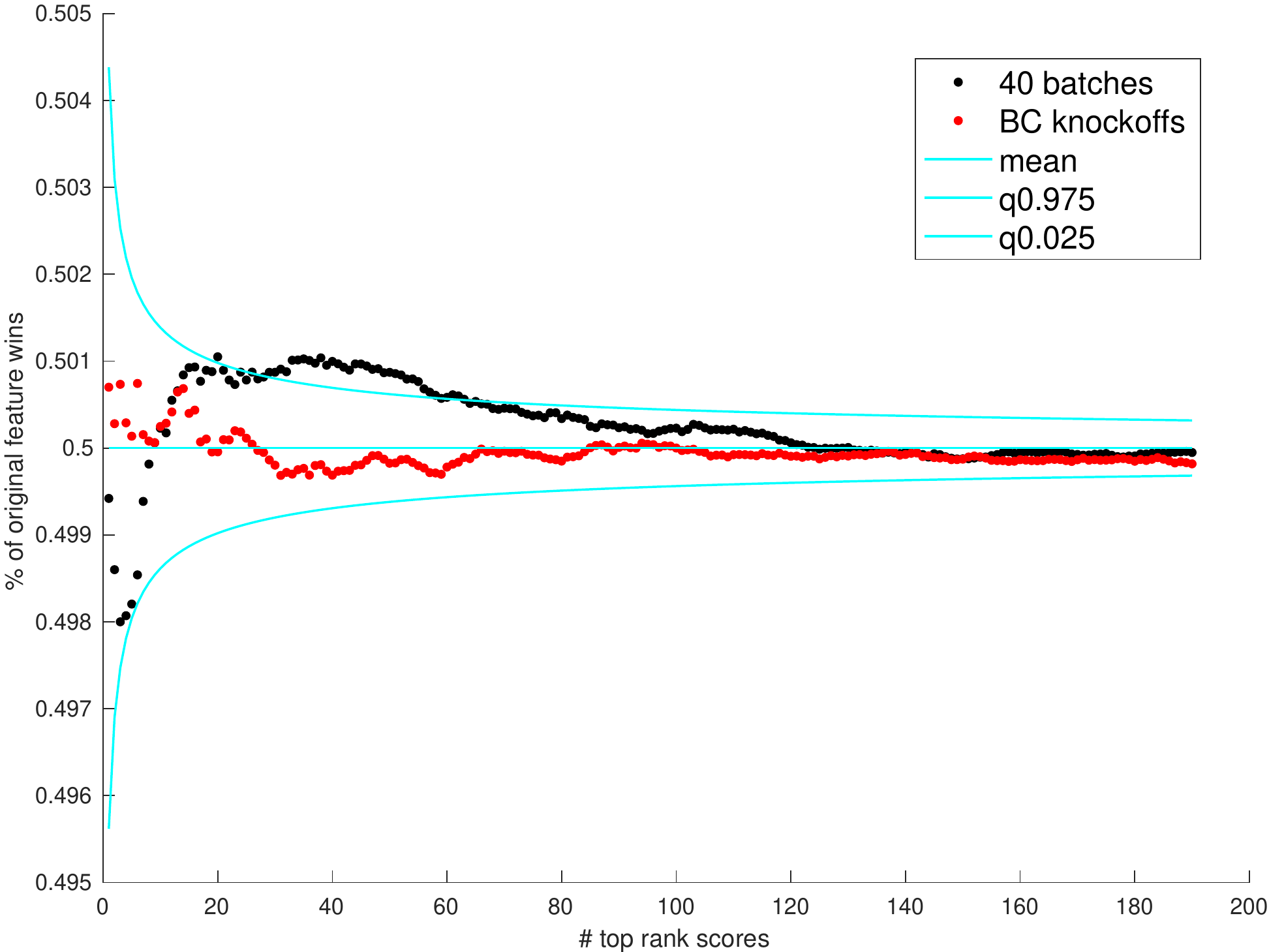} & \includegraphics[width=3in]{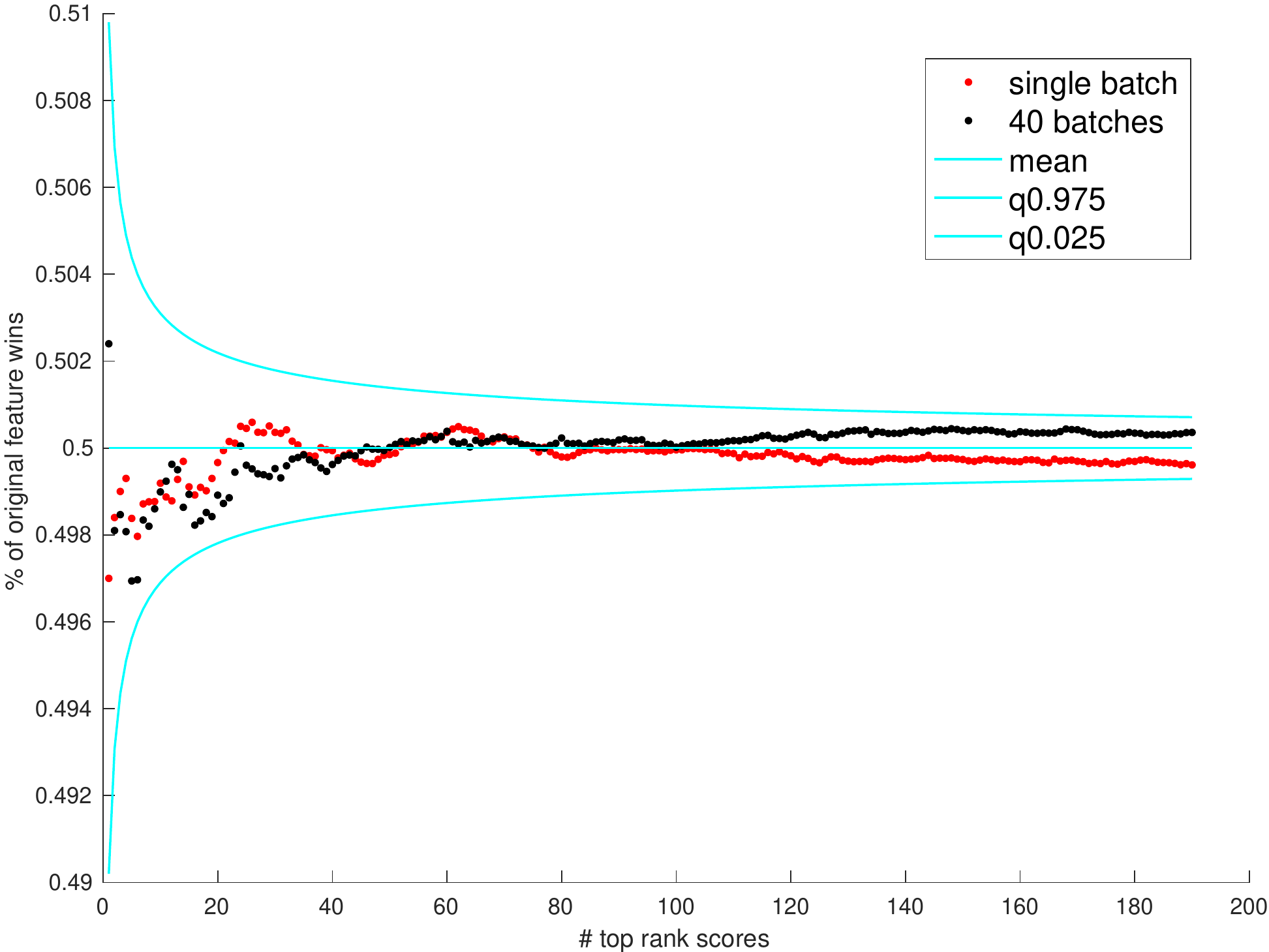}\tabularnewline
E. $X$ is extended ($d=3$) & F. $X$ is significantly extended\tabularnewline
\includegraphics[width=3in]{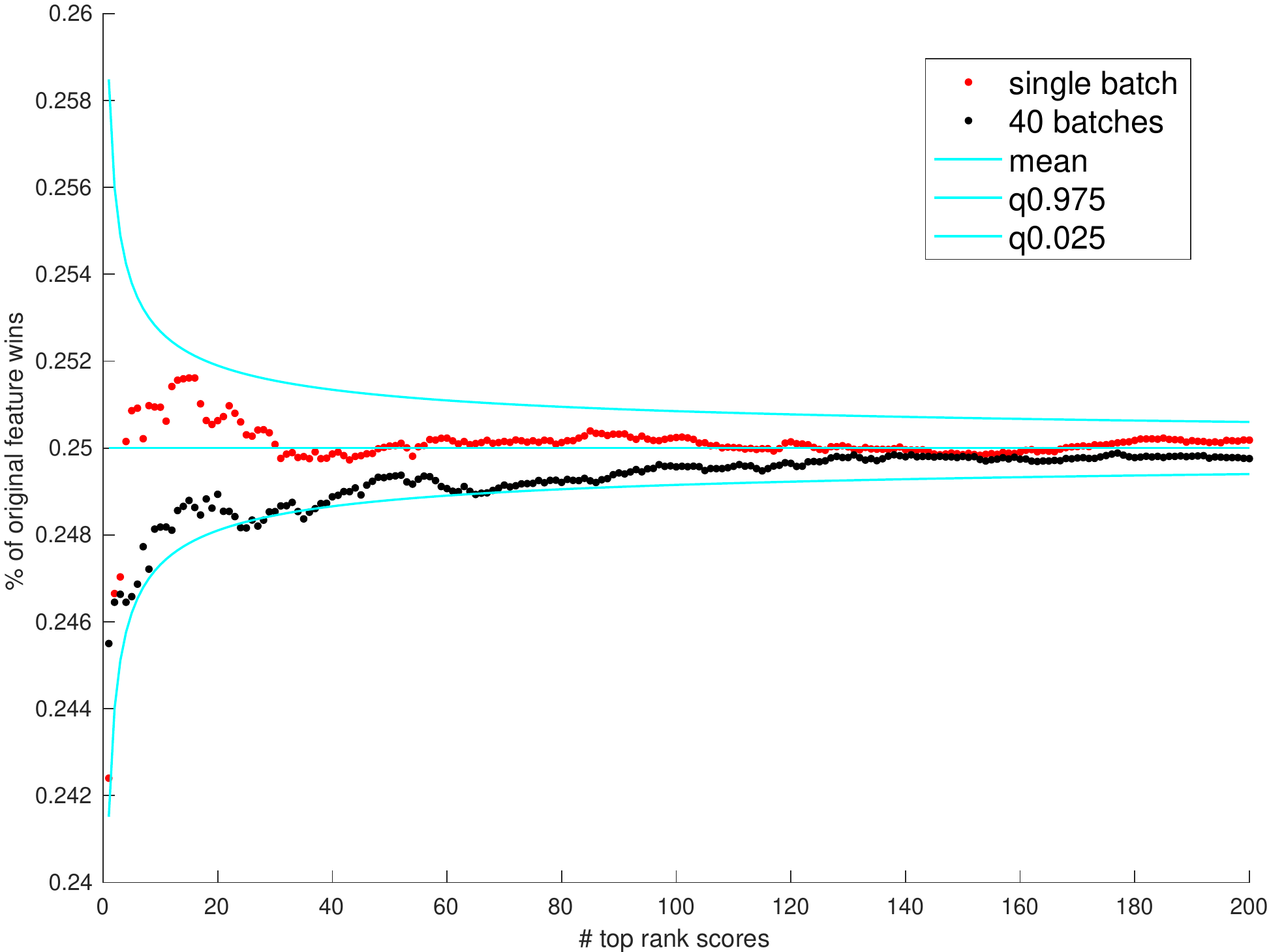} & \includegraphics[width=3in]{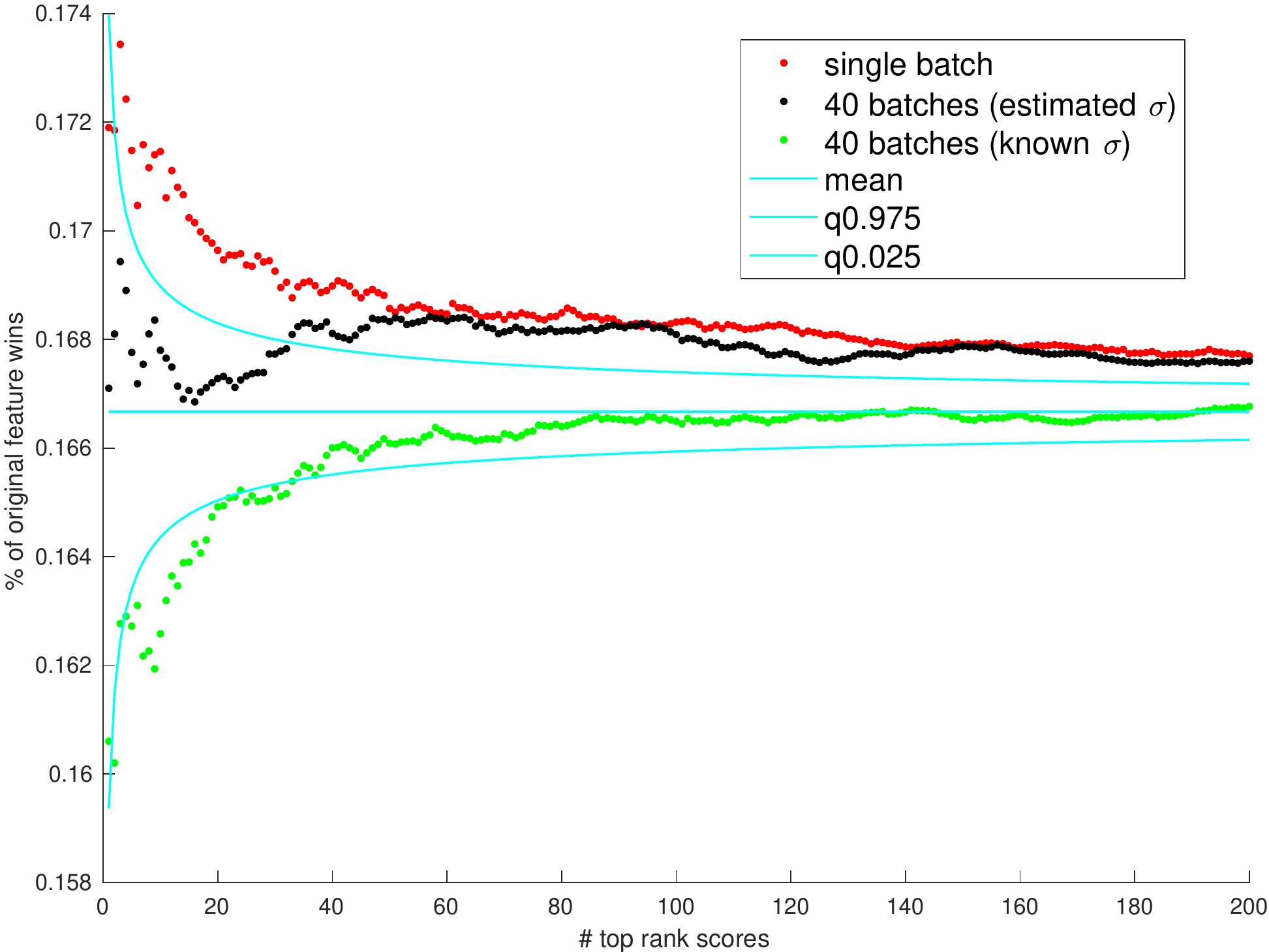}\tabularnewline
\end{tabular}\caption{\textbf{Examining the batched knockoff scores.} Each panel examines
the percentage of original feature wins among the top scoring true
null features. If the knockoffs satisfy the conditional null exchangeability
then those percentages can be modeled by the percentage of cumulative successes
in an iid sequence of Bernoulli($c$) RVs. The cyan colored curves
are the 0.975, 0.025 quantiles and the mean of the iid Bernoulli model.
(A) The batched knockoffs exhibit a clear liberal bias when each batch
consists of a single feature ($d=1$, $c=1/2$, $p=50$, $n=100$,
$\Theta_{\rho}=I_{p}$, $K=1$, $A=10.0$, 60K datasets). %
(B-C) clustered batches (B) seem to follow the model more closely
than the random uniform batches (C), which exhibit some distinct liberal
bias ($d=1$, $c=1/2$, $p=200$, $n=800$, $\protect\Omg_{\rho}$
with $\rho=0.7$, $K=10$, $A=2.8$, 40 batches, 50K datasets). %
(D) Comparing 40 clustered batches with non-batched knockoffs ($d=3$,
$c=2/4$, $p=200$, $n=800$, $\Theta_{\rho}=I_{p}$, $K=10$, $A=2.8$,
two distinct sets of 10K datasets). The design matrix $X$ was not
extended.%
{} (E) Similar to panel E but the design matrix is extended by 200 rows
of 0 ($d=3$, $c=1/4$, $p=200$, $n=600$, $\Theta_{\rho}=I_{p}$,
$K=0$, two distinct sets of 10K datasets).
(F) With $d=11$, $p=200$, $n=600$, we need to significantly extend
$X$ (to $n=2400$) which in turn creates a liberal bias regardless
of whether a single or 40 batches are used. The bias disappears once
we use the known value of $\sig=1$ rather than try to estimate
it ($c=2/12$, $\Theta_{\rho}=I_{p}$, $K=0$, 40 batches, three distinct
sets of 10K datasets).%
\label{fig:tgt-win-pct}}
\end{figure}

%

\begin{figure}
	\centering %
	\begin{tabular}{ll}
		A. $n=3000,p=1000$; varying $A$ (amplitude) & B. $n=800,p=200$; varying $A$\tabularnewline
		\includegraphics[width=3in]{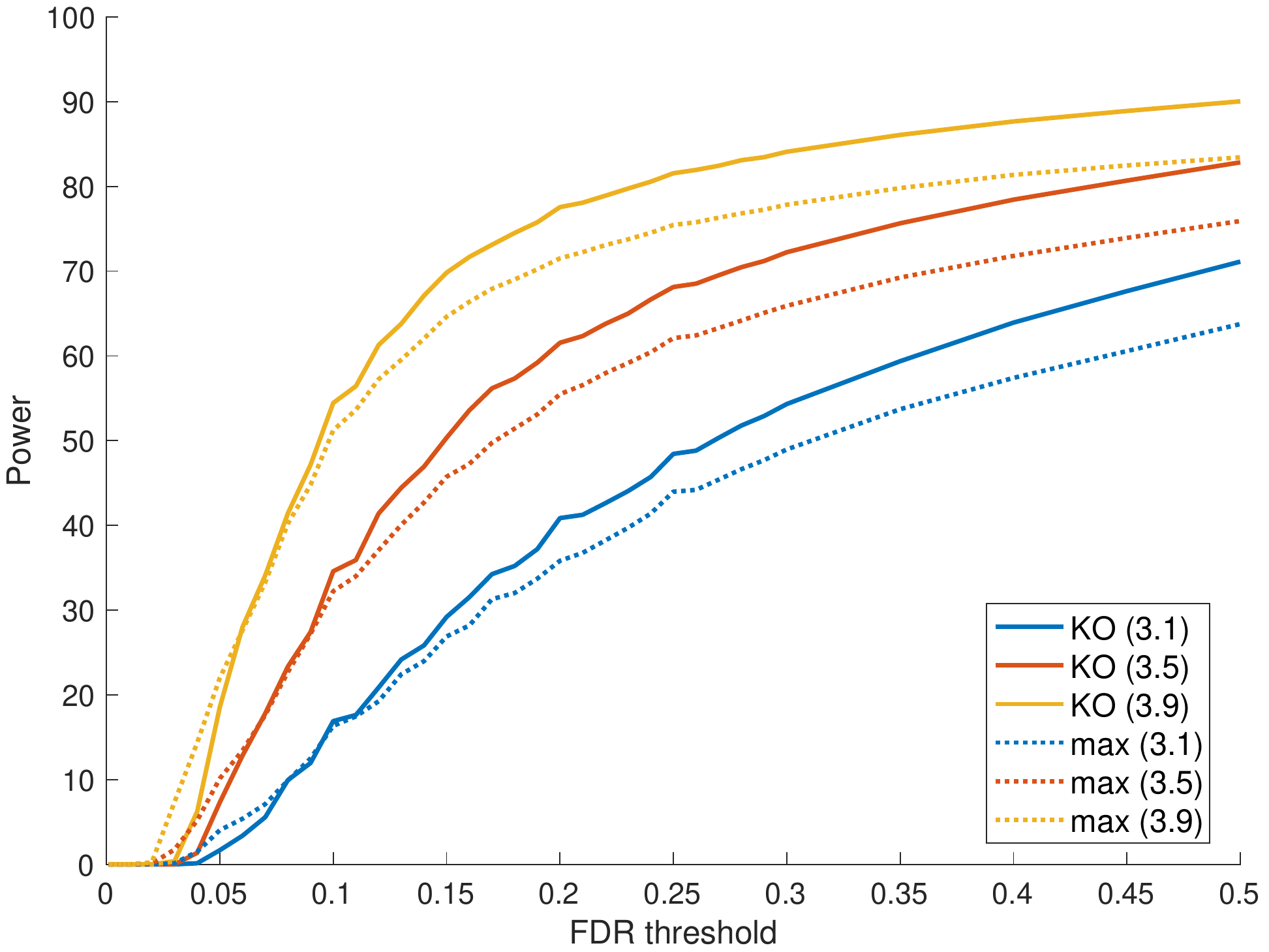} & \includegraphics[width=3in]{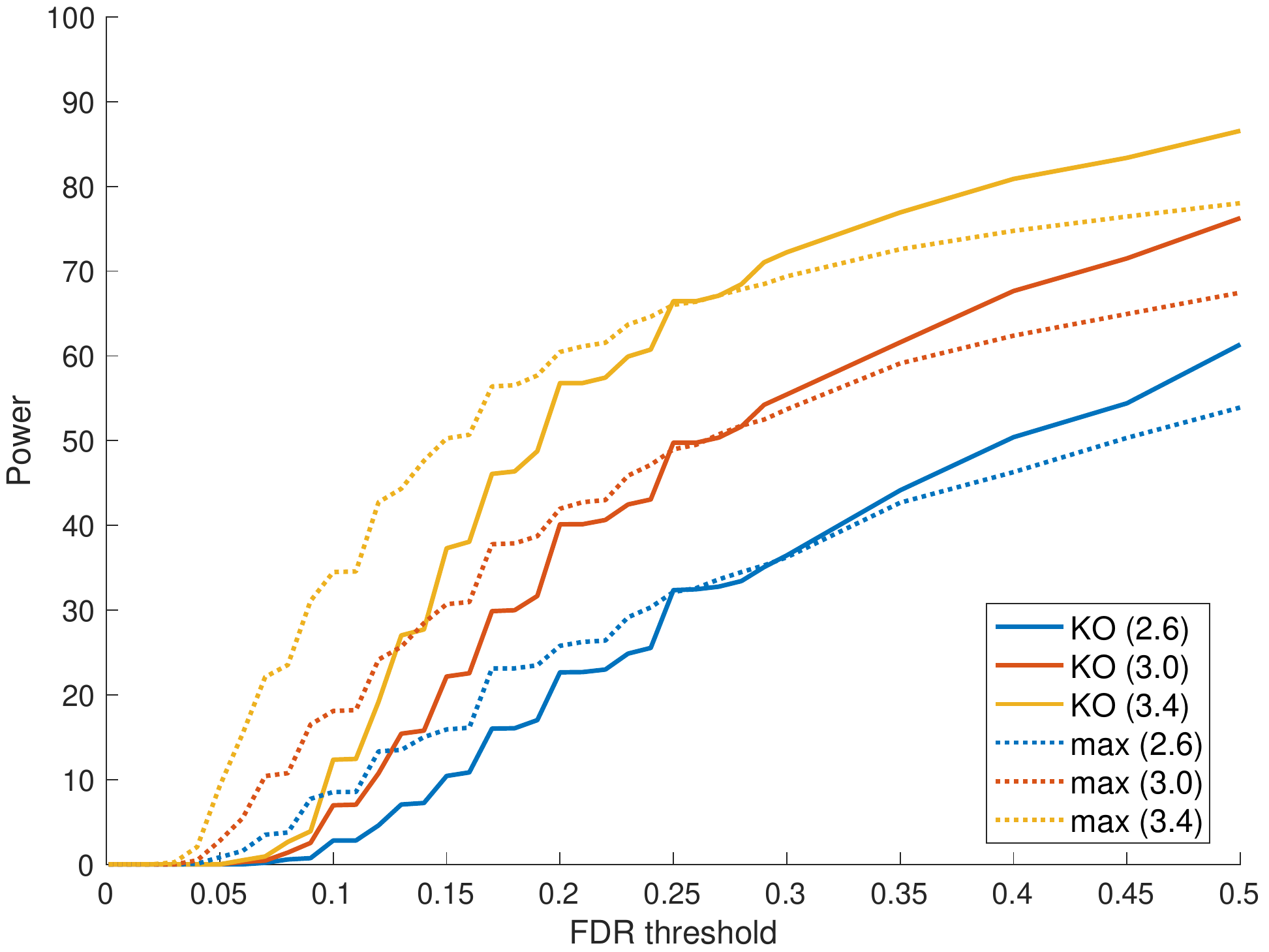}\tabularnewline
		C. varying $\rho$ (feature correlation) & D. varying $\rho$\tabularnewline
		\includegraphics[width=3in]{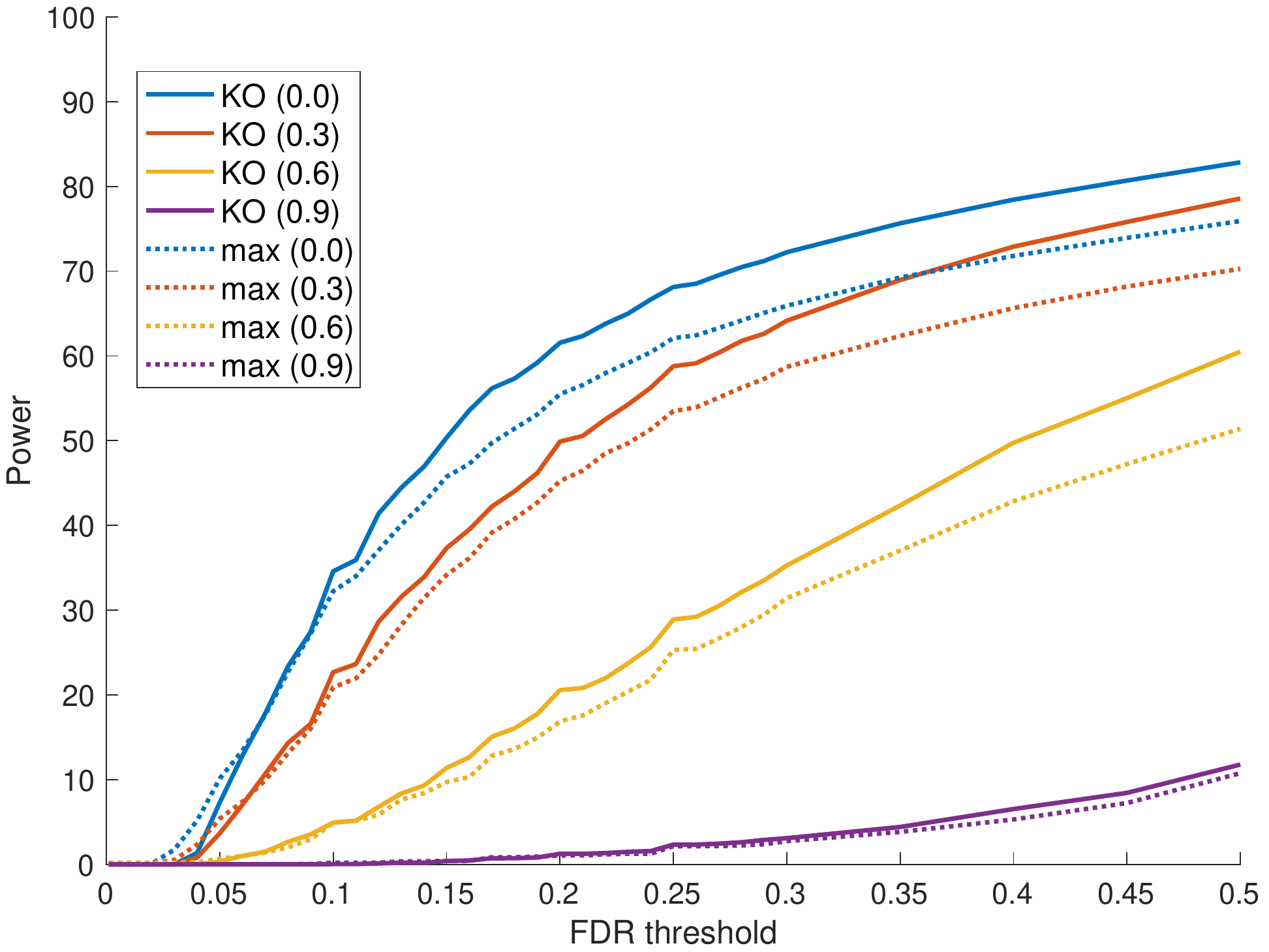} & \includegraphics[width=3in]{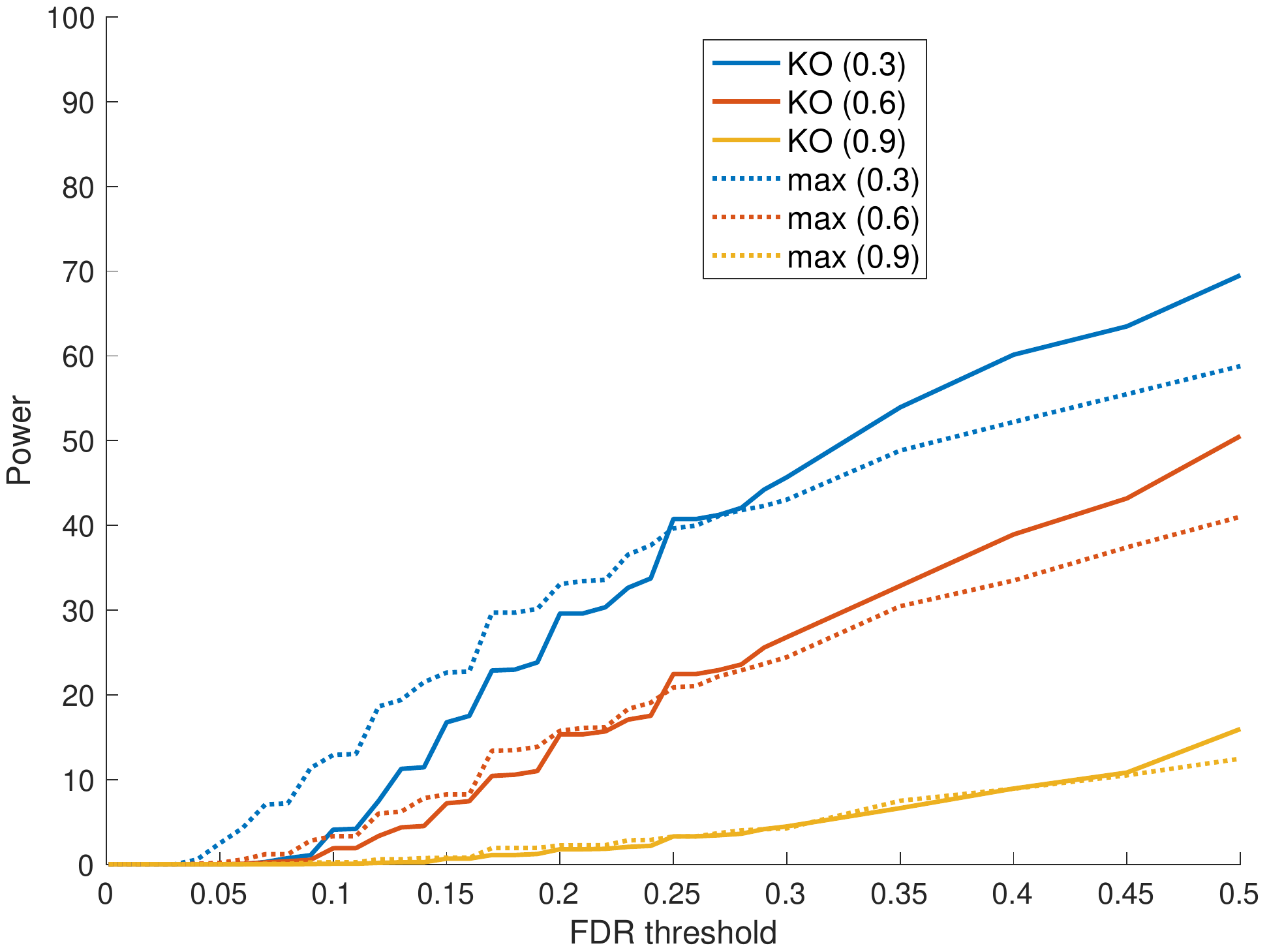}\tabularnewline
		E. varying $K$ (\# of features in model) & F. varying $K$;\tabularnewline
		\includegraphics[width=3in]{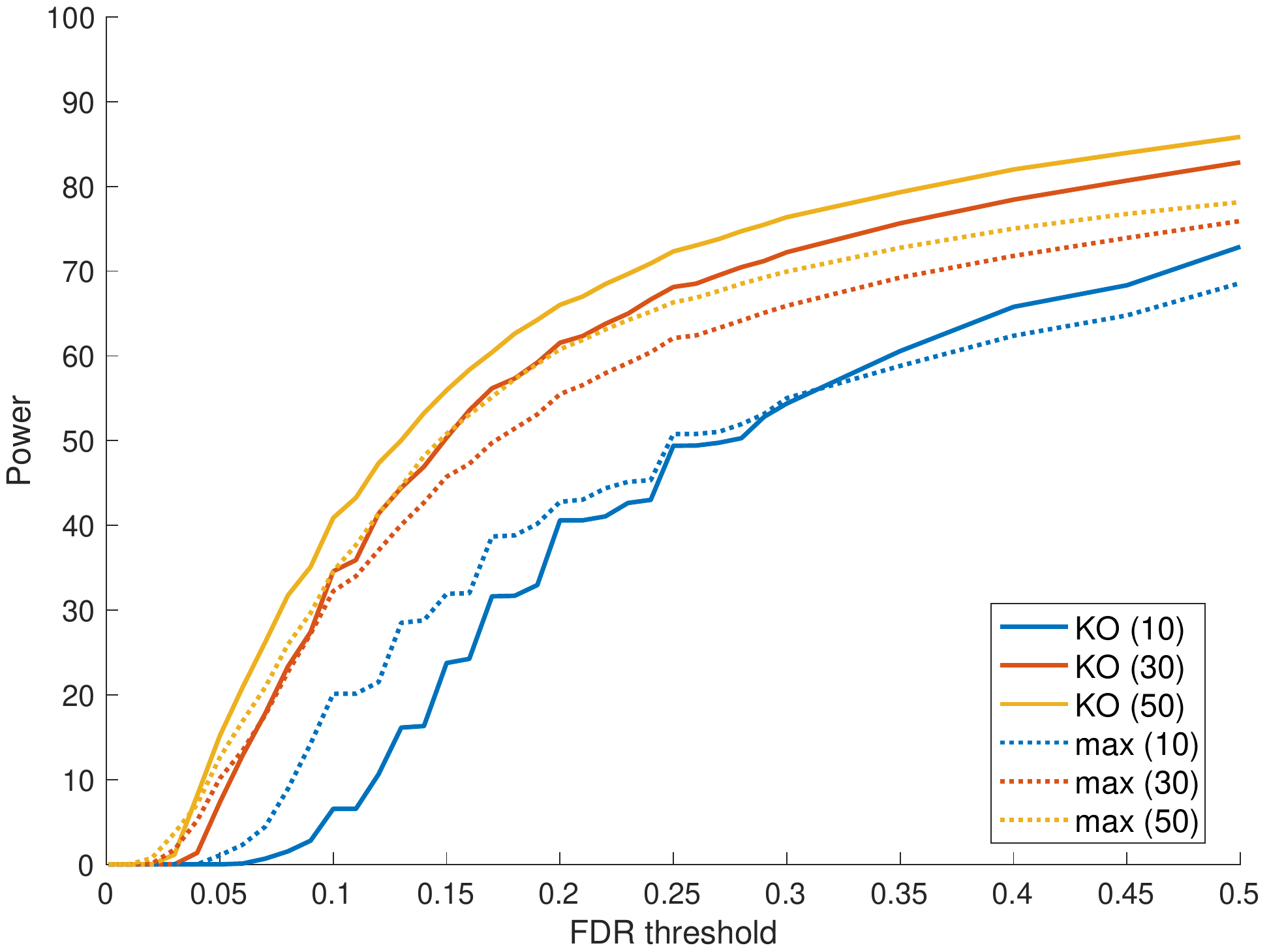} & \includegraphics[width=3in]{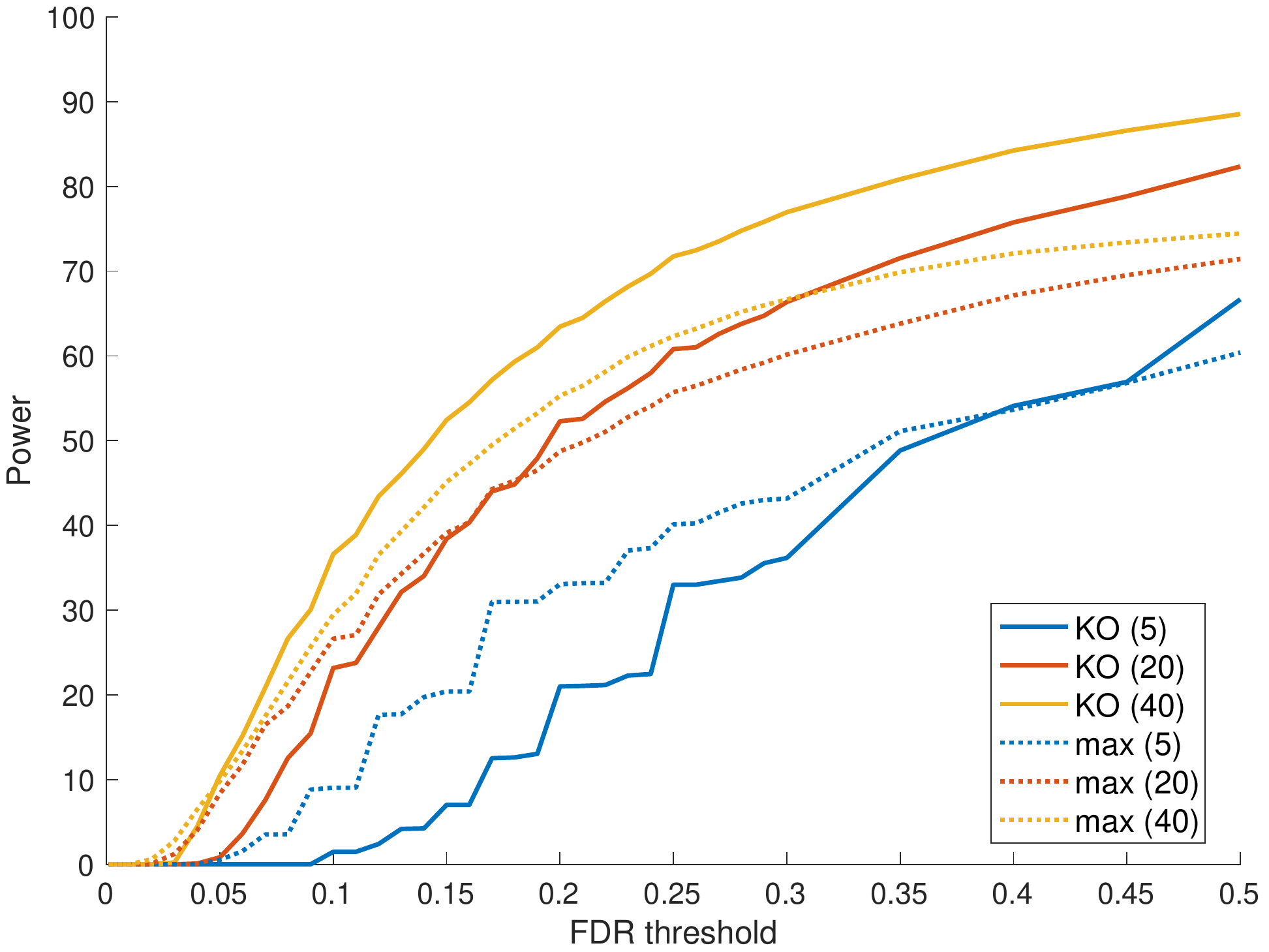}\tabularnewline
	\end{tabular}\caption{\textbf{Power plots of knockoff+ (KO) and max in specific cases of
			guaranteed FDR control. }Each left column panel gives the power of
		\BC' single knockoff juxtaposed with the max procedure using two
		knockoffs ($n=3000,p=1000$ and unless explicitly varied, $k=30,\rho=0,A=3.5$,
		\suppsec\ref{subsec:n3000_p1000_d1_2_b1}). The right column panels
		show knockoff+ and the max using three knockoffs ($n=800,p=200$
		and unless explicitly varied, $k=10,\rho=0,A=3.0$, \suppsec\ref{subsec:n800_p200_d1_3_b1}).
		(A) varying the amplitude: $A\in\left\{ 3.1,3.5,3.9\right\} $ . (B)
		varying the feature correlation strength of $\Theta_{\rho}$: $\rho\in\left\{ 0.0,0.3,0.6,0.9\right\} $.
		(C) varying the number of features, sparsity: $K\in\left\{ 10,30,50\right\} $.
		(D) varying the amplitude: $A\in\left\{ 2.6,3.0,3.4\right\} $. (E)
		varying $\Theta_{\rho}$'s $\rho$: $\rho\in\left\{ 0.0,0.3,0.6,0.9\right\} $.
		(F) varying $K\in\left\{ 5,20,40\right\} $. Overall max tends to
		do better for smaller FDR thresholds, sparser models and a larger
		$d$ but the results are mixed.\label{fig:supp-rigor-FDR-control}}
\end{figure}

\begin{figure}
	\centering %
	\begin{tabular}{ll}
		A. Max ($d=2,3$) vs.~knockoff+ & B. Empirical FDR (max)\tabularnewline
		\includegraphics[width=3in]{figures/guaranteed_FDR_r2376_r2396_power_KO_vs_max_} & \includegraphics[width=3in]{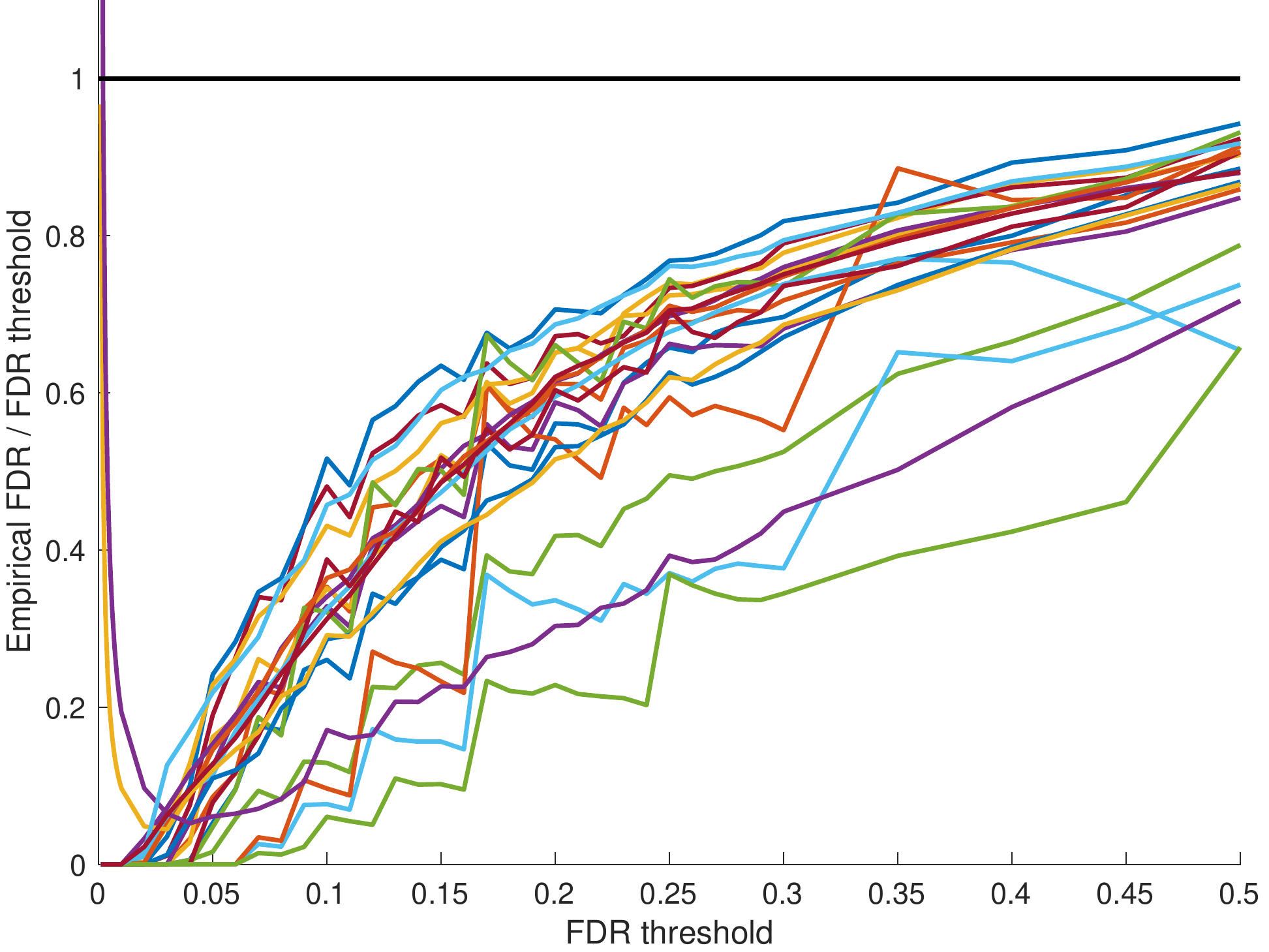}\tabularnewline
		C. Mirror ($d=2,3$) vs.~knockoff+ & D. Empirical FDR (mirror)\tabularnewline
		\includegraphics[width=3in]{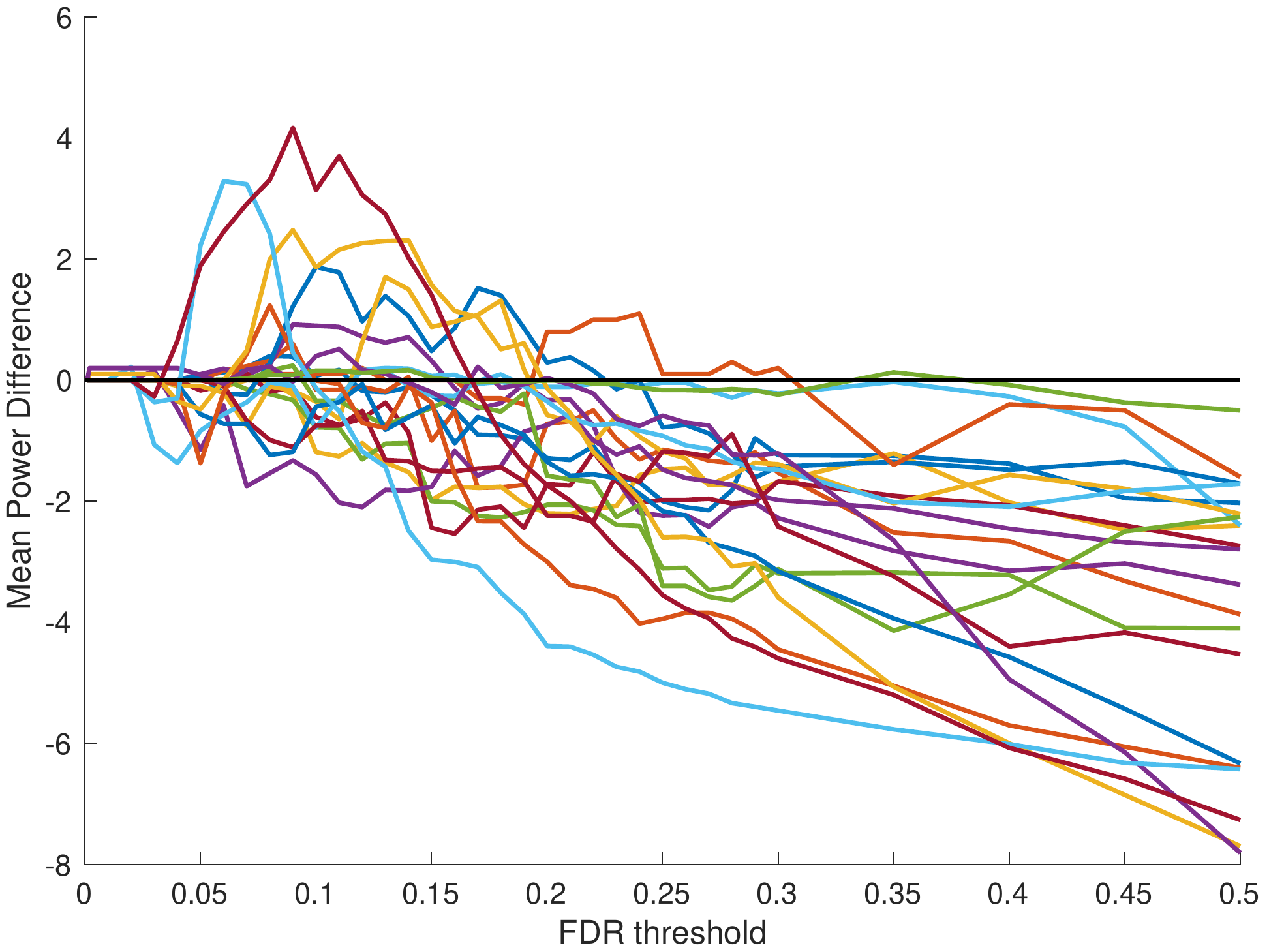} & \includegraphics[width=3in]{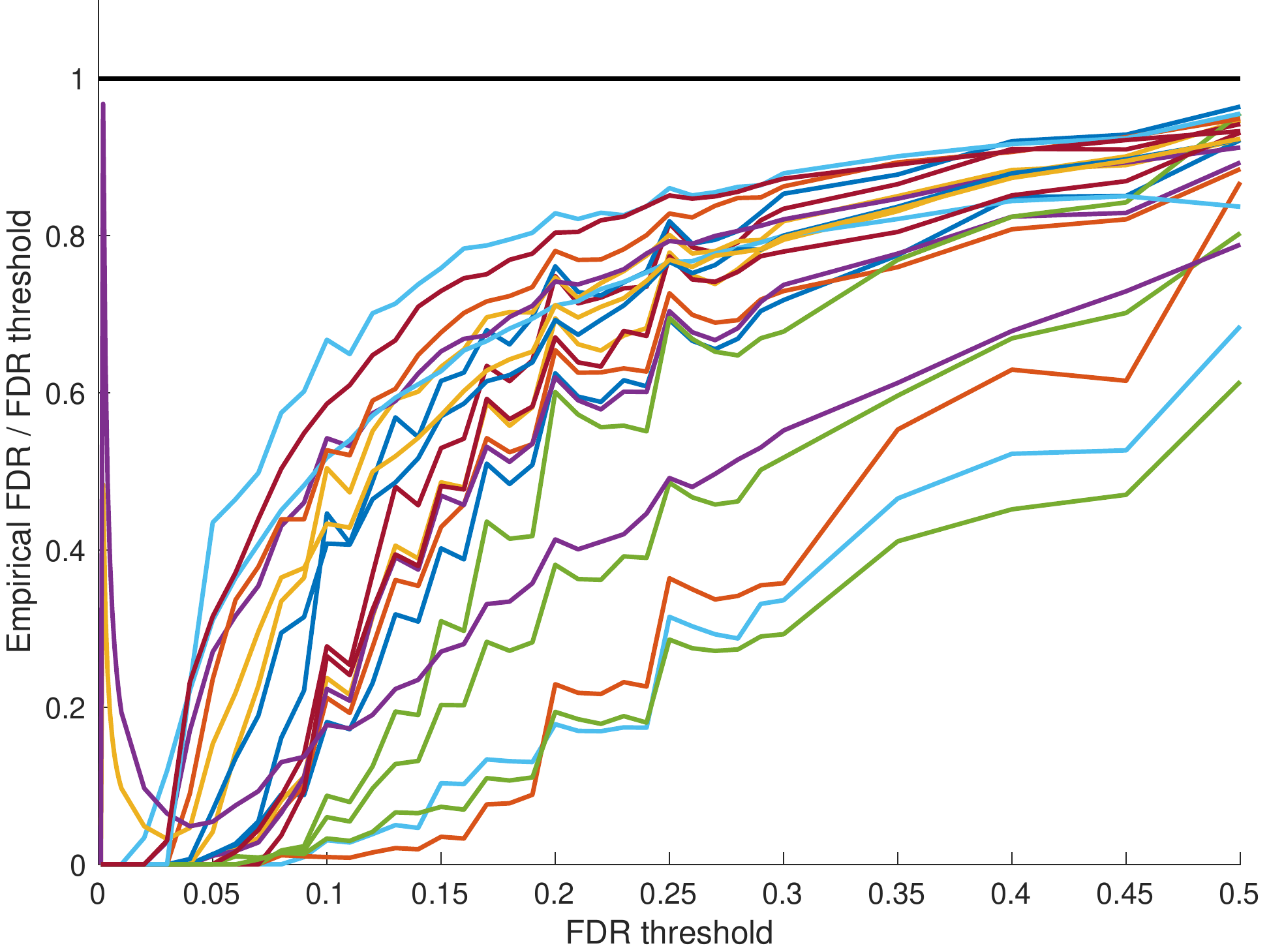}\tabularnewline
		E. Batched-knockoff+ (1 batch) vs.~knockoff+ & F. Empirical FDR (batched-knockoff+)\tabularnewline
		\includegraphics[width=3in]{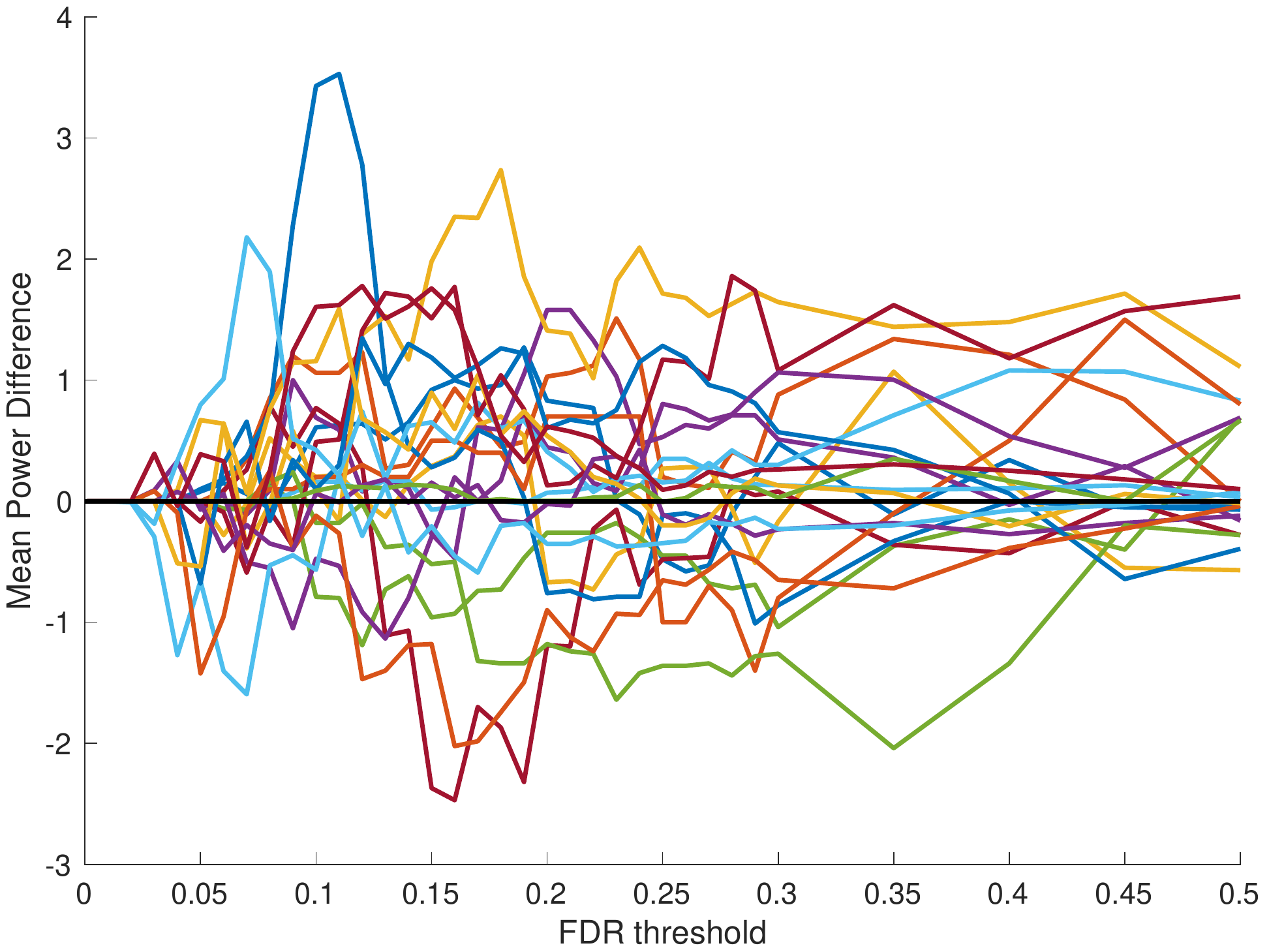} & \includegraphics[width=3in]{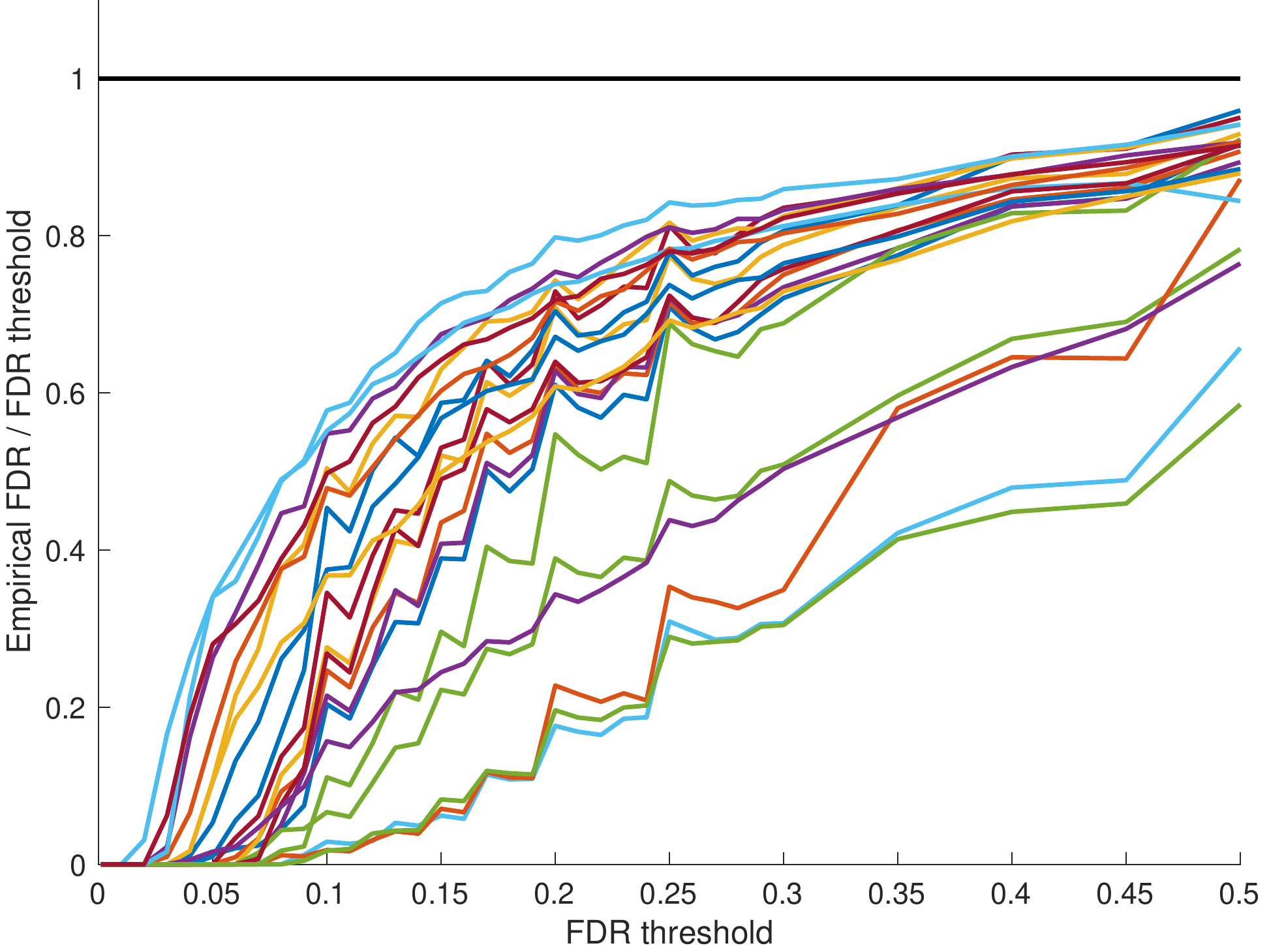}\tabularnewline
	\end{tabular}\caption{\textbf{Comparison with knockoff+ and FDR control in guaranteed settings.
		}The left column panels show the difference between the power of the
		considered method and knockoff+' power (negative values indicate knockoff+
		is more powerful at that threshold). The right column panels examine
		the ratio of the empirical FDR of the considered method (averaged
		over 1K runs) to the FDR threshold. The data for both columns consists
		of experiments in which all methods have guaranteed FDR control: $n=3000,p=1000,d=2,b=1$
		(\suppsec\ref{subsec:n3000_p1000_d1_2_b1}) and $n=800,p=200,d=3,b=1$
		(\suppsec\ref{subsec:n800_p200_d1_3_b1}). 
		(B) There is a single (essentially
		random) spike at the FDR threshold of 0.001 (0.1\%) where for that
		particular set of parameters from \suppsec\ref{subsec:n3000_p1000_d1_2_b1}
		the empirical FDR of max is just below 0.2\% so it is almost 20\%
		over the threshold.
		(E) batched-knockoff+ vs.~knockoff+. As $b=1$ here the two methods are essentially equivalent so variations
		in power are essentially random.  (F) batched-knockoff+ (knockoff+'
		FDR plot is naturally quite similar to this one). \label{fig:supp-rigor-FDR-control-2}}
\end{figure}

%

\begin{figure}
	\centering %
	\begin{tabular}{ll}
		A. Power of mirror ($n=600,p=200$) & B. FDR of mirror\tabularnewline
		\includegraphics[width=3in]{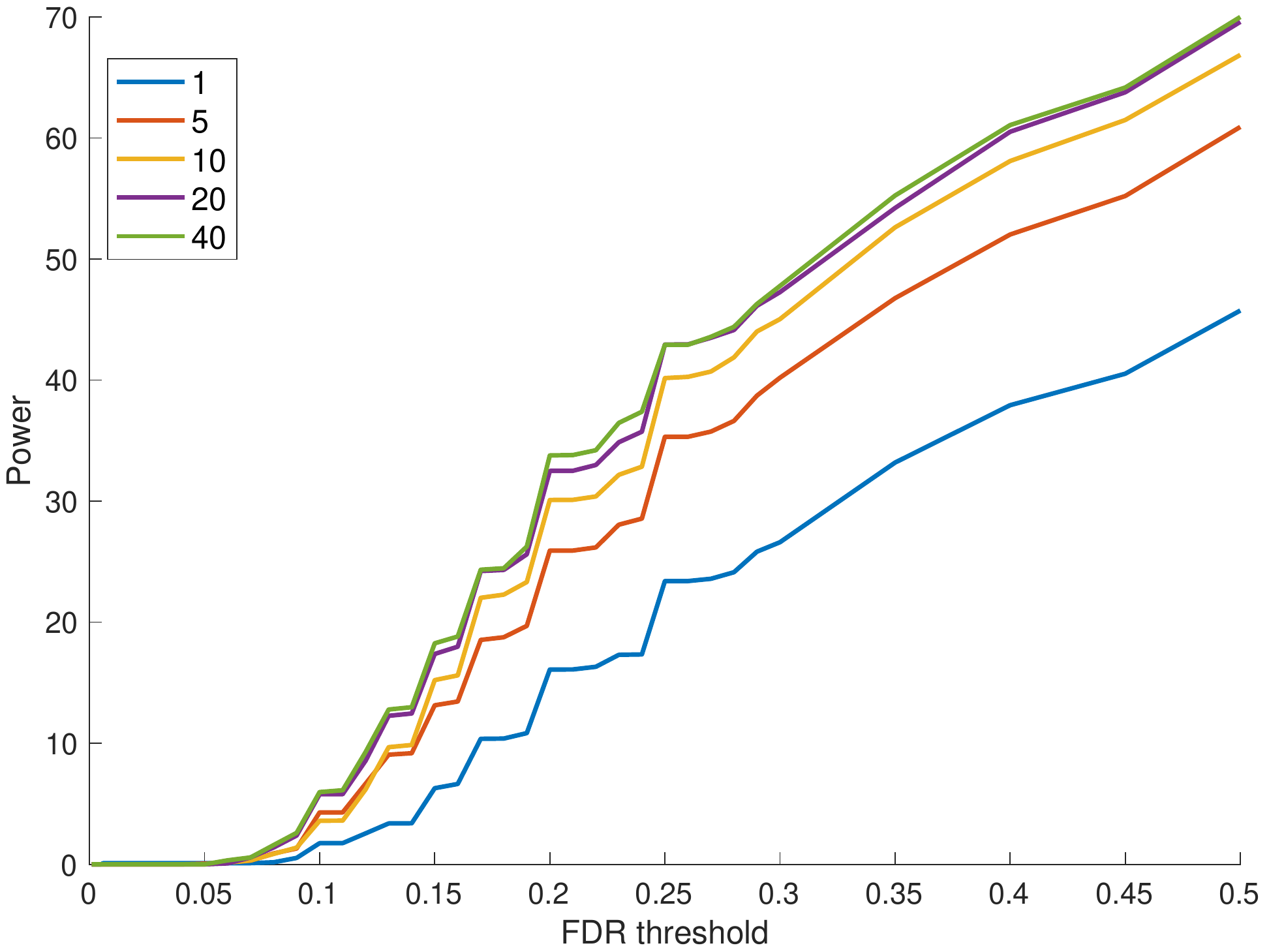} & \includegraphics[width=3in]{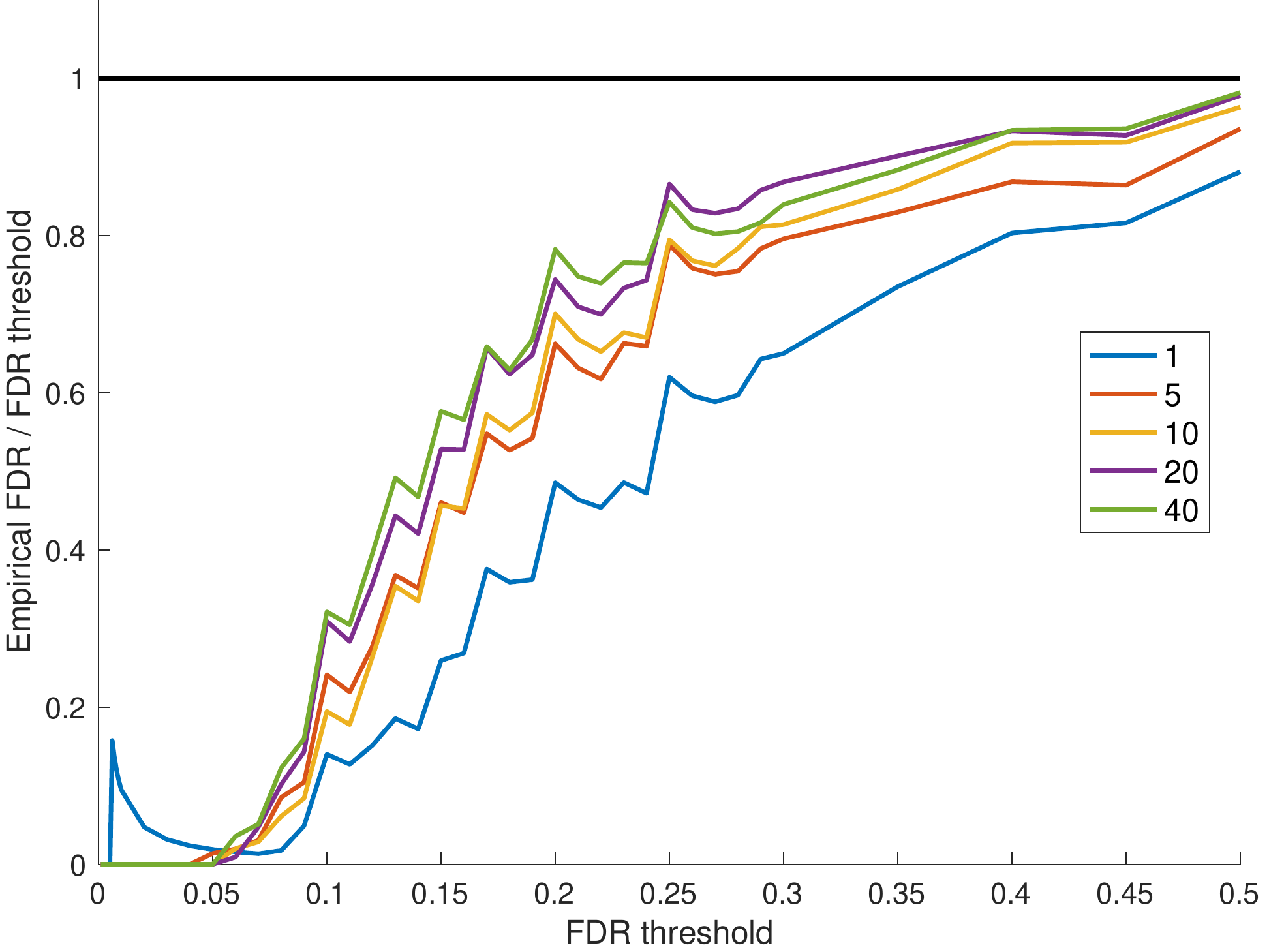}\tabularnewline
		C. Power of batched-knockoff+ & D. FDR of batched-knockoff+\tabularnewline
		\includegraphics[width=3in]{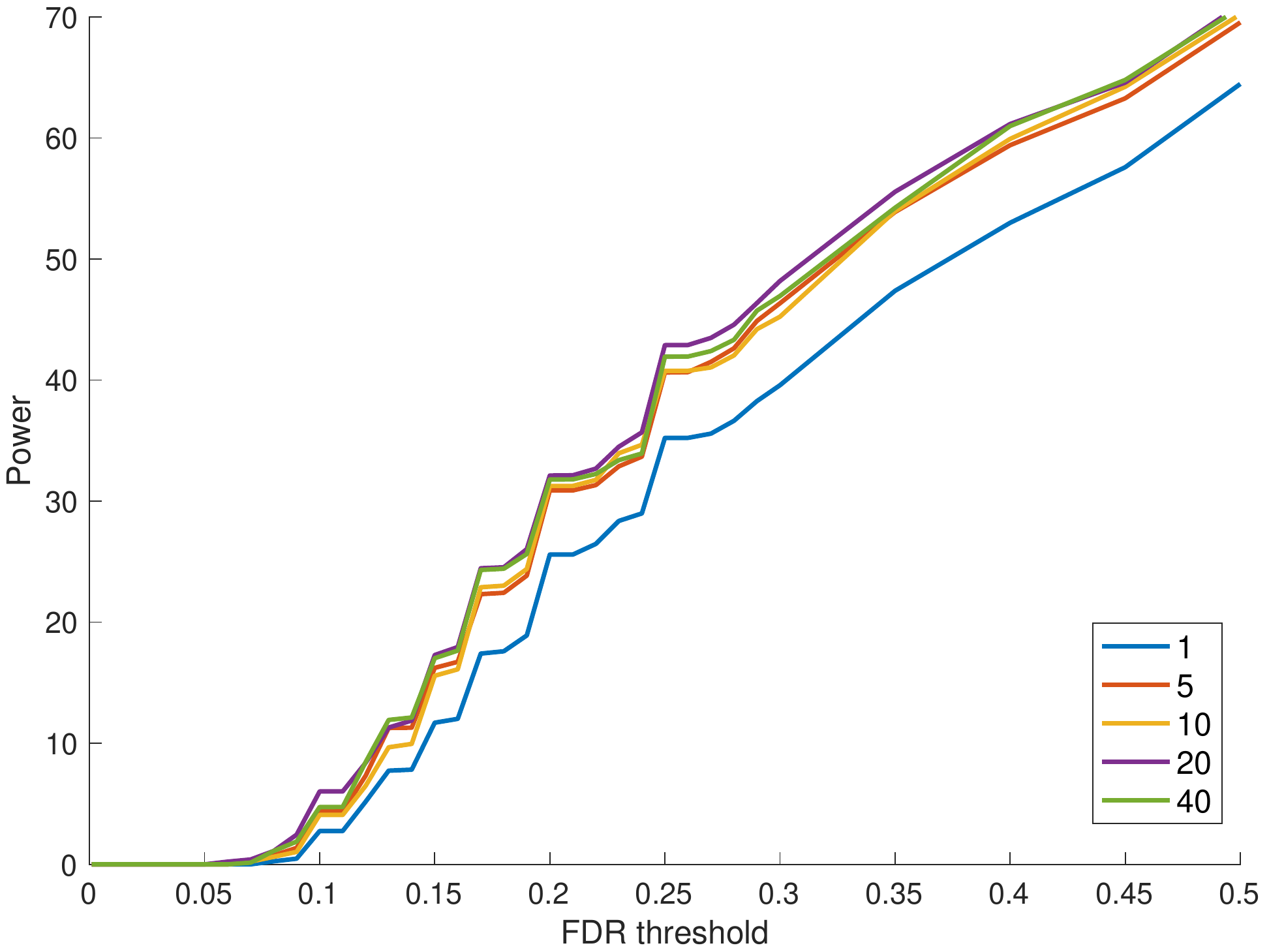} & \includegraphics[width=3in]{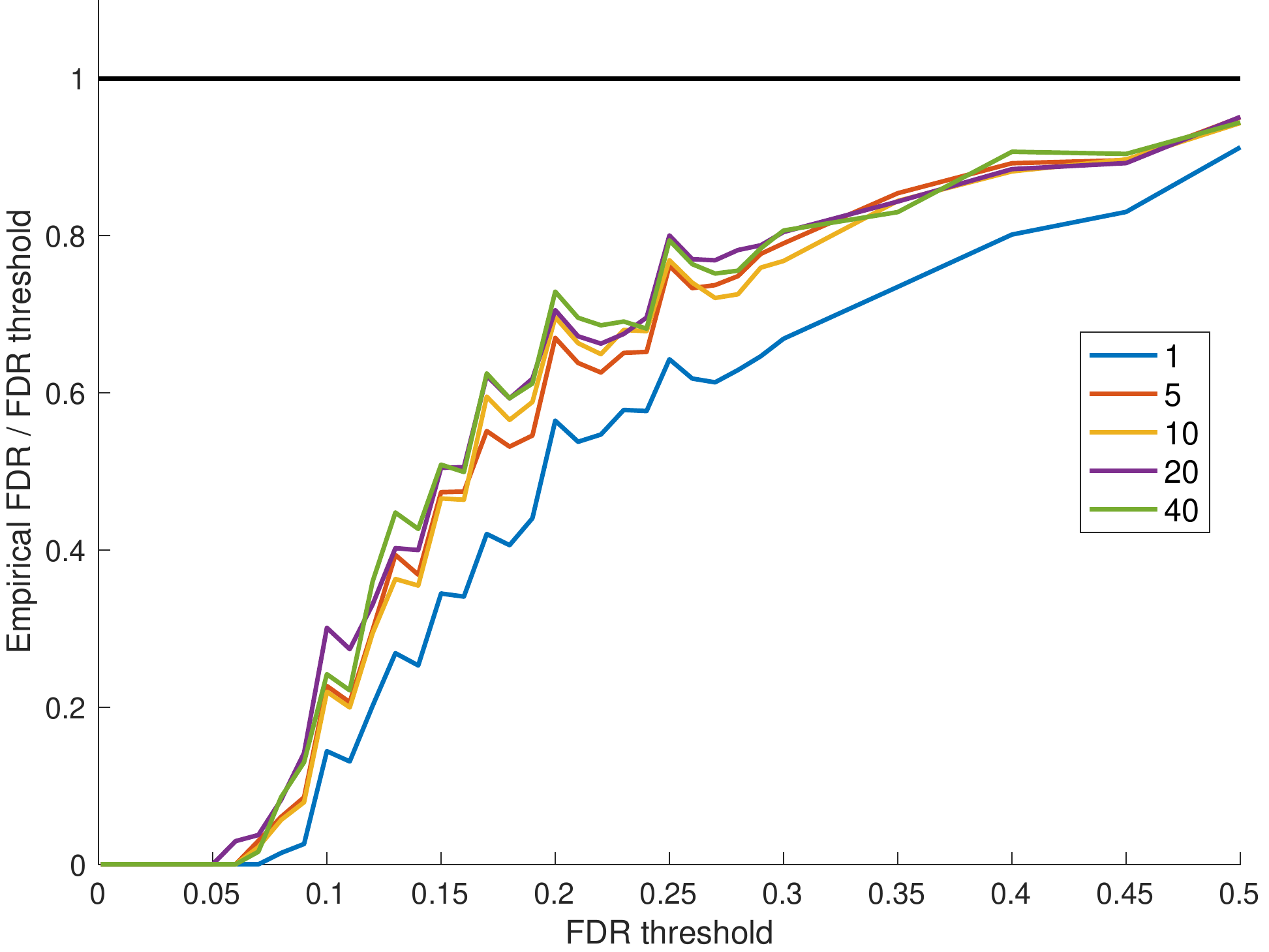}\tabularnewline
		E. Power of knockoff+ (variations are random) & F. FDR of max\tabularnewline
		\includegraphics[width=3in]{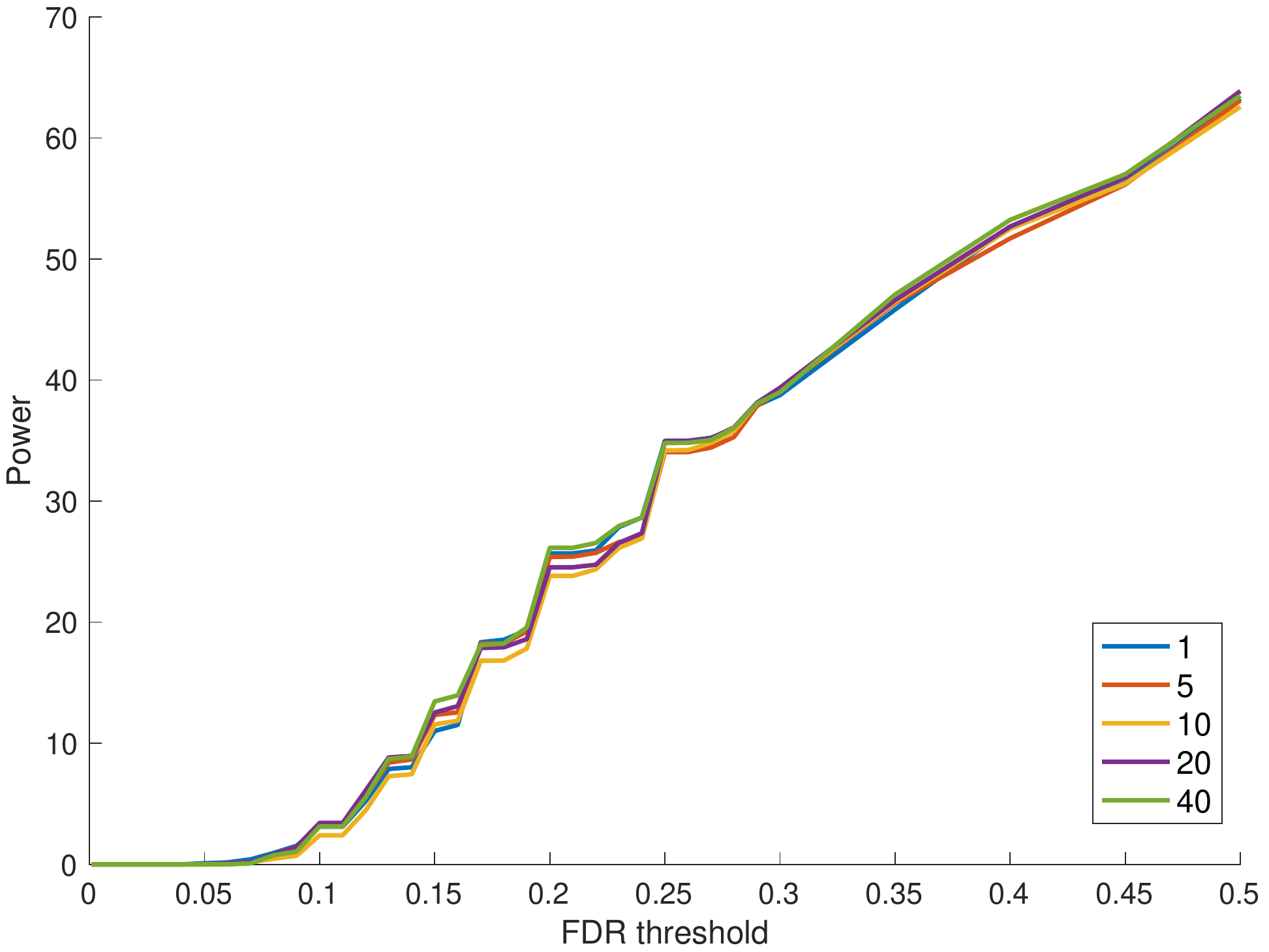} & \includegraphics[width=3in]{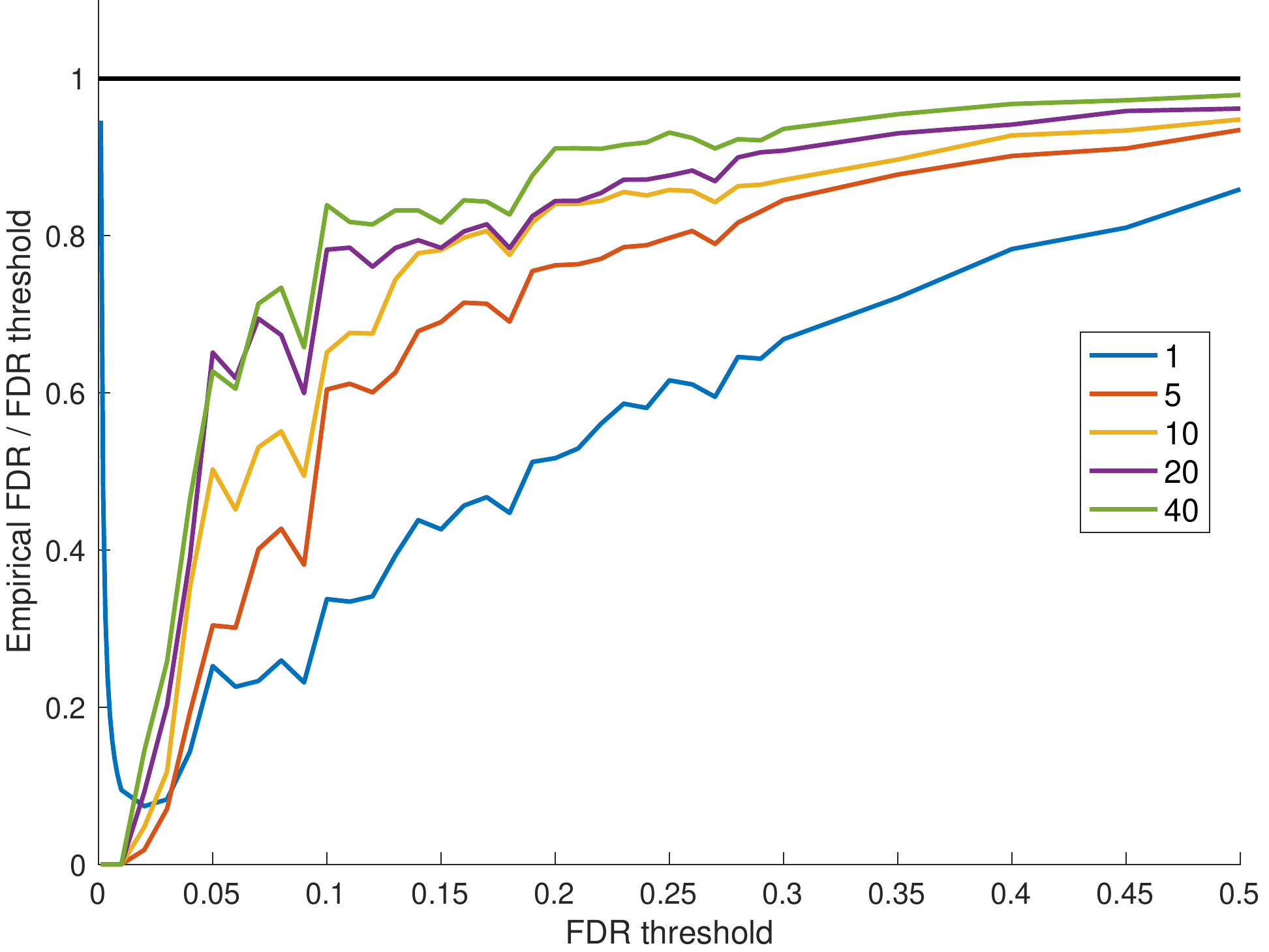}\tabularnewline
	\end{tabular}\caption{\textbf{Power increases with the number of batches (I).}
		Each of the left column panels shows the power of one method applied
		using a different number of batches $b\in\left\{ 1,5,10,20,40\right\} $
		to construct the knockoffs. The design of the experiment involved
		randomly drawing a new set of 1K datasets with for each value of $b$
		($n=600,p=200$, \suppsec\ref{subsec:n600_p200_d11_varying_b}).
		Each of the right column panels shows the ratio of the empirical FDR
		to the FDR threshold. (E) knockoff+ is not affected by the number
		of batches hence the observed power variations give us some idea on
		the magnitude of the differences due to the randomly generated datasets.\label{fig:supp_varying_b_n600}}
\end{figure}

\begin{figure}
	\centering %
	\begin{tabular}{ll}
		A. Power of max ($n=3000,p=1000$) & B. FDR of max\tabularnewline
		\includegraphics[width=3in]{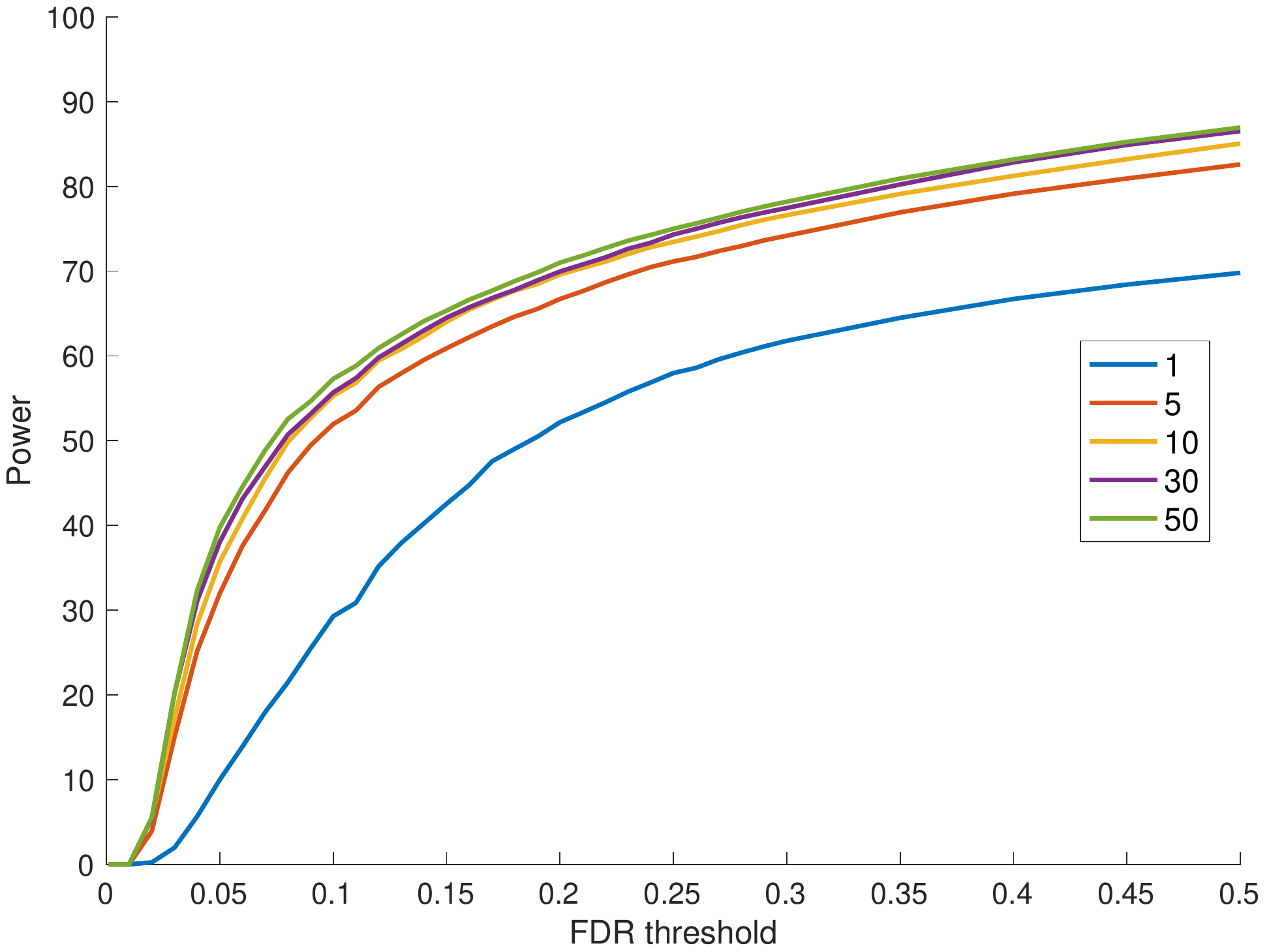} & \includegraphics[width=3in]{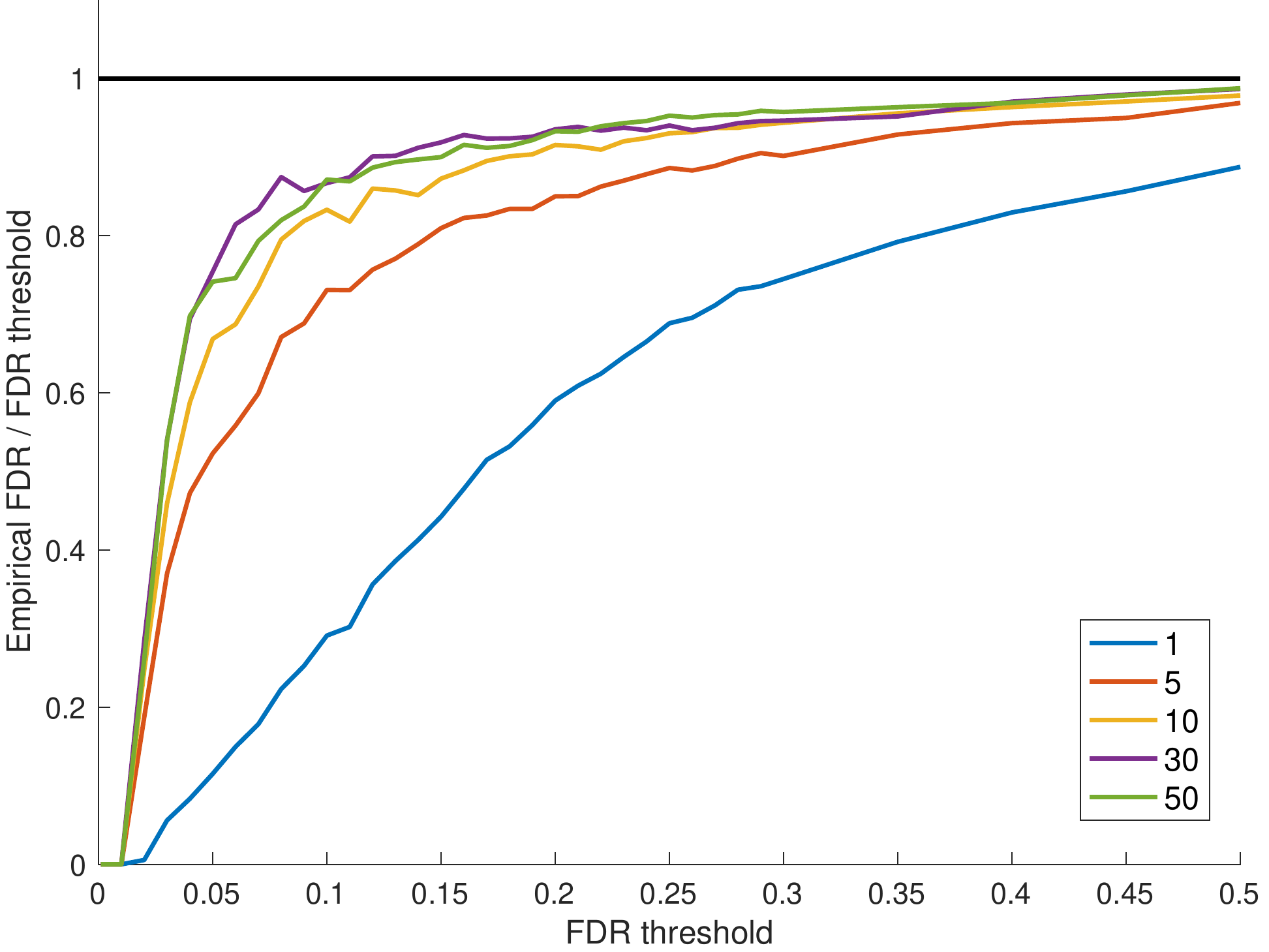}\tabularnewline
		C. Power of batched-knockoff+ & D. FDR of batched-knockoff+\tabularnewline
		\includegraphics[width=3in]{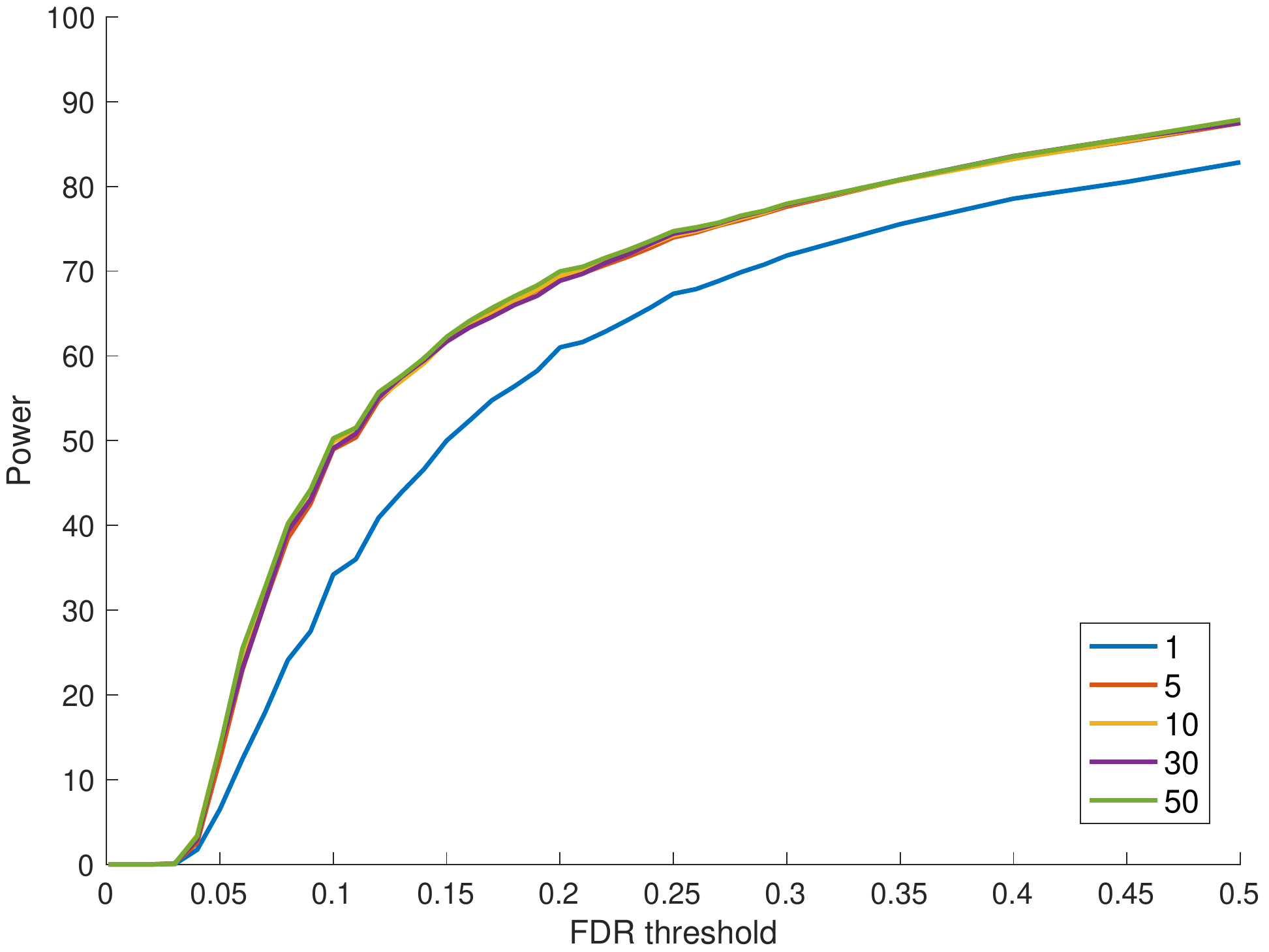} & \includegraphics[width=3in]{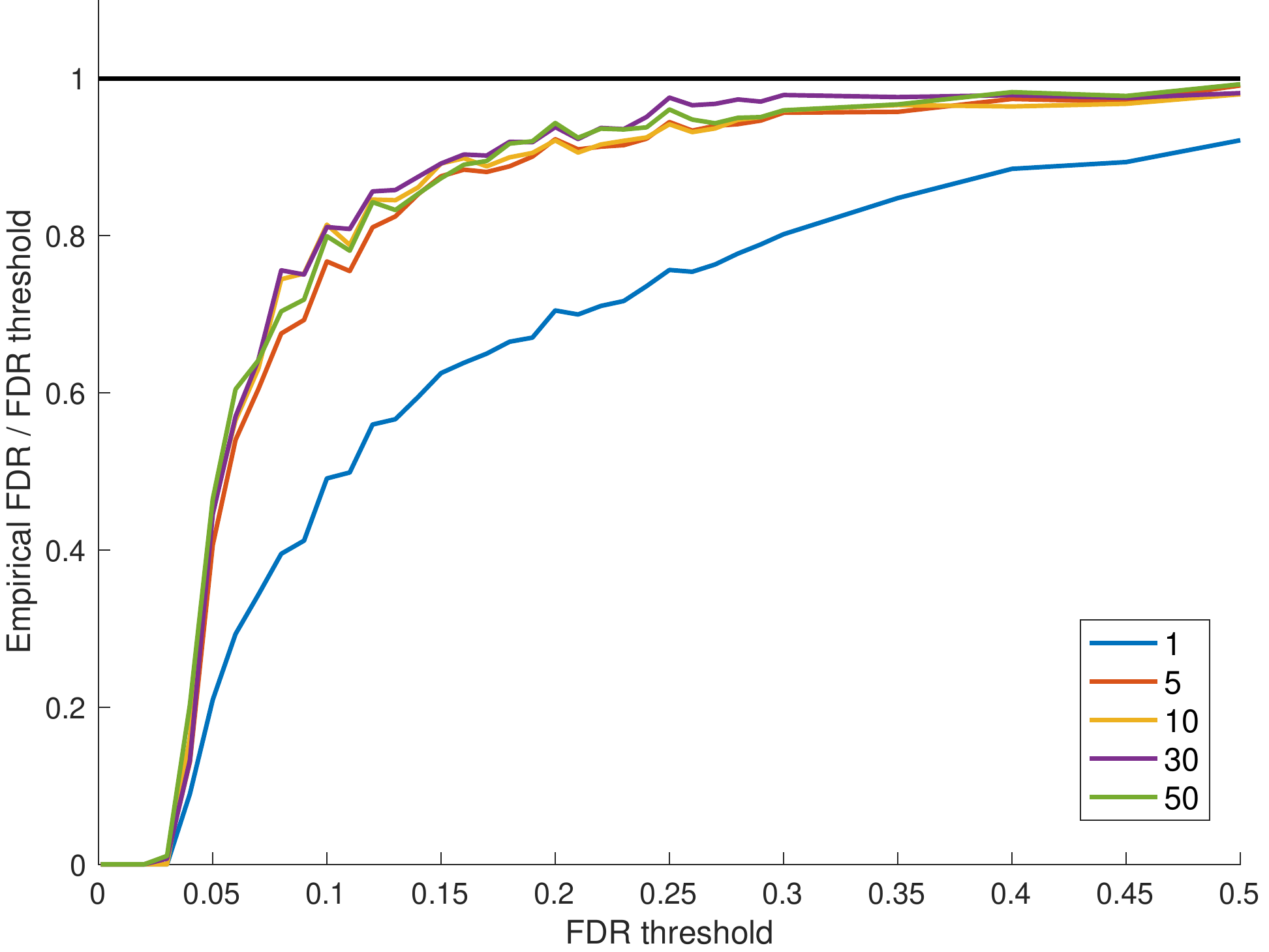}\tabularnewline
		E. Power of knockoff+ (variations are random) & F. FDR of mirror\tabularnewline
		\includegraphics[width=3in]{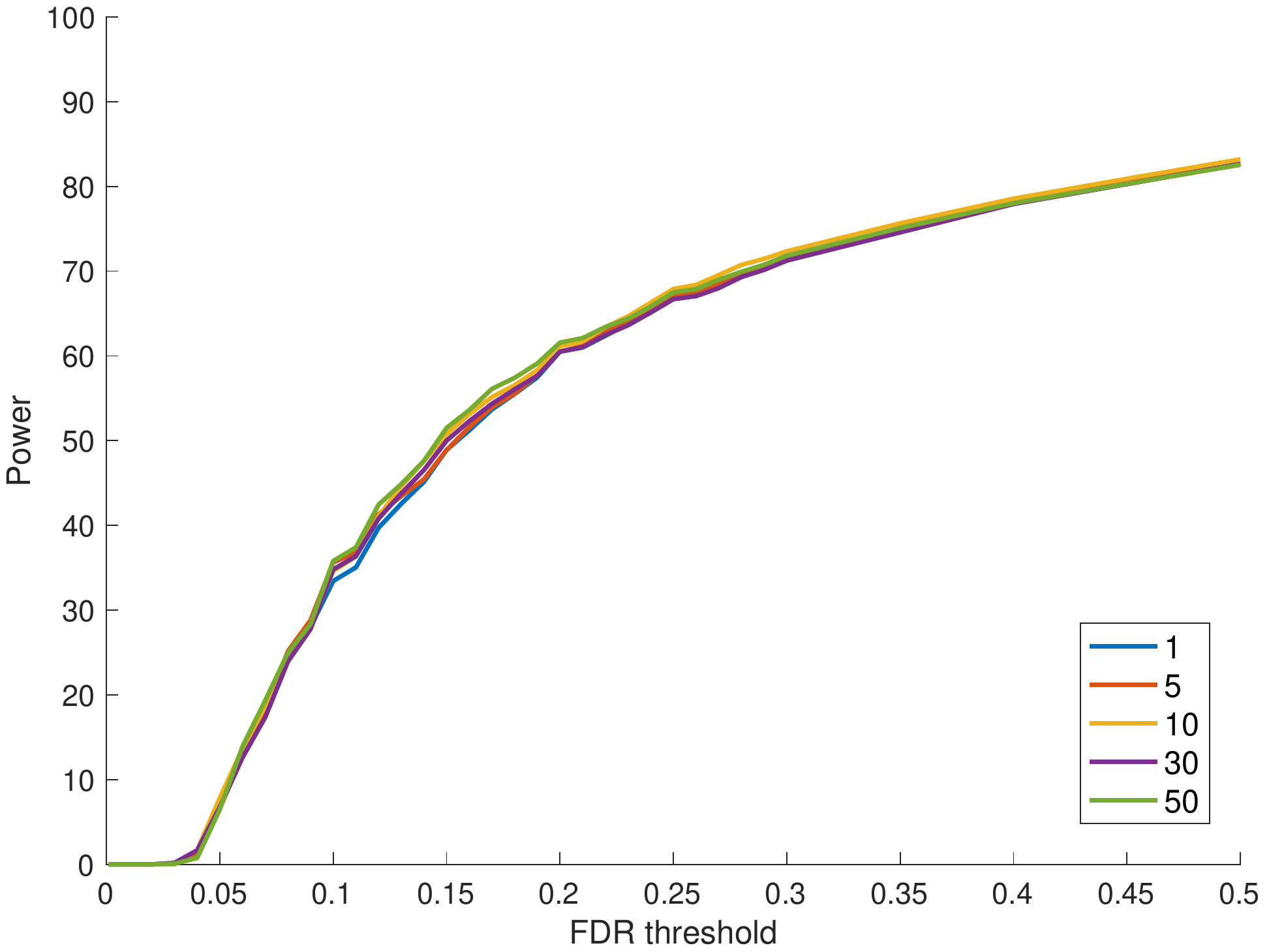} & \includegraphics[width=3in]{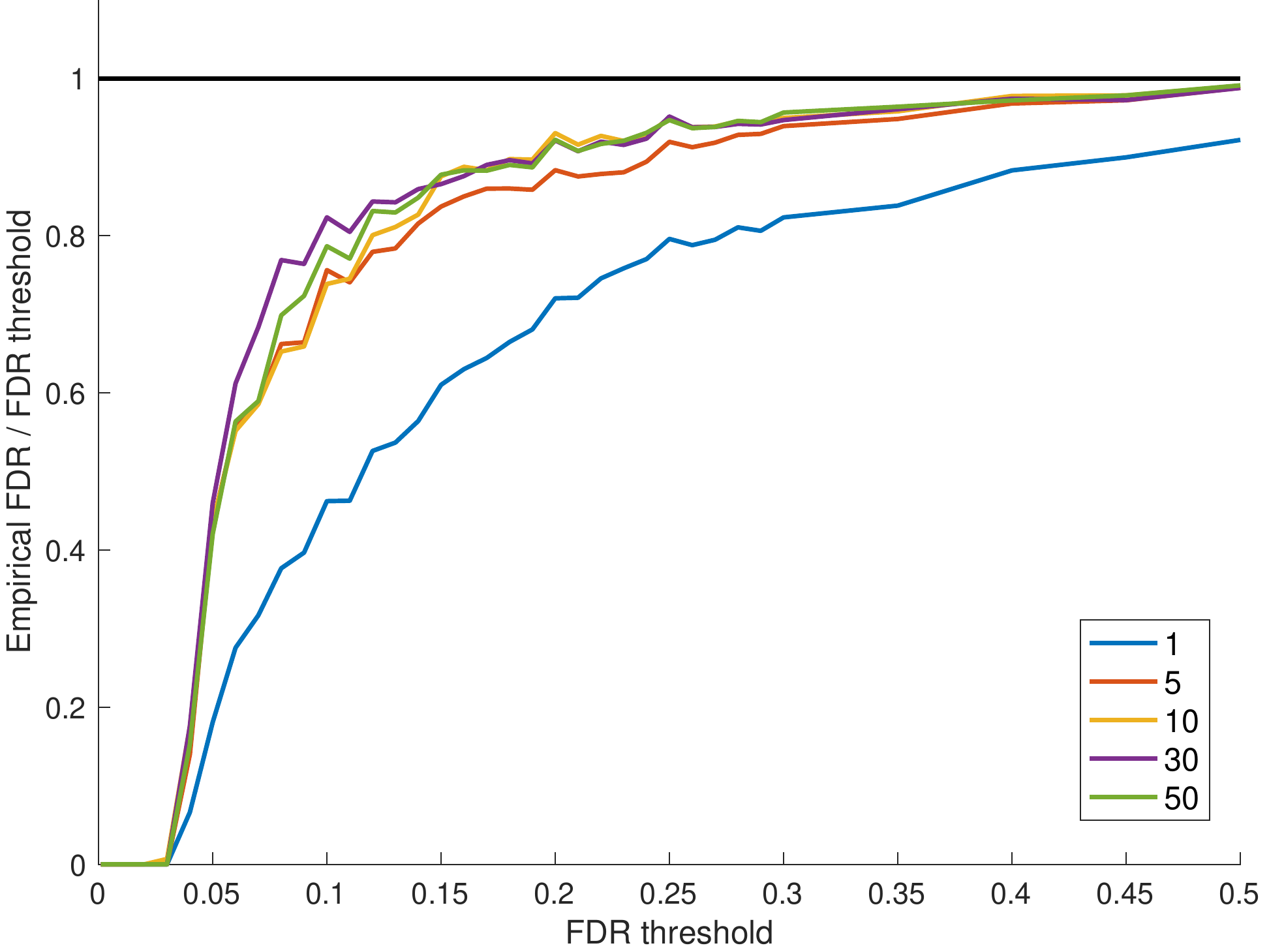}\tabularnewline
	\end{tabular}\caption{\textbf{Power increases with the number of batches (II).}
		Same as \suppfig\ref{fig:supp_varying_b_n600} except $b\in\left\{ 1,5,10,30,50\right\} $
		and $n=3000,p=1000$ (\suppsec\ref{subsec:n3000_p1000_d3_varying_b}).\label{fig:supp_varying_b_n3000}}
\end{figure}

\begin{figure}
	\centering %
	\begin{tabular}{ll}
		A. Max using $b=1$ vs.~$b=40$ batches & B. Mirror\tabularnewline
		\includegraphics[width=3in]{figures/n800_nsko1_3_nboot32_r2376_r2396_batch_paired_power_diff_max} & \includegraphics[width=3in]{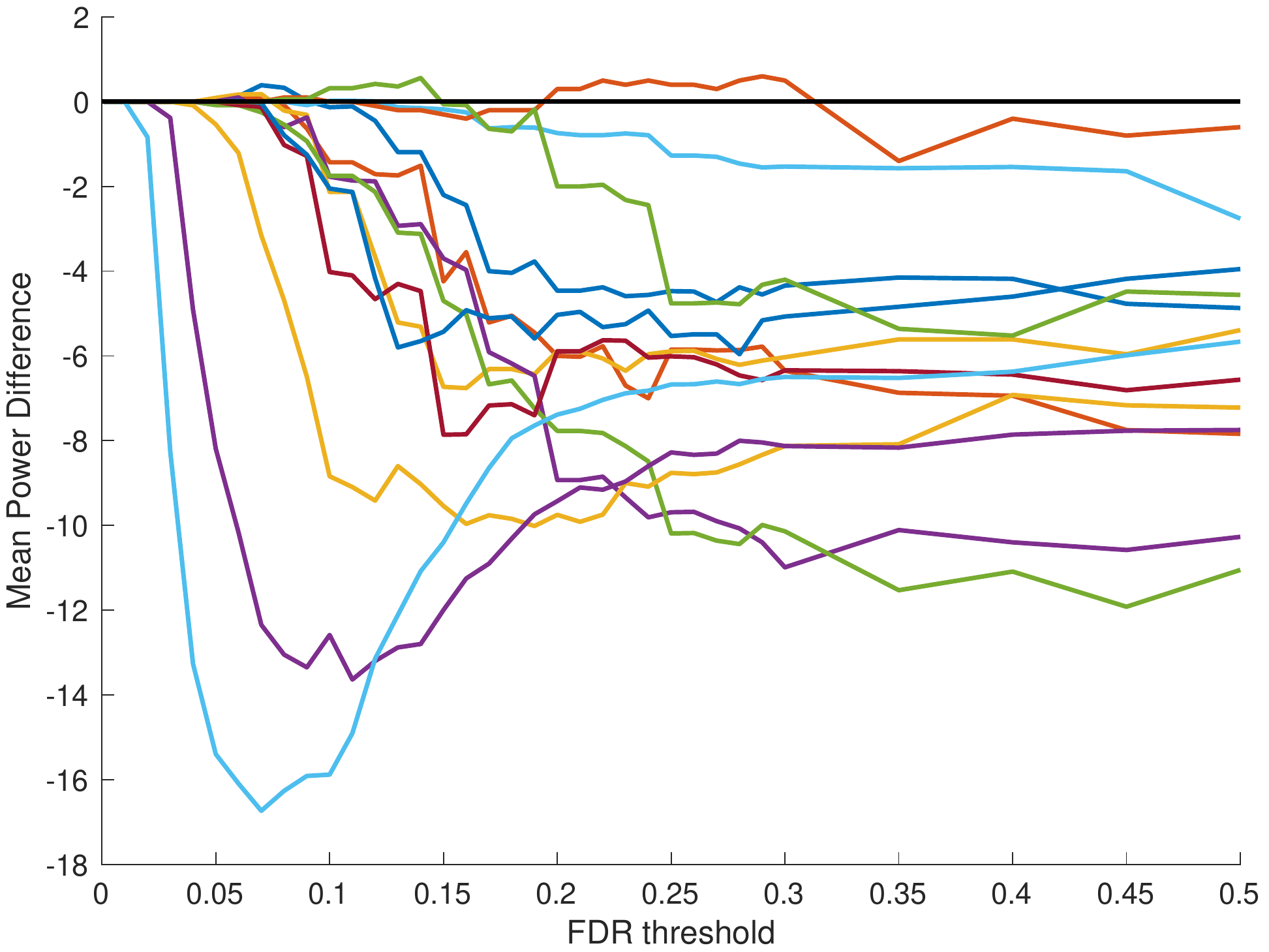}\tabularnewline
		C. batched-knockoff+ & D. knockoff+ \tabularnewline
		\includegraphics[width=3in]{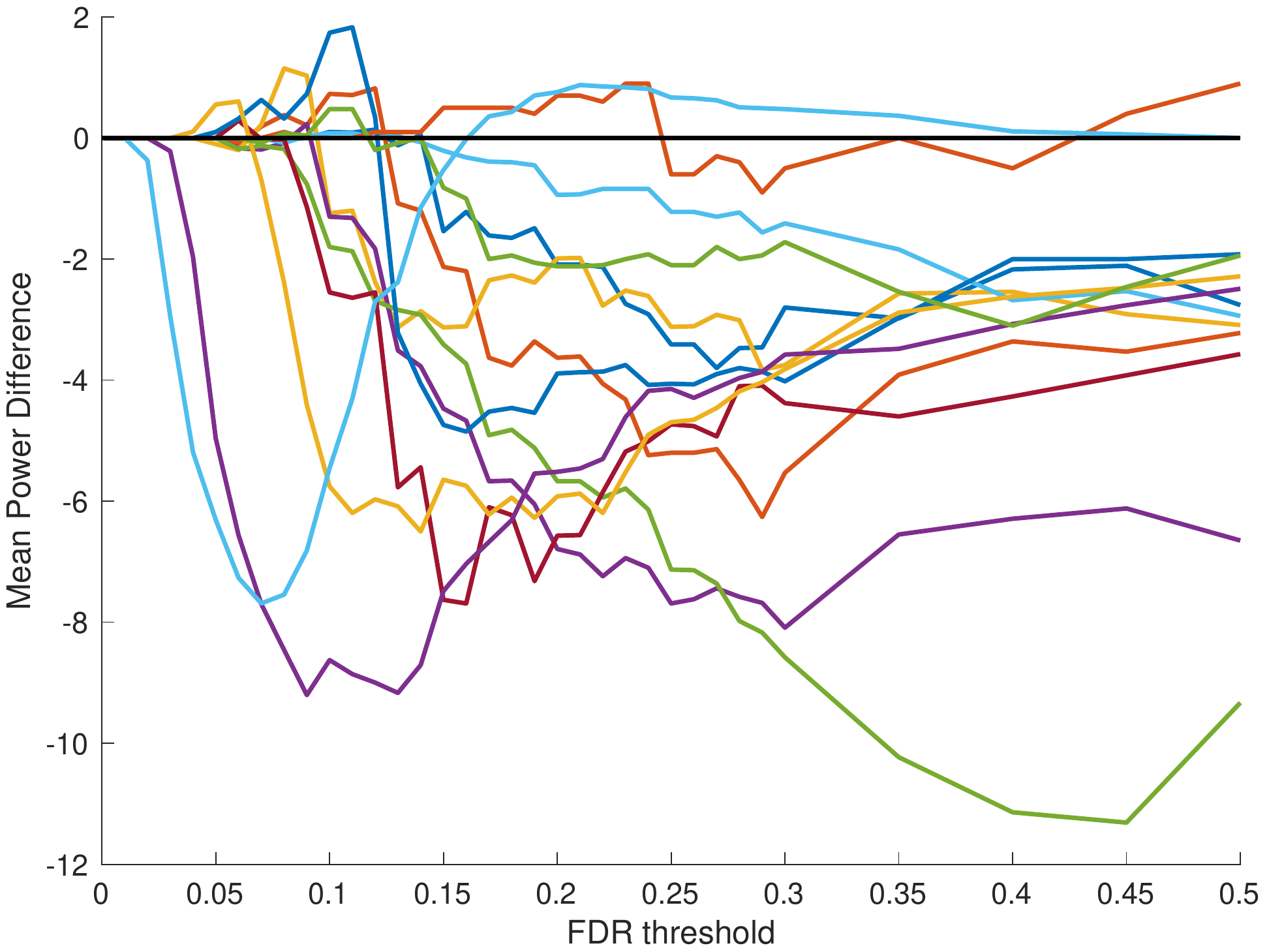} & \includegraphics[width=3in]{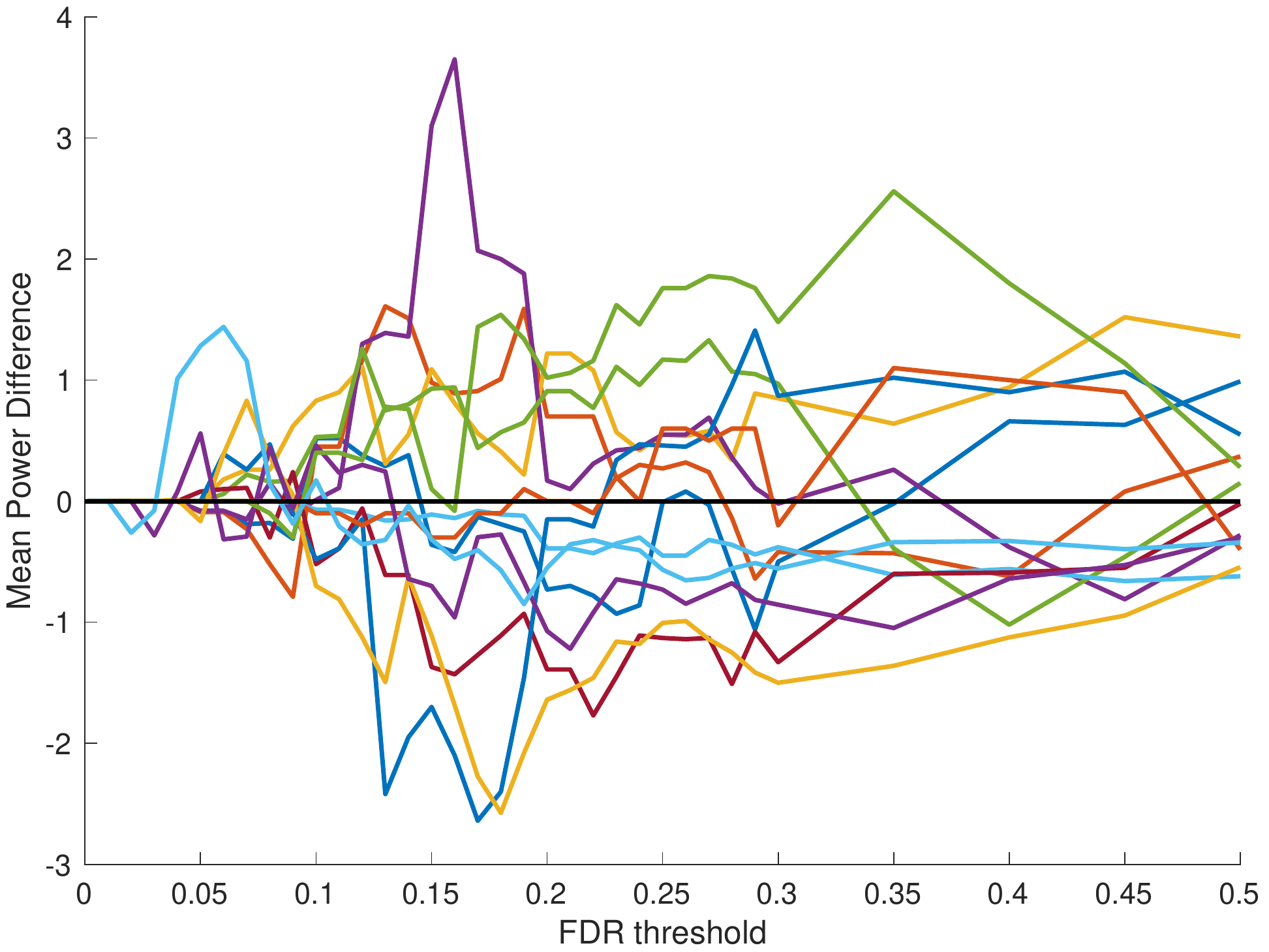}\tabularnewline
	\end{tabular}\caption{\textbf{Batching increases the power (single vs.~40 batches). }Each
		of the panels shows the difference in the power of one method applied to
		multiple datasets using $b=1$ ($n=800,p=200$, \suppsec\ref{subsec:n800_p200_d1_3_b1}),
		and $b=40$ batches ($n=800,p=200$, \suppsec\ref{subsec:n800_p200_d1_3_b40}).
		The design of the experiment involved drawing a new set of 1K datasets for each value of $b$.
		(D) knockoff+' does not use batching so variations are simply due to the differences in the randomly drawn datasets. \label{fig:supp_batching_power_1vs40}}
\end{figure}

\begin{figure}
	\centering %
	\begin{tabular}{ll}
		A. Empirical FDR (max) & B. Empirical FDR (mirror)\tabularnewline
		\includegraphics[width=3in]{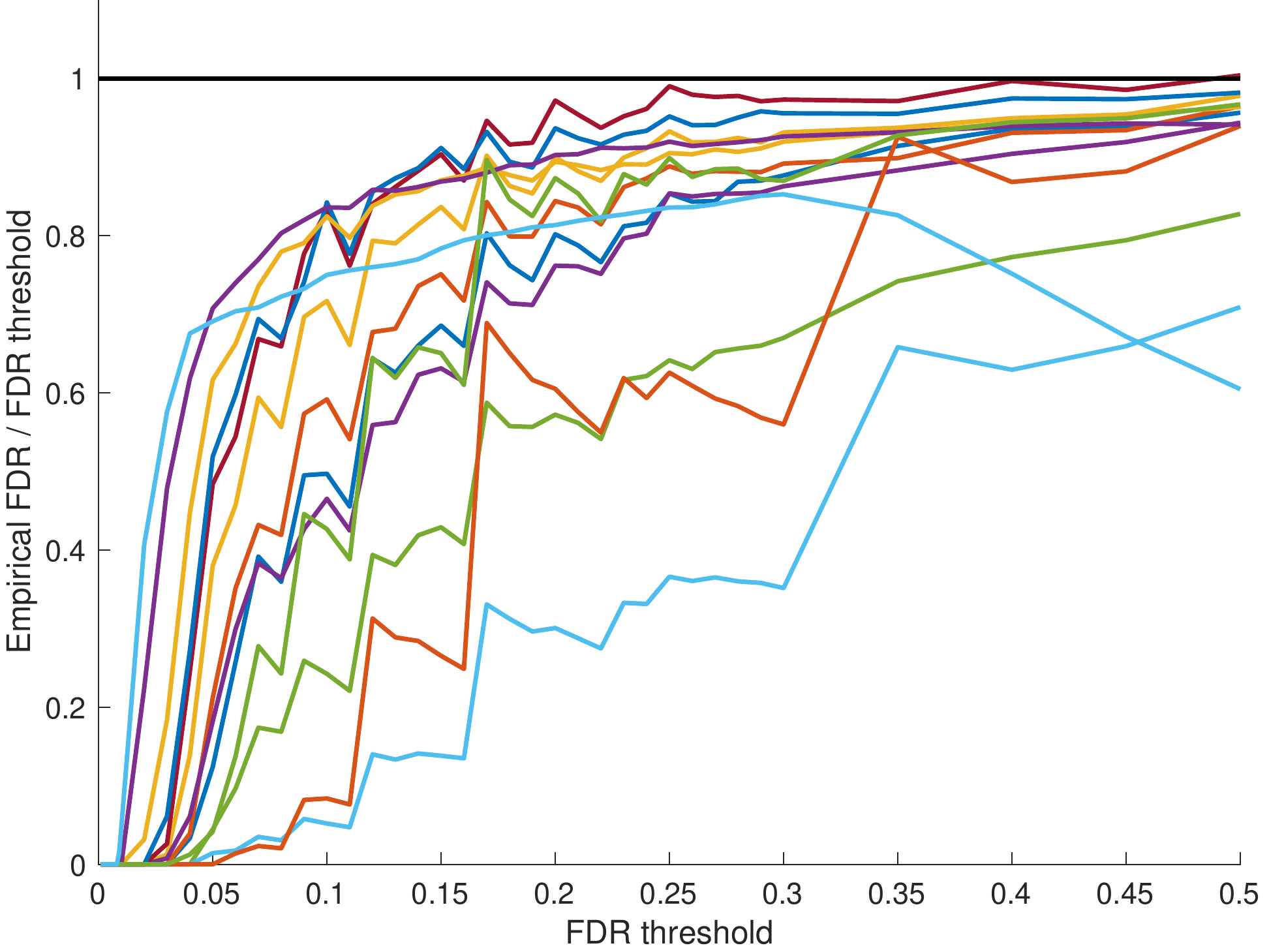} & \includegraphics[width=3in]{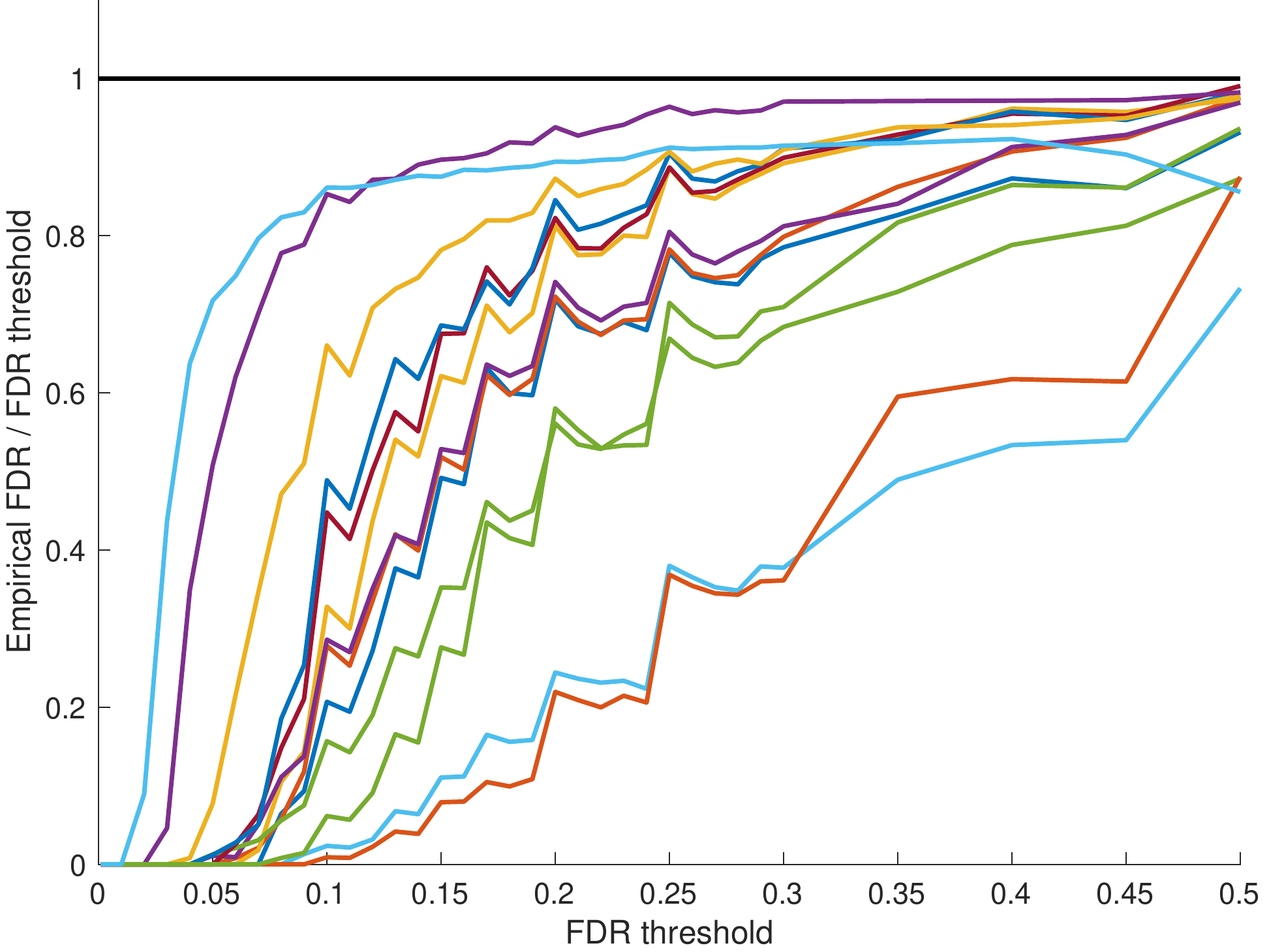}\tabularnewline
		C. Empirical FDR (batched-knockoff+) & D. Empirical FDR (knockoff+)\tabularnewline
		\includegraphics[width=3in]{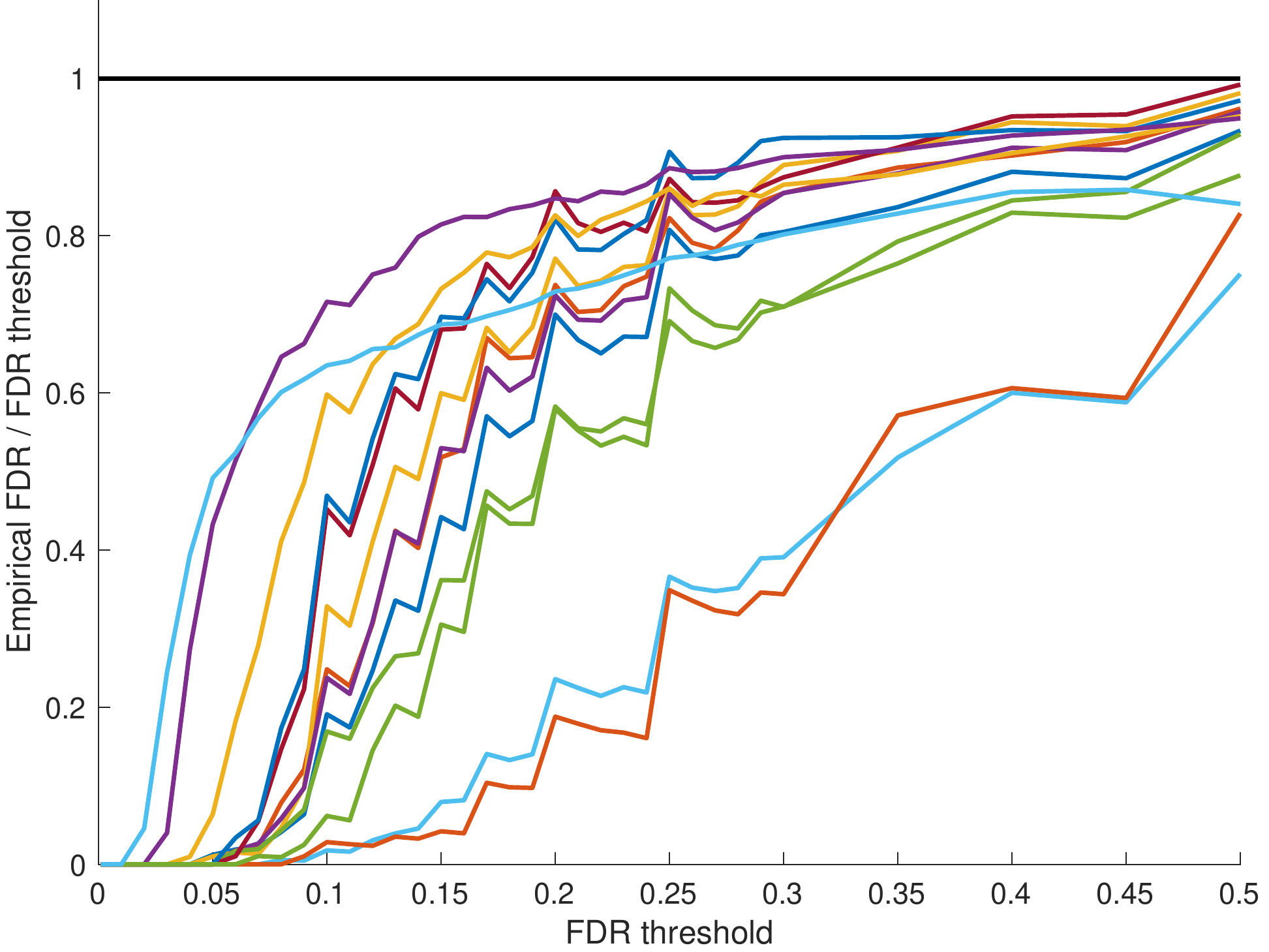} & \includegraphics[width=3in]{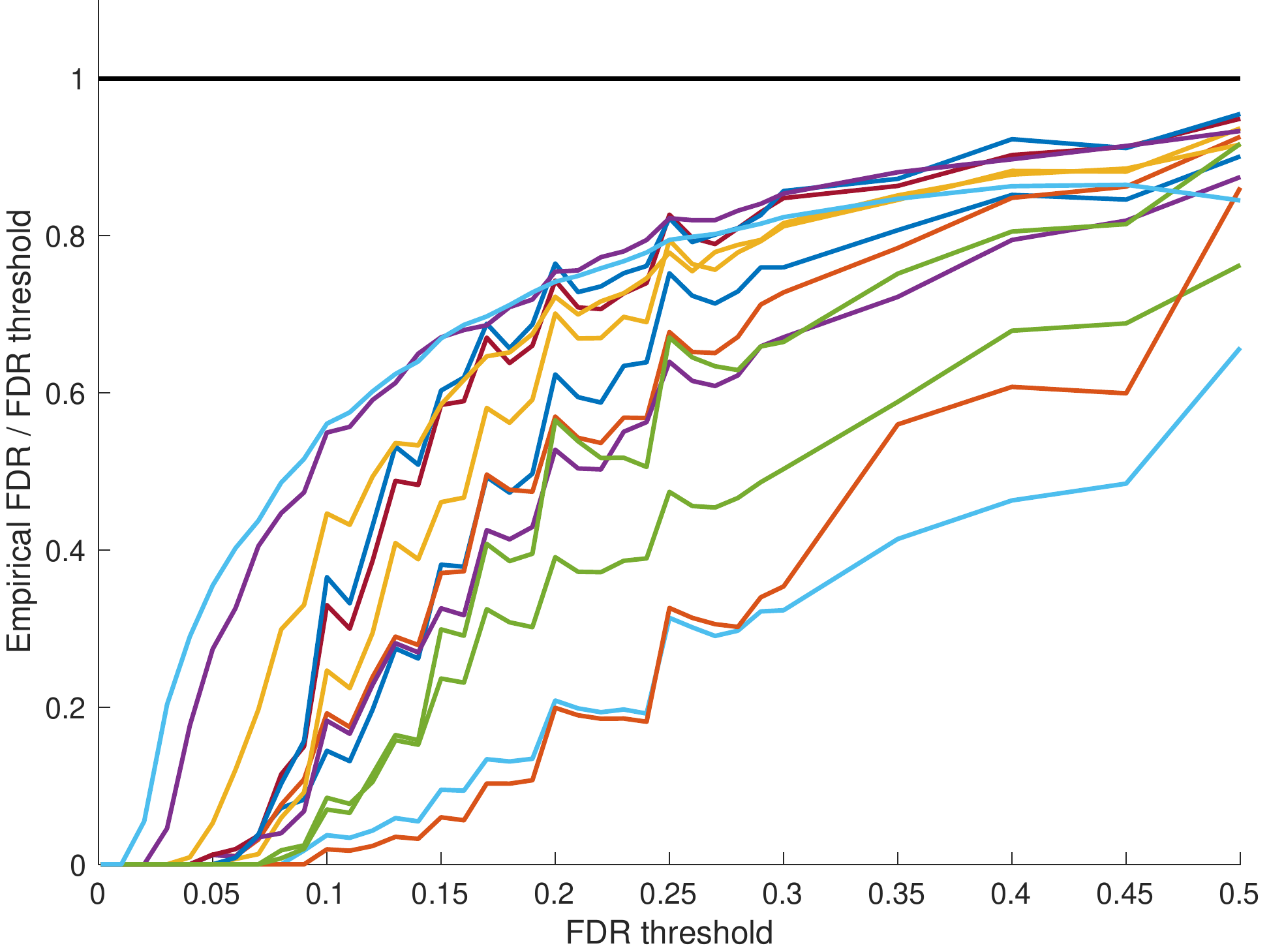}\tabularnewline
	\end{tabular}\caption{\textbf{Batching does not seem to compromise the finite sample FDR control. }Each of the
		panels shows the ratio of the empirical FDR to the FDR threshold of
		one method applied to multiple datasets using $b=40$ batches ($n=800,p=200$,
		\suppsec\ref{subsec:n800_p200_d1_3_b40}). The graphs show that
		in all the cases the methods seem to essentially control the FDR.
		(D) knockoff+' control of the FDR is guaranteed in this setting where $n\ge2p$. \label{fig:supp_batching_power_1vs40_FDR}}
\end{figure}

\begin{figure}
	\centering %
	\begin{tabular}{ll}
		A. Max vs.~knockoff+ & B. Empirical FDR (max)\tabularnewline
		\includegraphics[width=3in]{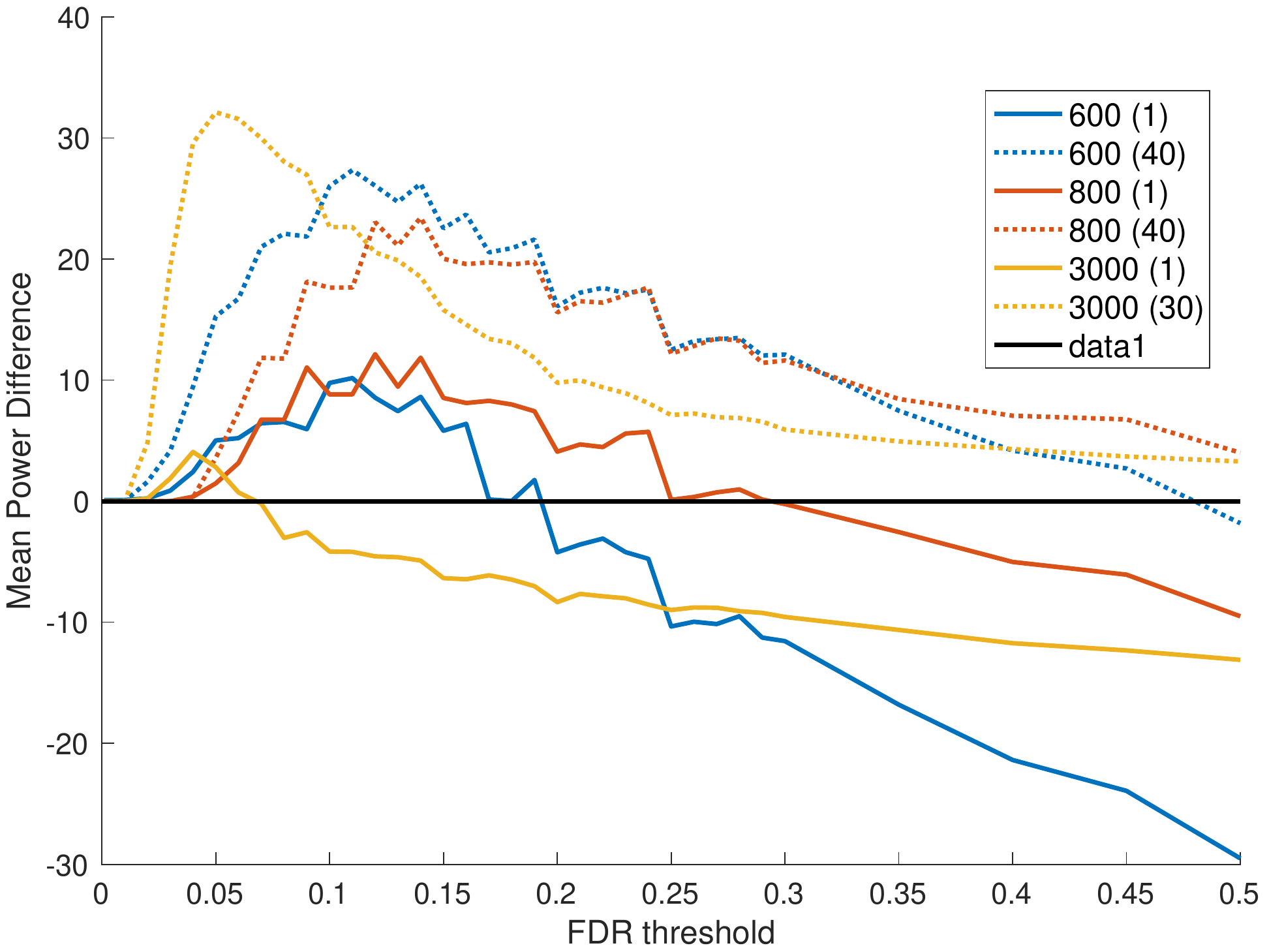} & \includegraphics[width=3in]{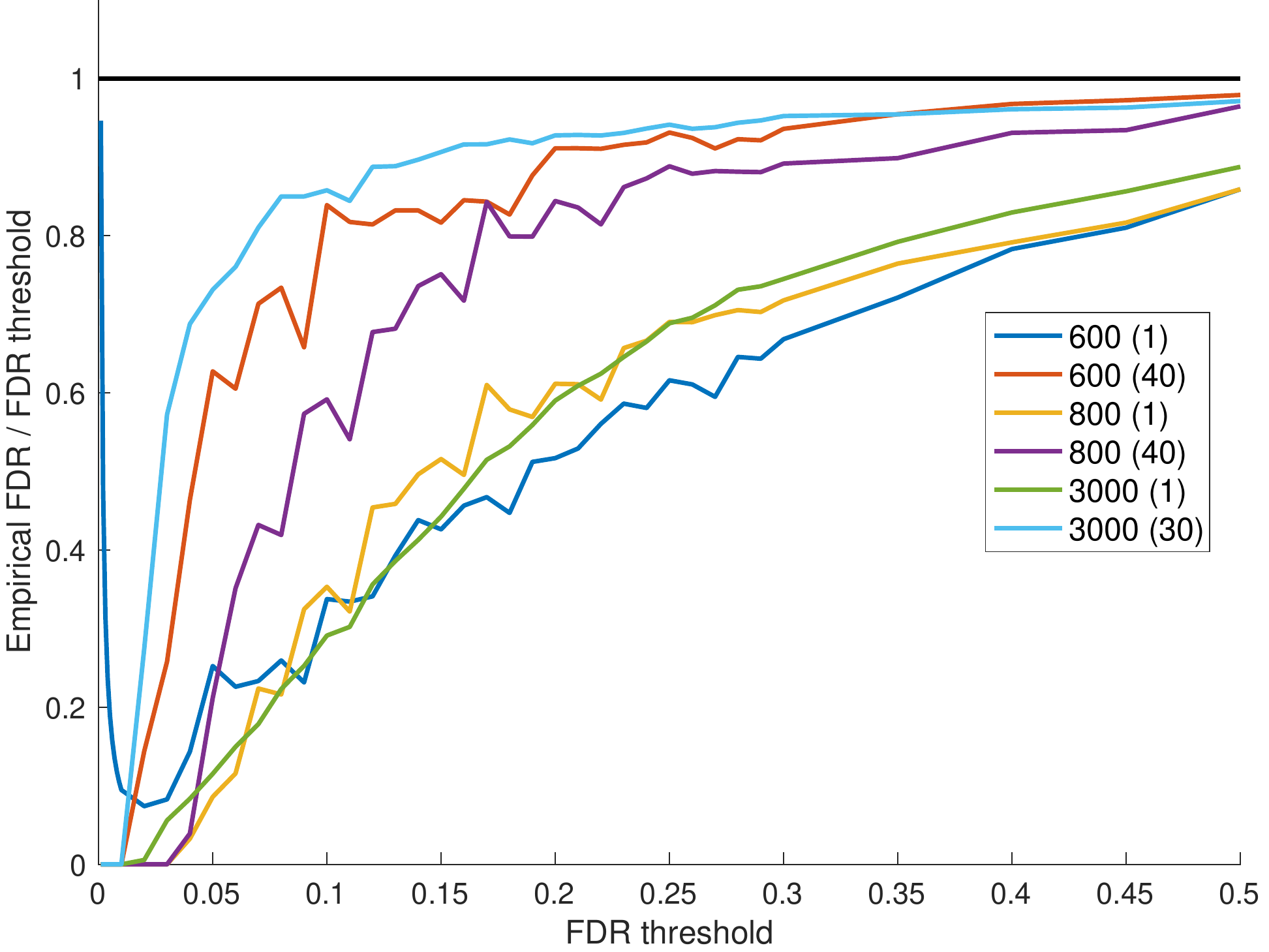}\tabularnewline
		C. Mirror vs.~knockoff+ & D. Empirical FDR (mirror)\tabularnewline
		\includegraphics[width=3in]{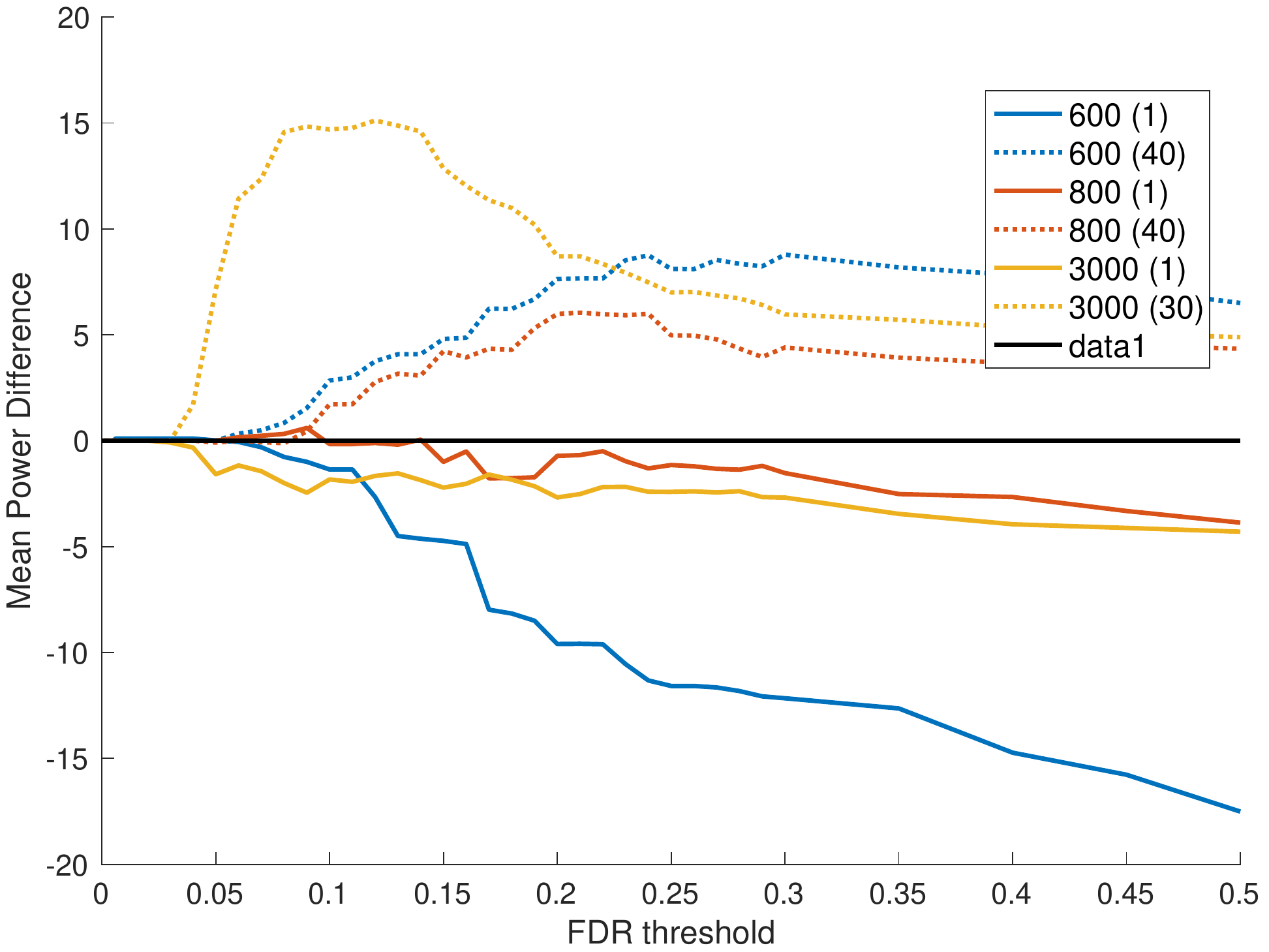} & \includegraphics[width=3in]{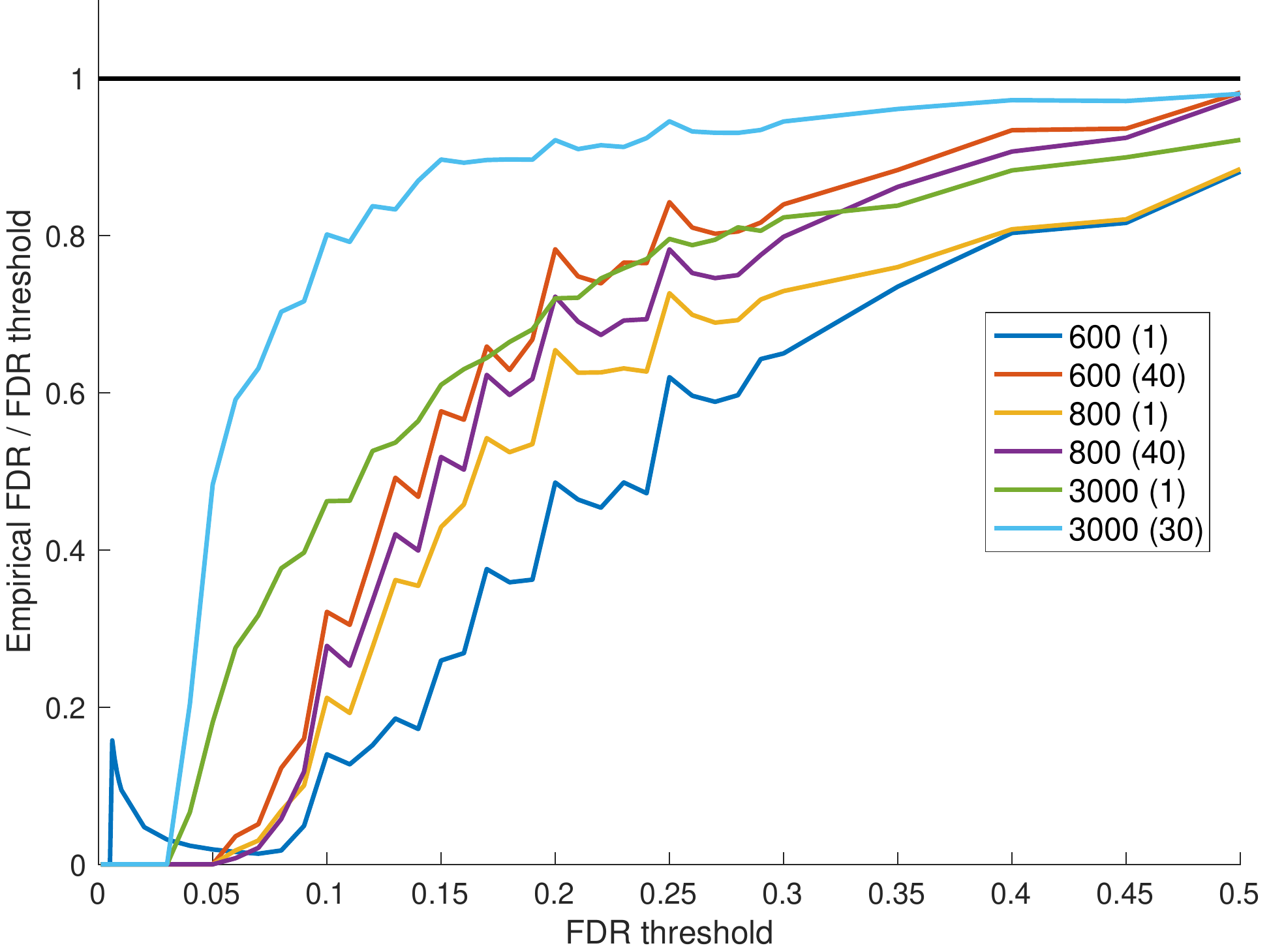}\tabularnewline
		E. Batched-knockoff+ vs.~knockoff+ & F. Empirical FDR (batched-knockoff+)\tabularnewline
		\includegraphics[width=3in]{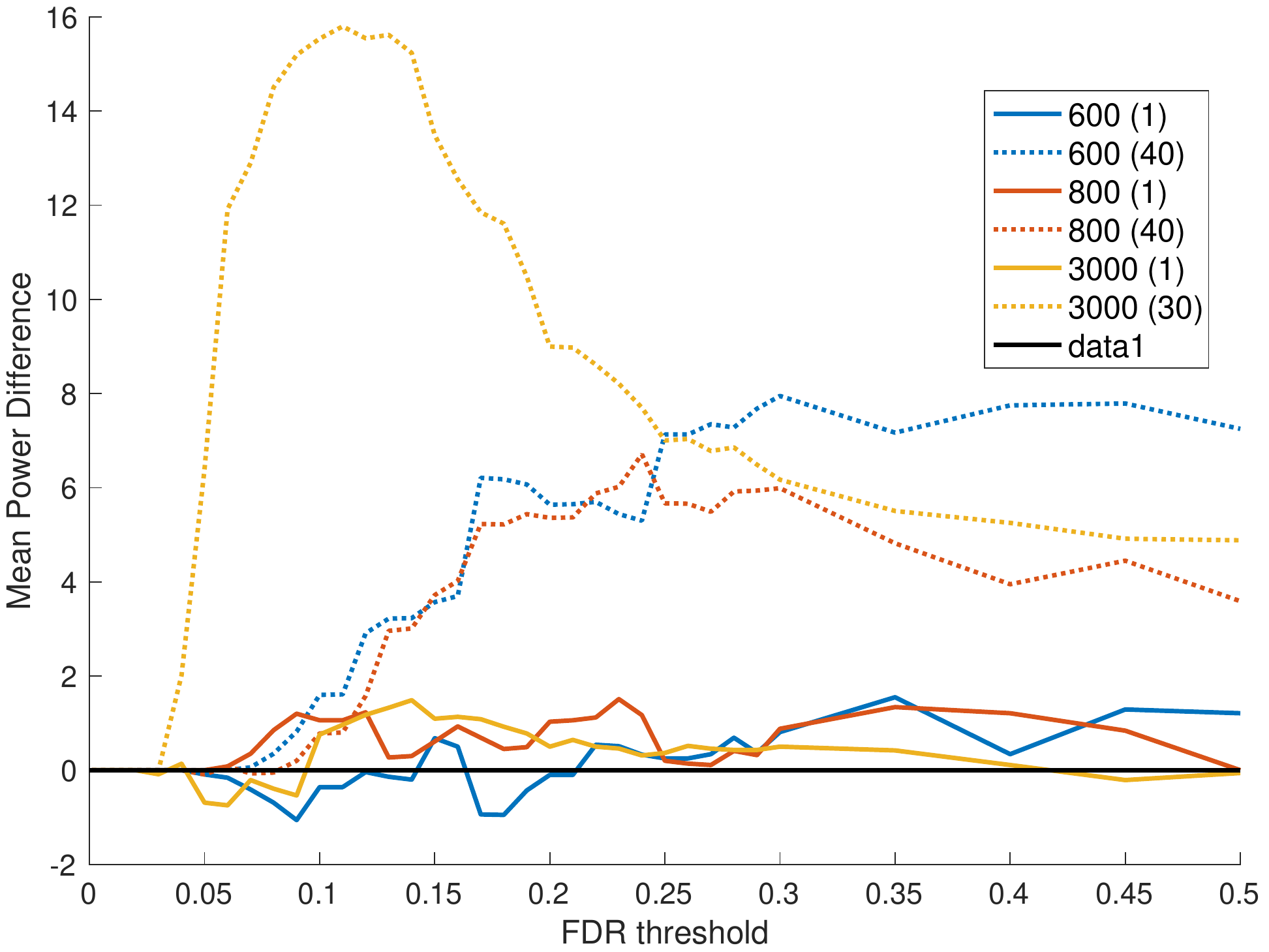} & \includegraphics[width=3in]{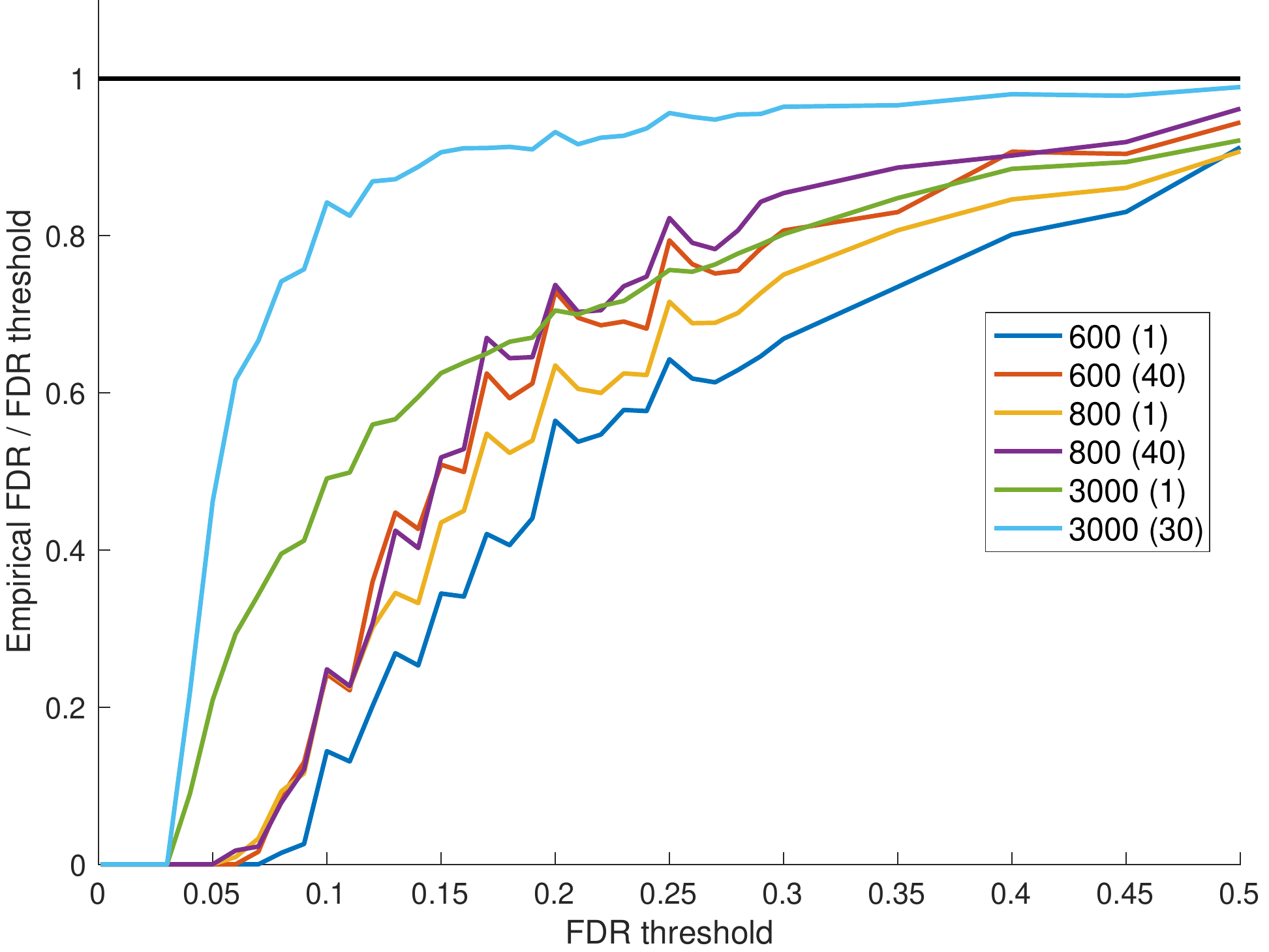}\tabularnewline
	\end{tabular}\caption{\textbf{When the batching effect is more pronounced. }Each of the left column panels
		shows the difference in the power of one method vs.~knockoff+ applied
		using either $b=1$ or $b>1$ batches: $b=30$ for $n=3000,p=1000,K=30,A=3.5,d\in\left\{ 1,3\right\}$,
		and $b=40$ for $n=800,p=200,K=10,A=3.0,d\in\left\{ 1,3\right\}$,
		and for $n=600,p=200,K=10,A=2.8,d\in\left\{ 1,11\right\}$, $\rho=0$
		in all cases, (\suppsec\ref{subsec:Eclectic-bacthing-example})
		The design of the experiment involved drawing a new set of 1K datasets
		for each value of $b$ but knockoff+ and the method to which it is
		compared were applied to the \emph{same} dataset each time. Negative
		values indicate knockoff+ is more powerful. Each right column panel
		uses the same datasets as the panel to its left to show the ratio
		of the empirical FDR of the considered method to the FDR threshold.
		The ratios are all below 1 indicating the methods seem to control
		the FDR in all these cases.\label{fig:batching_larger_sets}}
\end{figure}

%

\begin{figure}[h]
	\centering %
	\begin{tabular}{ll}
		A. LBM vs.~multi-knockoff ($d\ge7$) & B. LBM vs.~multi-knockoff ($d\le3$)\tabularnewline
		\protect\includegraphics[width=3in]{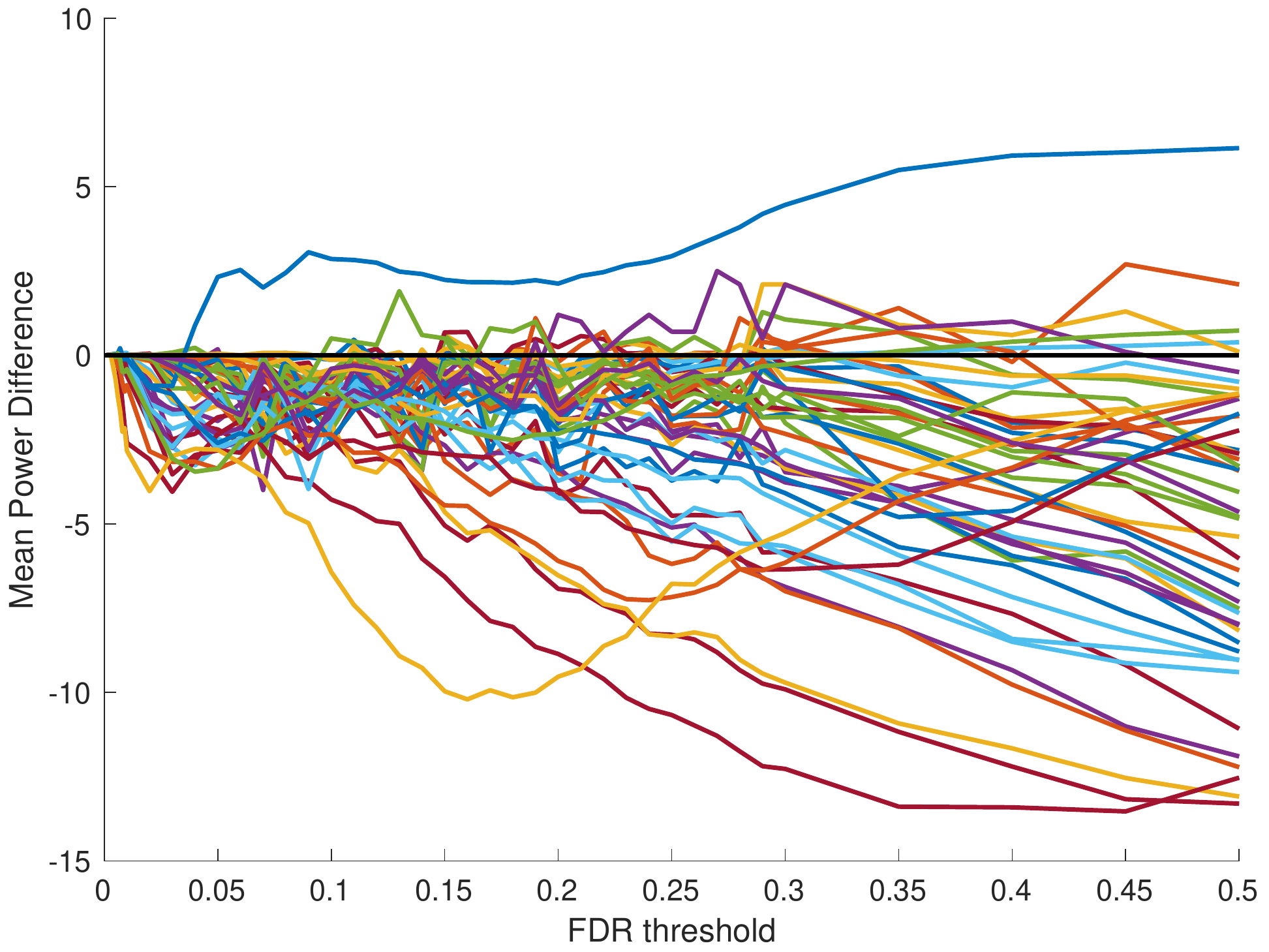} & \protect\includegraphics[width=3in]{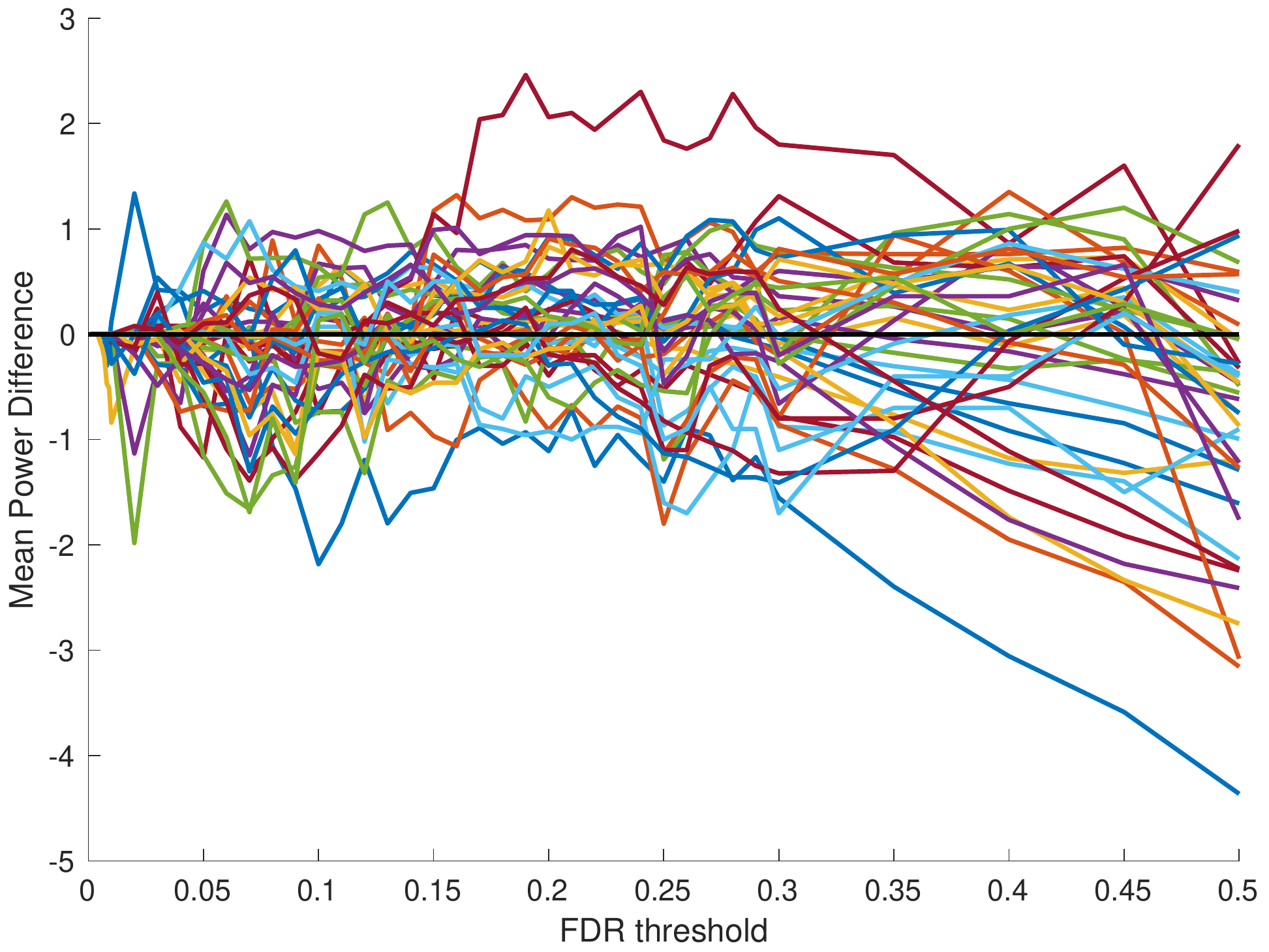}\tabularnewline
		C. Empirical FDR (multi-knockoff) & D. Empirical FDR (LBM)\tabularnewline
		\protect\includegraphics[width=3in]{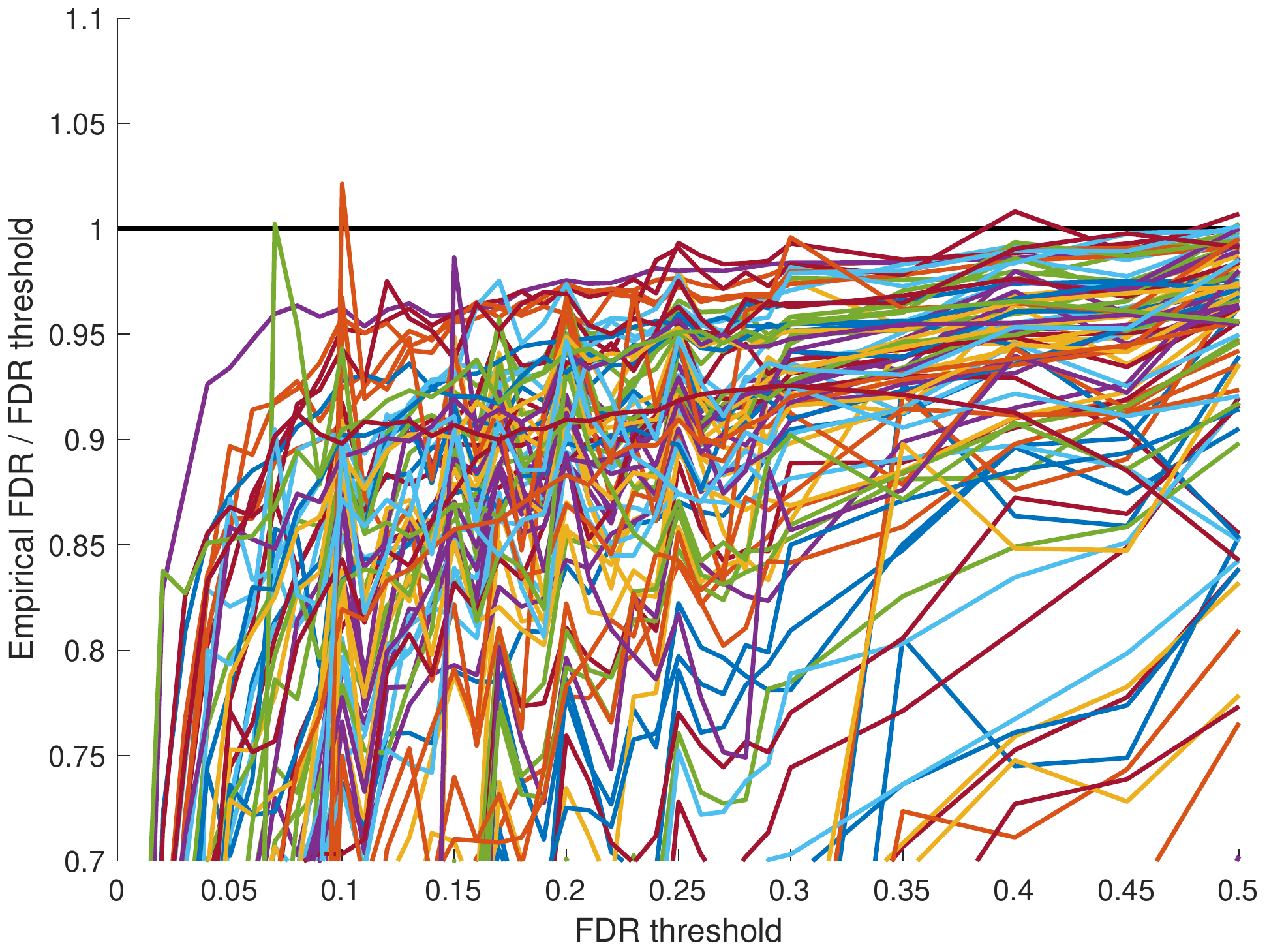} & \protect\includegraphics[width=3in]{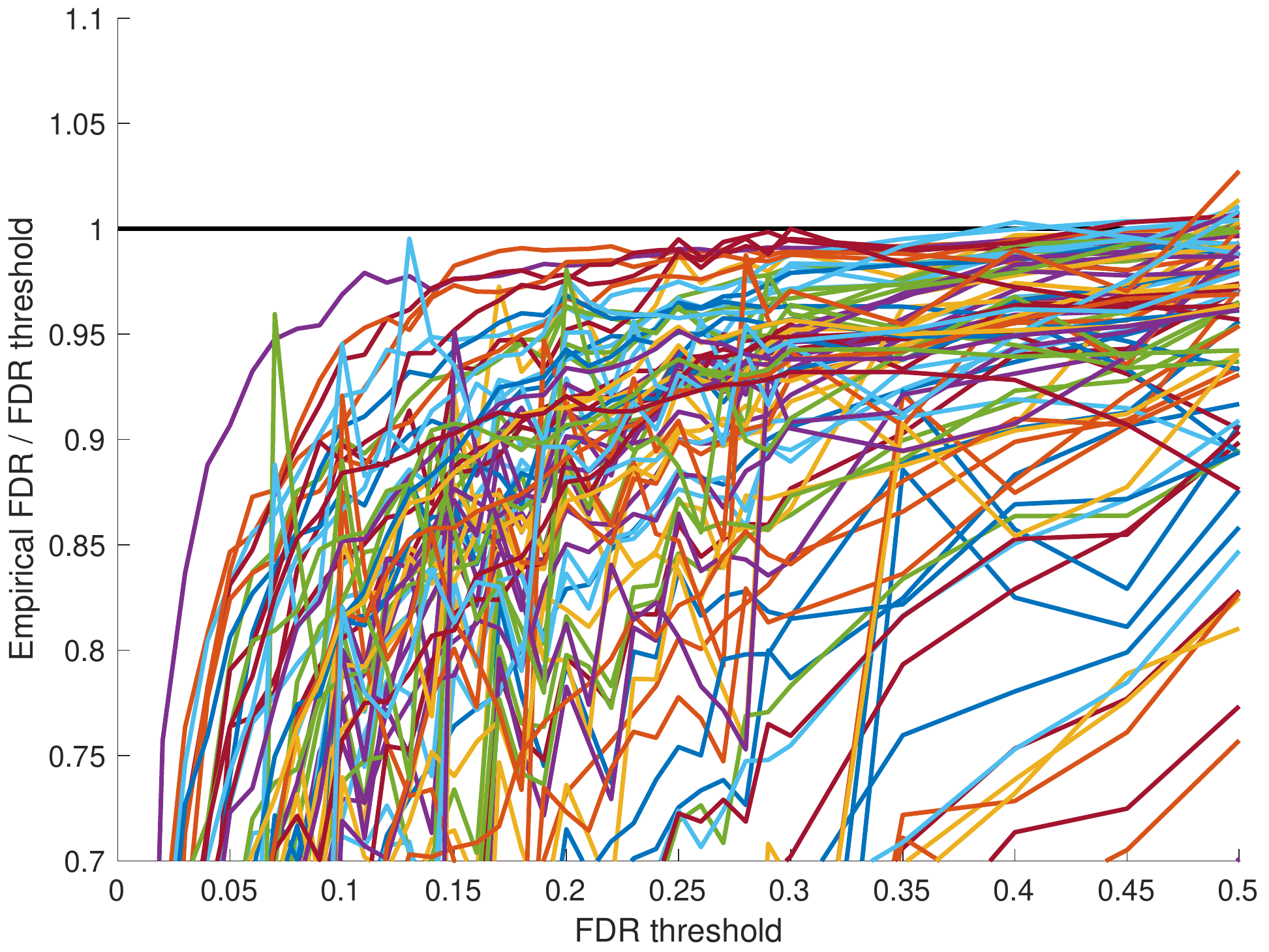}\tabularnewline
	\end{tabular}\protect\caption{\textbf{Multi-knockoff vs.~LBM. }Comparison of the two resampling-based methods
		for selecting the ($c,\lam$) tuning parameters.
		Both methods were applied to all the datasets in our combined collection
		of experiments, which spans a wide range of parameter values and
		is described in \suppsec\ref{subsec:Combined-dataset}. (A) Power
		difference between LBM and multi-knockoff (negative numbers mean multi-knockoff
		is better). Only experiments with $d\ge7$ knockoffs are shown. 
		The one example where LBM is moderately better than multi-knockoff
		(cyan colored) corresponds to a realistically borderline 80\% proportion
		of features in the model: $K=160$ and $p=200$ ($n=600$, $d=11$). (B)
		Same as A but with $d\le3$ (same as $d<7$ in this case).
		(C-D) Empirical FDR on the entire ``combined'' set.
		 \label{fig:mKO-vs-LBM}}
\end{figure}

\begin{figure}
	\centering %
	\begin{tabular}{ll}
		A. Knockoff+ vs.~multi-knockoff & B. Batched-knockoff+ vs.~multi-knockoff\tabularnewline
		\protect\includegraphics[width=3in]{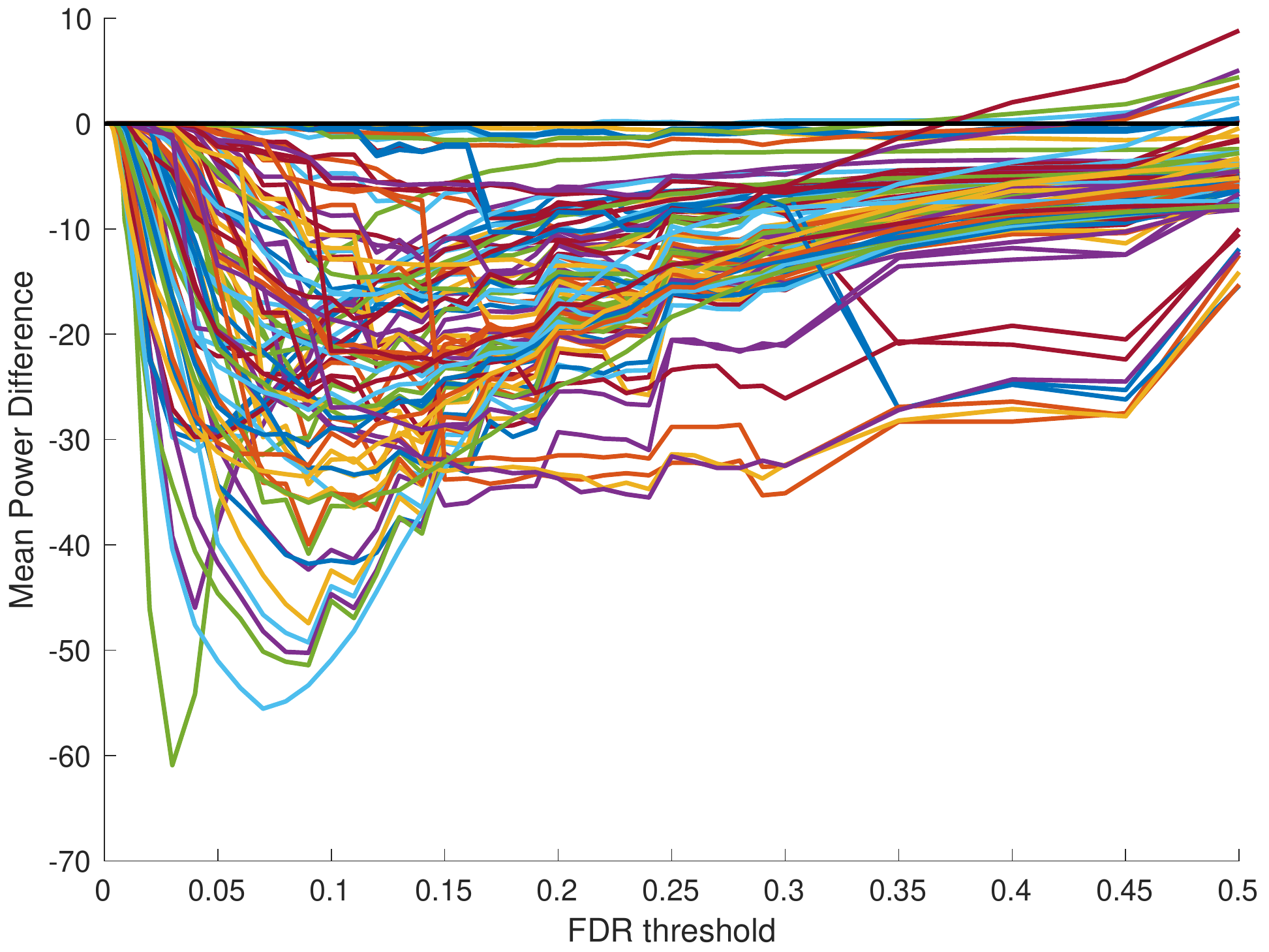} & \protect\includegraphics[width=3in]{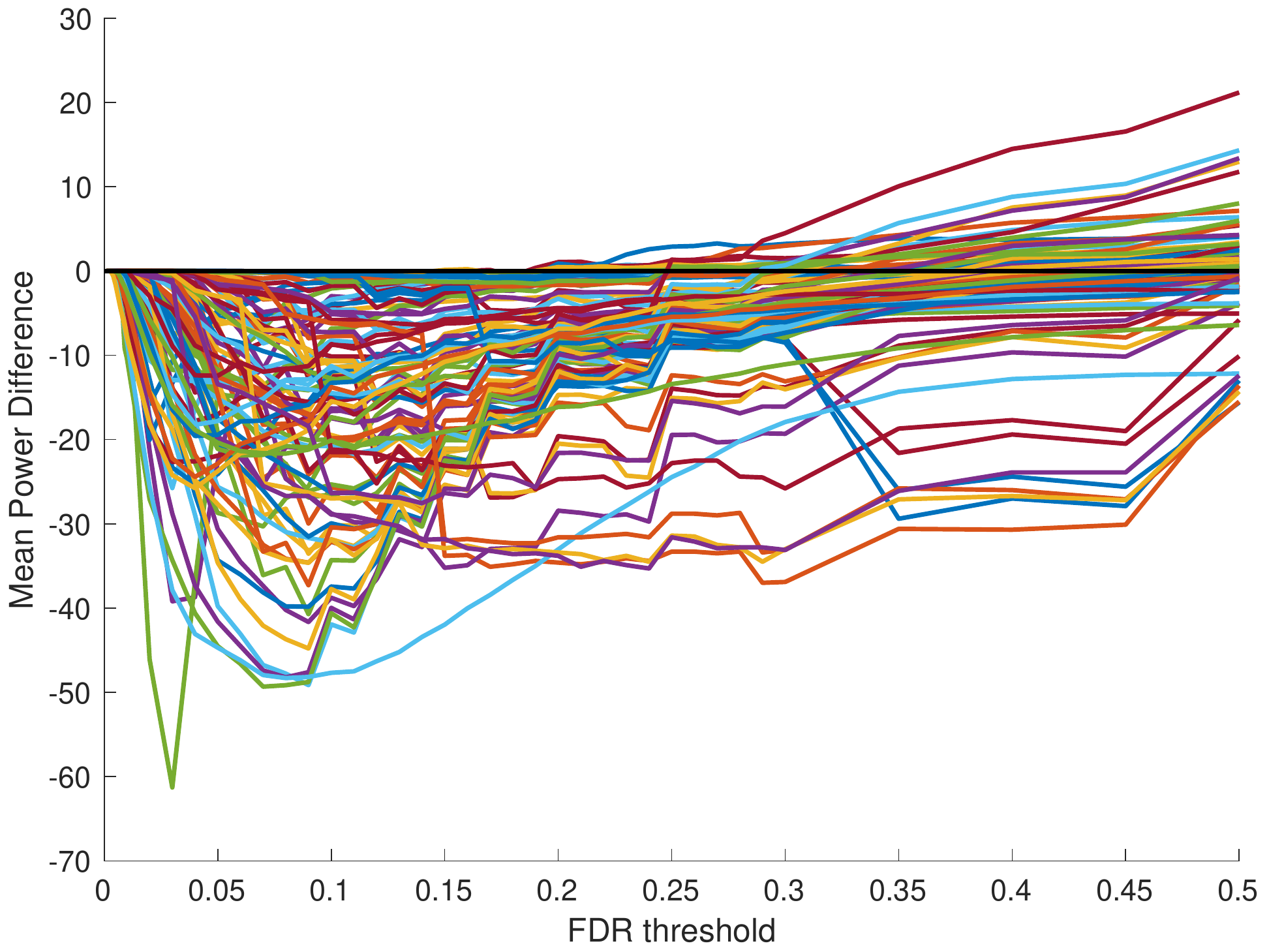}\tabularnewline
		C. Mirror vs.~multi-knockoff & D. Max vs.~multi-knockoff\tabularnewline
		\protect\includegraphics[width=3in]{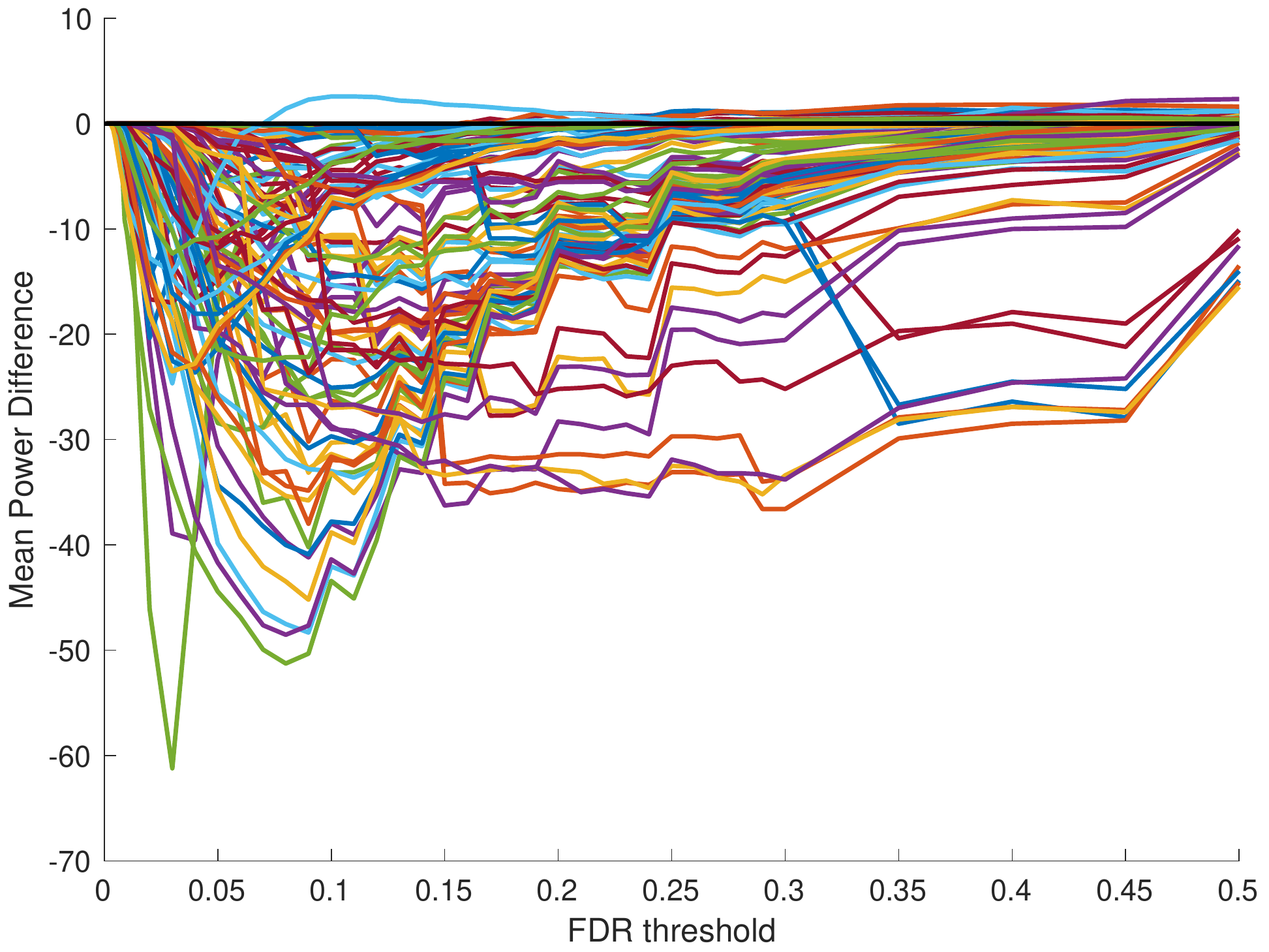} & \protect\includegraphics[width=3in]{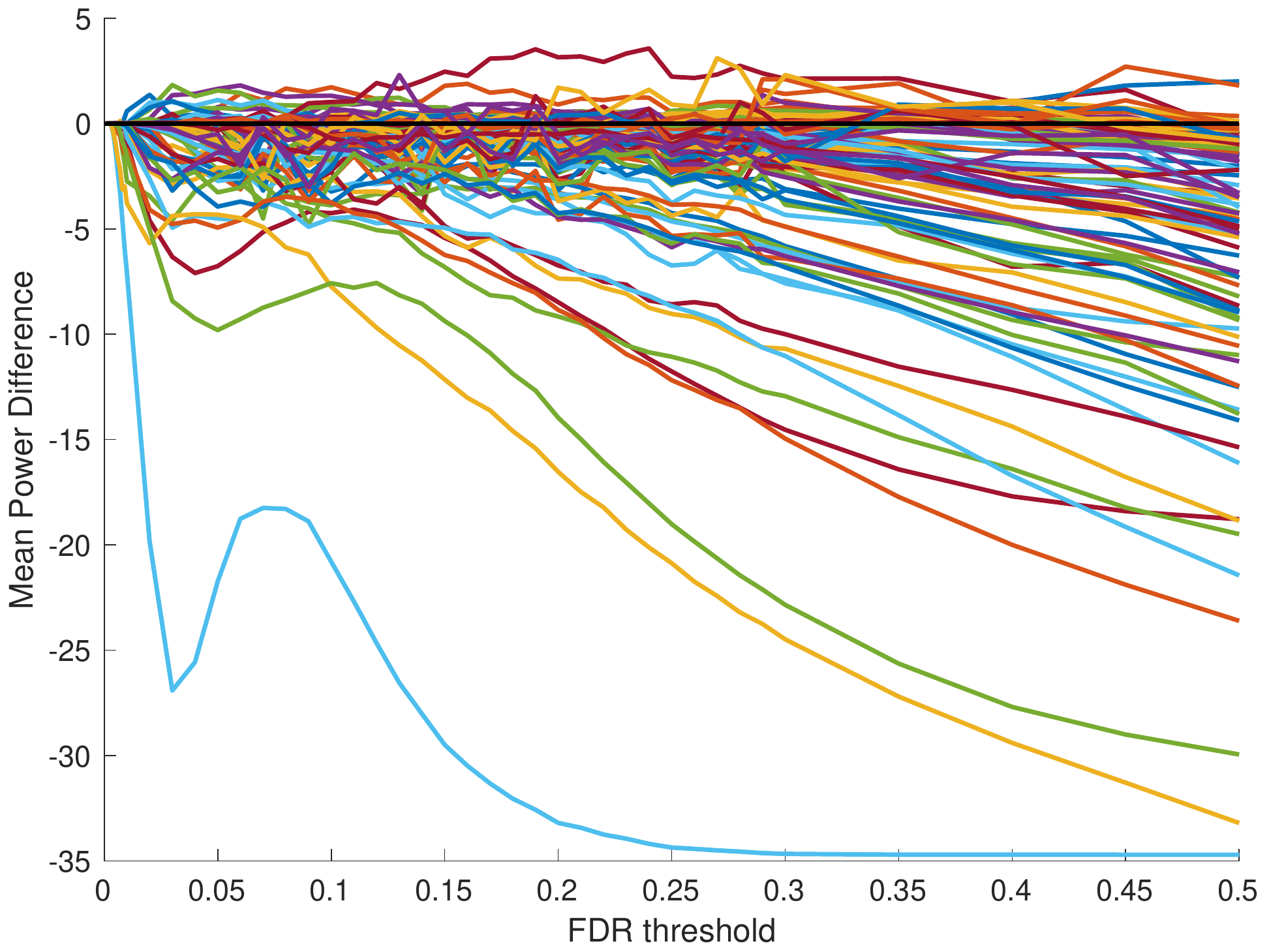}\tabularnewline
		E. LBM vs.~multi-knockoff & F. Multi-knockoff-select vs.~multi-knockoff\tabularnewline
		\protect\includegraphics[width=3in]{figures/n3000_n800_n600_r2376_r2398_power_mKO_vs_LBM_} & \protect\includegraphics[width=3in]{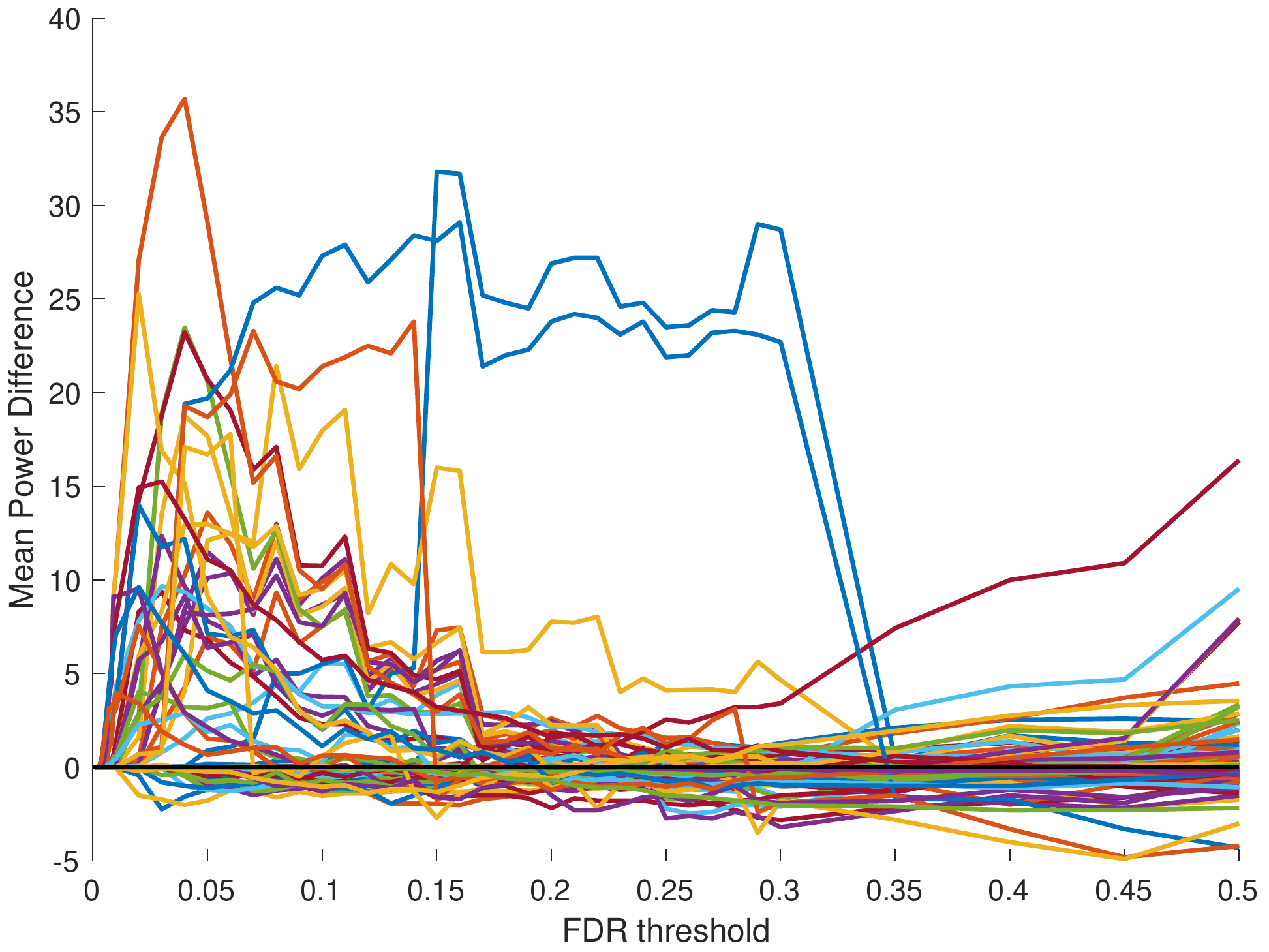}\tabularnewline
	\end{tabular}\protect\caption{\textbf{Power difference vs.~multi-knockoff.} Each panel shows the
		difference in power between one of the methods considered in this
		paper and multi-knockoff. All the methods were applied to all the datasets
		that are included the combined set described in \suppsec\ref{subsec:Combined-dataset}.
		Note that the scale of the $y$-axis varies across the panels. \label{fig:all-vs-multiKO}}
\end{figure}

%

\begin{figure}
	\centering %
	\begin{tabular}{ll}
		A. The optimal \# of knockoffs, $d$, varies (max) & B. The optimal $d$ varies (multi-knockoff)\tabularnewline
		\protect\includegraphics[width=3in]{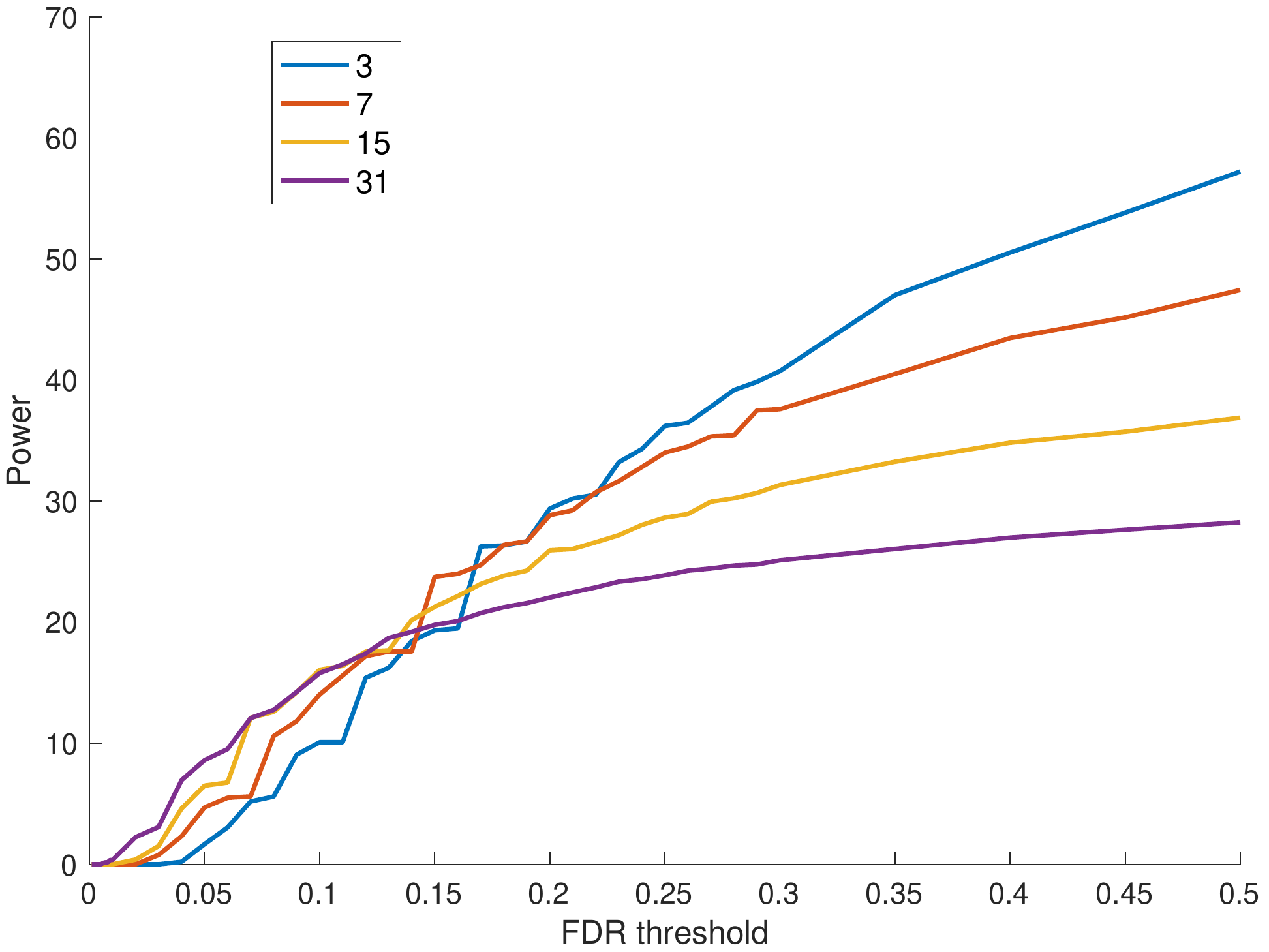} & \protect\includegraphics[width=3in]{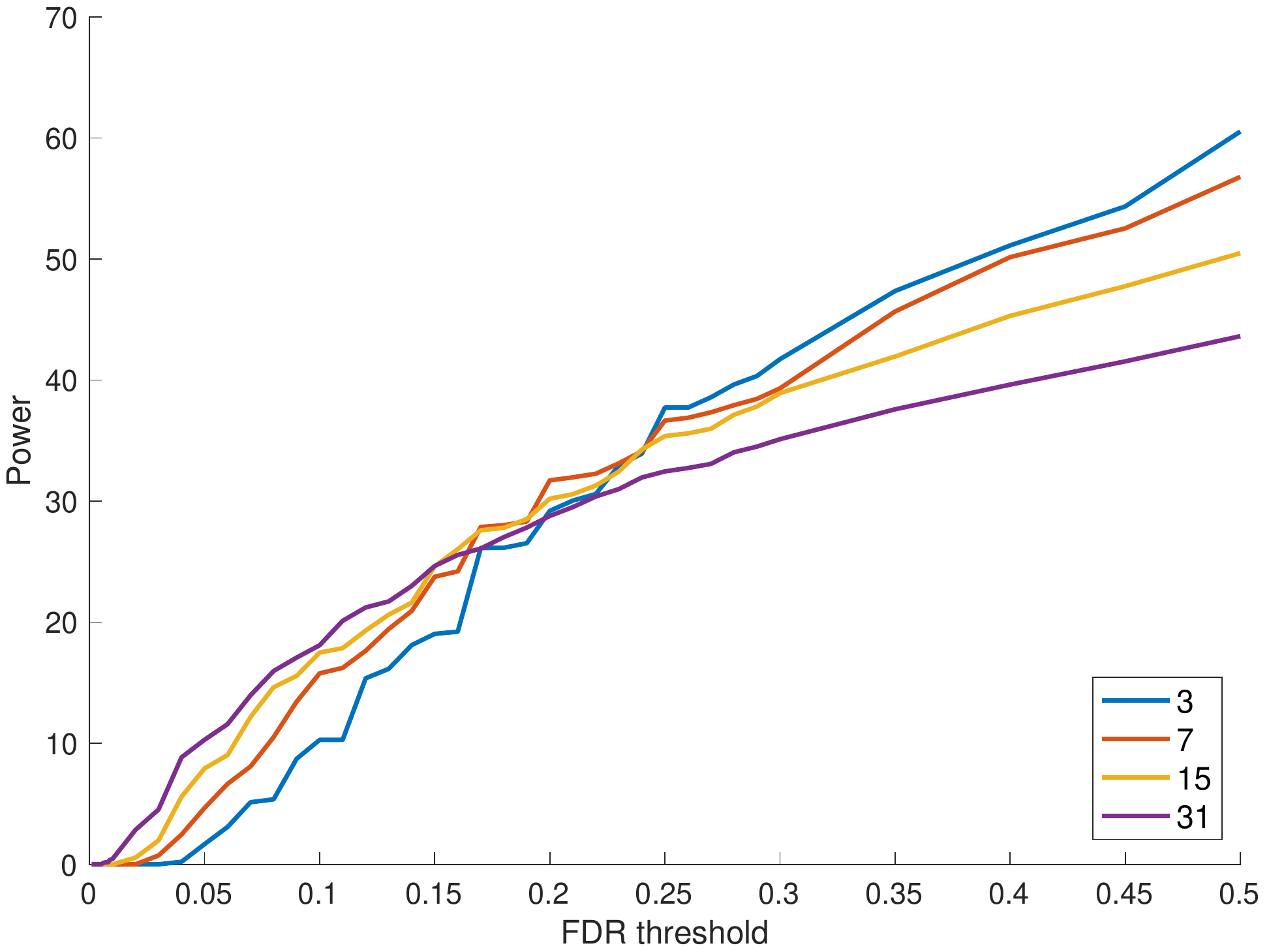}\tabularnewline
		C. Same as B but with multi-knockoff-select & D. multi-knockoff vs. multi-knockoff-select\tabularnewline
		\protect\includegraphics[width=3in]{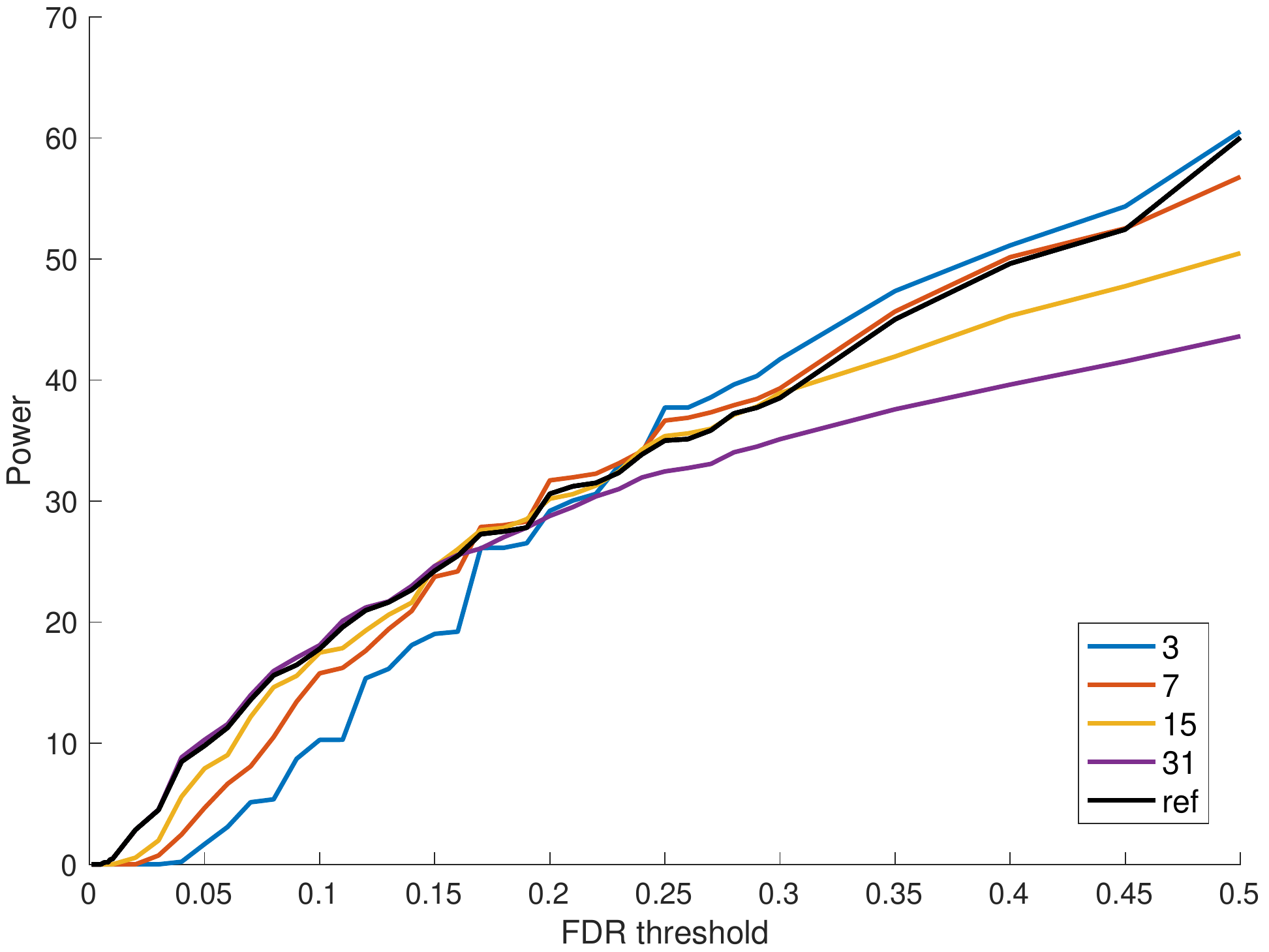} & \protect\includegraphics[width=3in]{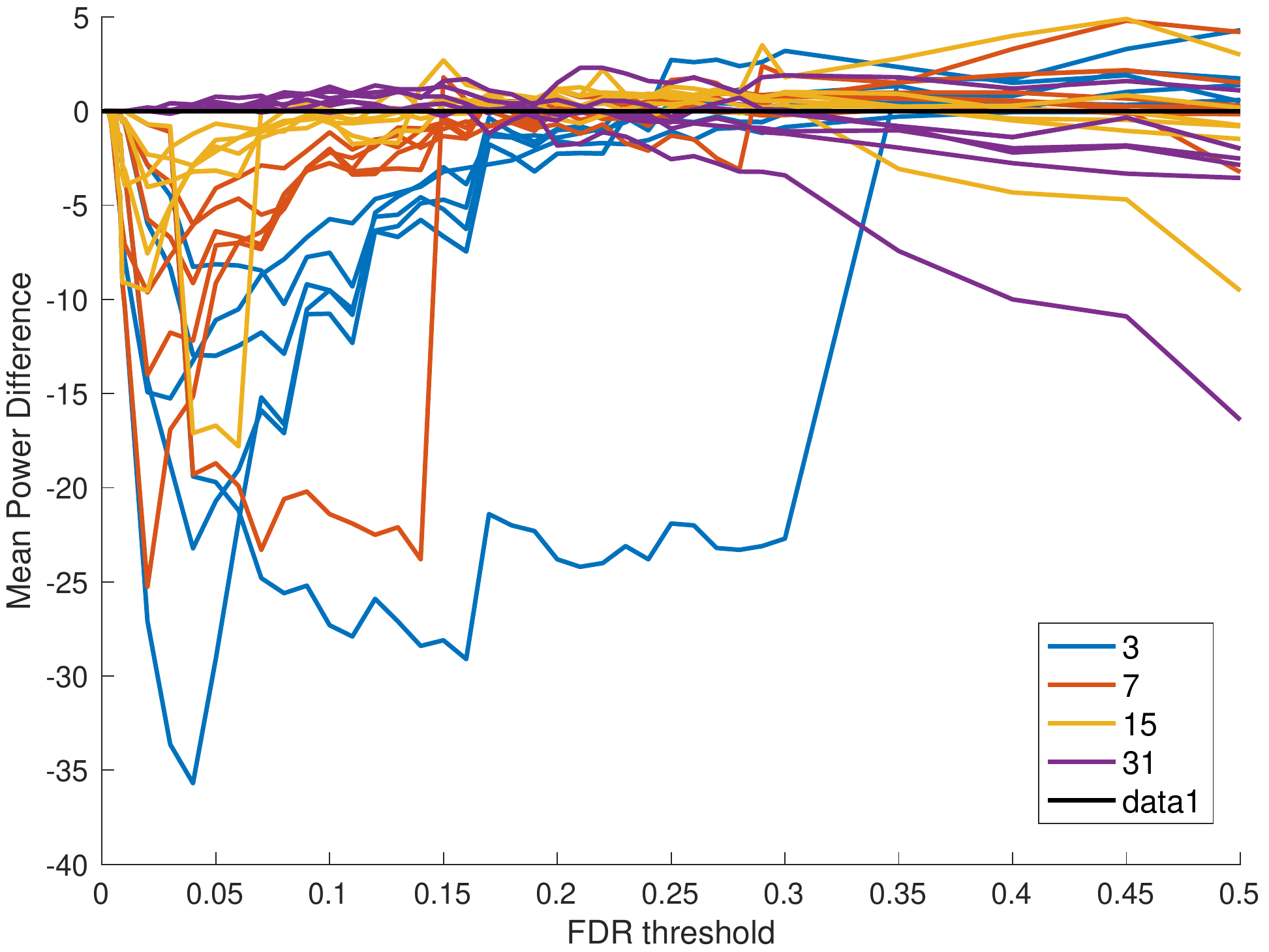}\tabularnewline
		E. knockoff+ vs. multi-knockoff-select & F. Empirical FDR (multi-knockoff-select)\tabularnewline
		\protect\includegraphics[width=3in]{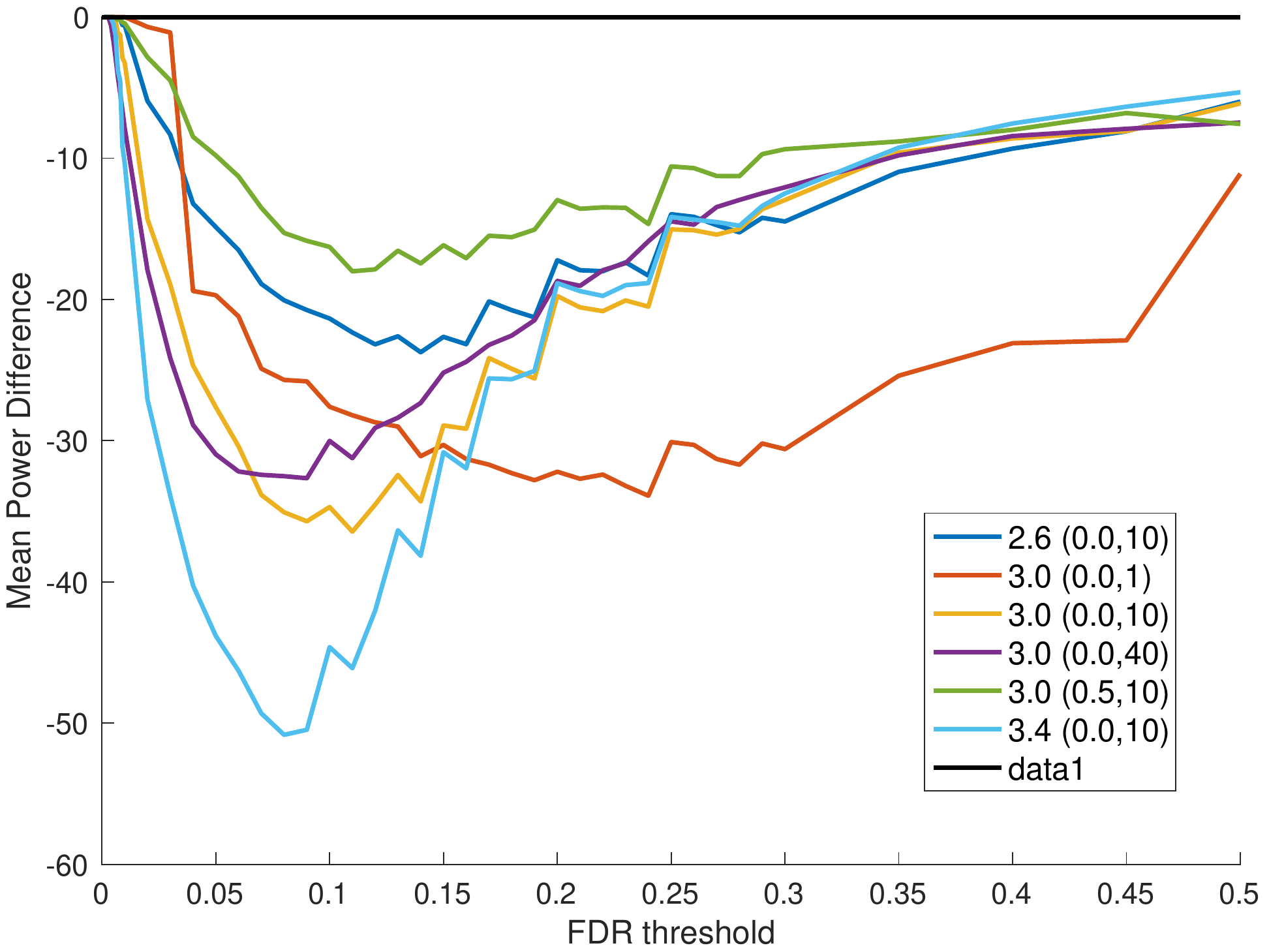} & \protect\includegraphics[width=3in]{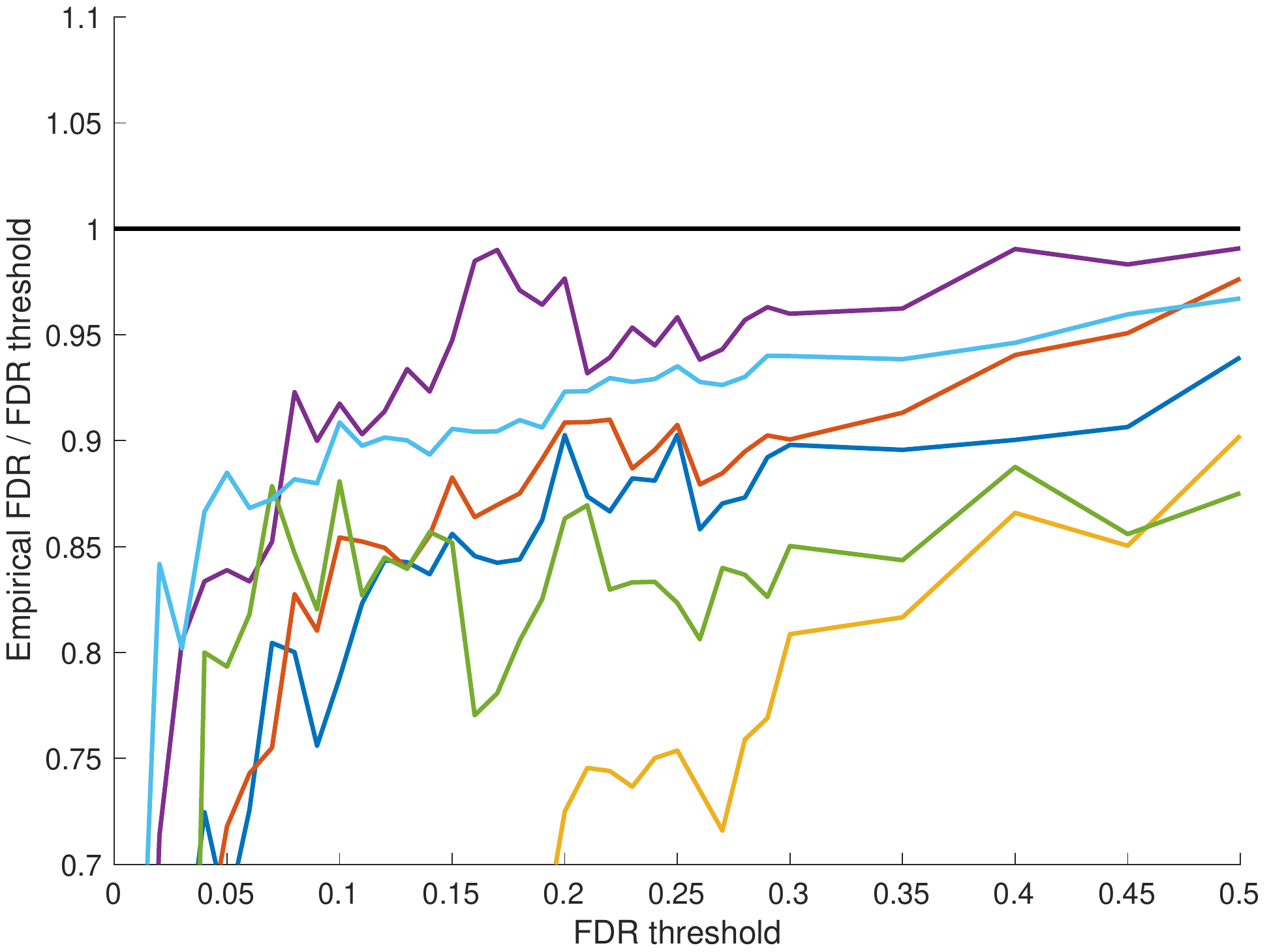}\tabularnewline
	\end{tabular}\protect\caption{\textbf{Varying the number of knockoffs} All the plots were
		created using the set of experiments defined in \suppsec\ref{subsec:n600_p200_d1_3_7_15_31_b40}.
		When applying any multiple-knockoff selection method to a randomly
		drawn dataset we used the increasing sequence of $d=3,7,15,31$ knockoffs so
		we can examine how the method's power varies with $d$. (A-B) The
		optimal value of $d$ varies for the max method and multi-knockoff
		in this example where $n=600$, $p=200$, $K=10$ with $A=3.0$ and
		0 feature correlation $\Theta_{0}=I_{p}$. (C) multi-knockoff-select  (black ``ref''
		curve) seems to do a good job at tracking the near-optimal value of $d$
		(same data as in panel B). (D) The difference in average power of
		multi-knockoff vs. multi-knockoff-select using all six datasets in
		\suppsec\ref{subsec:n600_p200_d1_3_7_15_31_b40}. (E) Same as in
		panel D but now the comparison is with knockoff+, which is evidently
		uniformly weaker than multi-knockoff-select. (F) Empirical evidence
		that multi-knockoff-select controls the FDR when applied to these
		datasets  (same legend as panel E). \label{fig:multiple_nKO}}
\end{figure}

\begin{figure}
	\centering %
	\begin{tabular}{ll}
		A. Knockoff+ vs.~multi-knockoff-select & B. Batched-knockoff+ vs.~multi-knockoff-select\tabularnewline
		\protect\includegraphics[width=3in]{figures/n3000_n800_n600_r2376_r2398_power_mKOsel_vs_KO_} & \protect\includegraphics[width=3in]{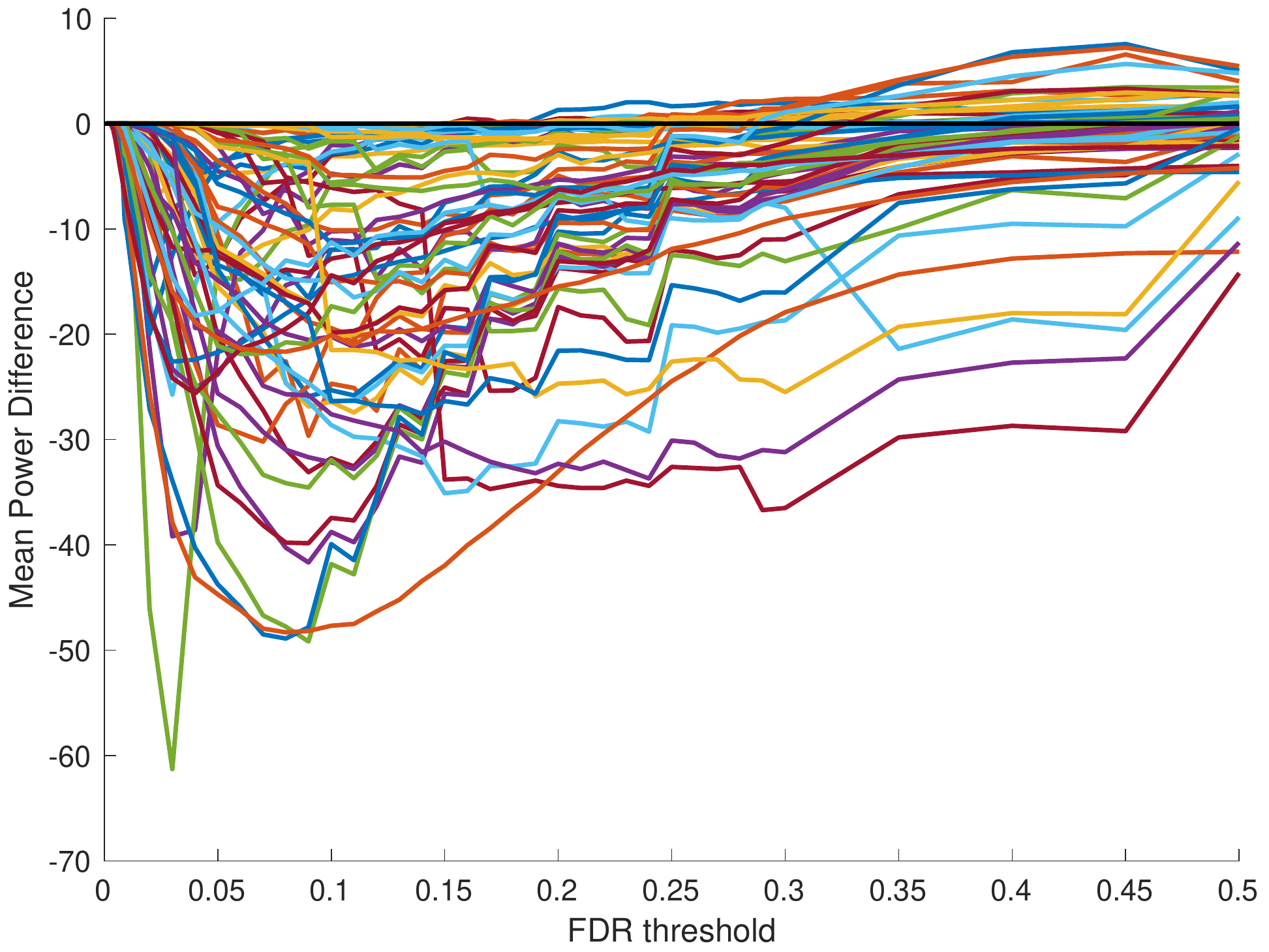}\tabularnewline
		C. Mirror vs.~multi-knockoff-select & D. Max vs.~multi-knockoff-select\tabularnewline
		\protect\includegraphics[width=3in]{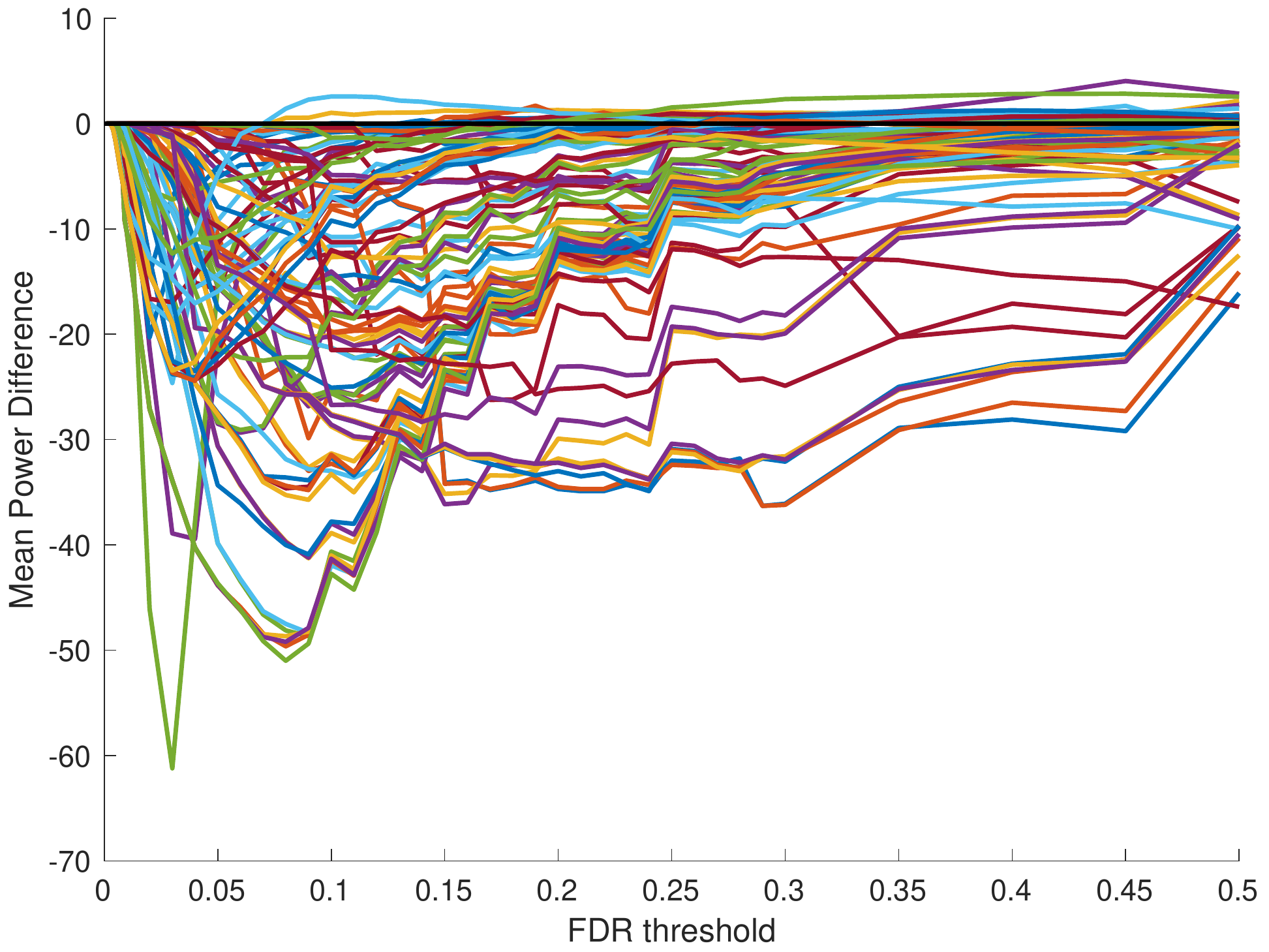} & \protect\includegraphics[width=3in]{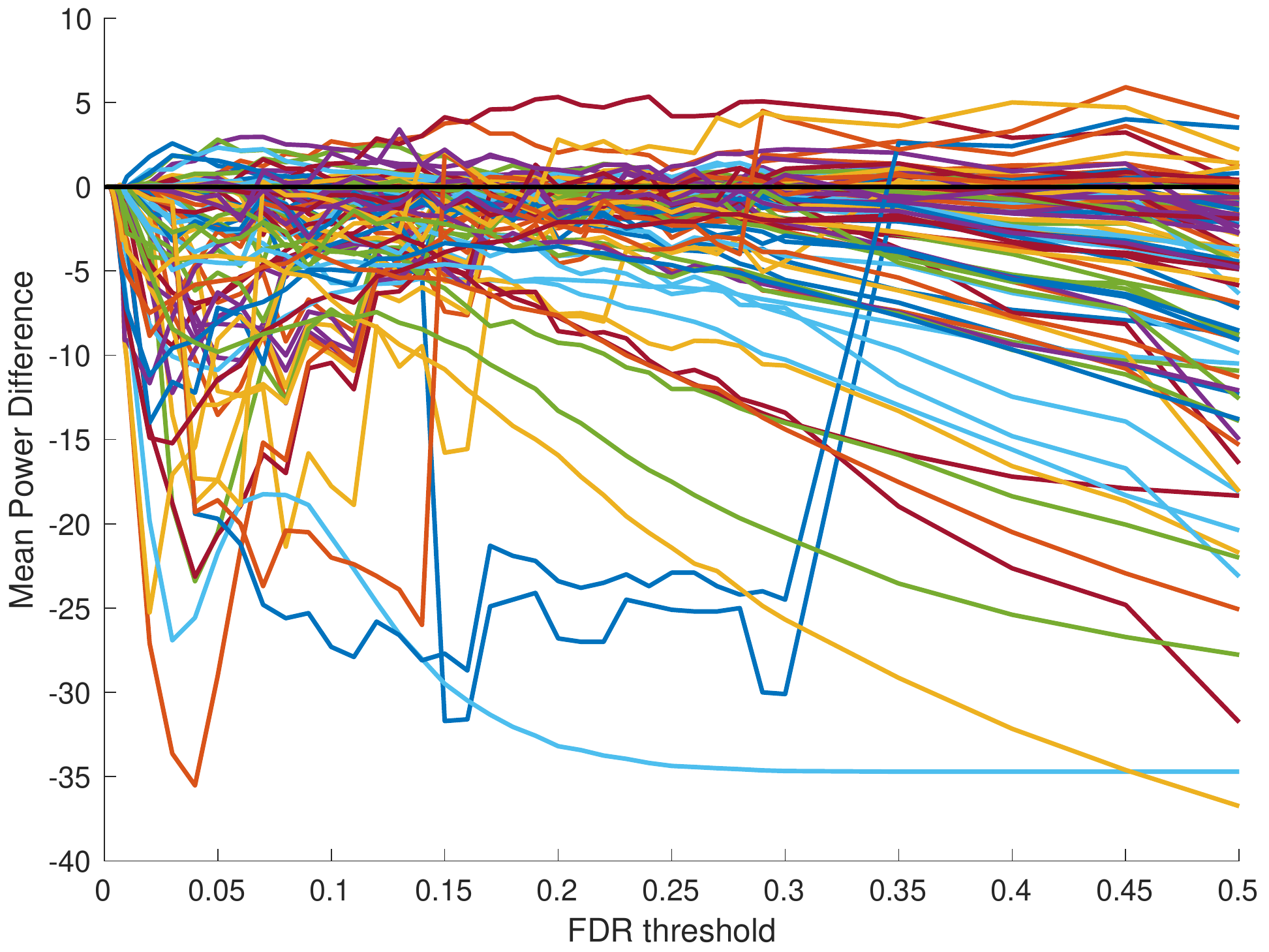}\tabularnewline
		E. LBM vs.~multi-knockoff-select & F. Multi-knockoff vs.~multi-knockoff-select\tabularnewline
		\protect\includegraphics[width=3in]{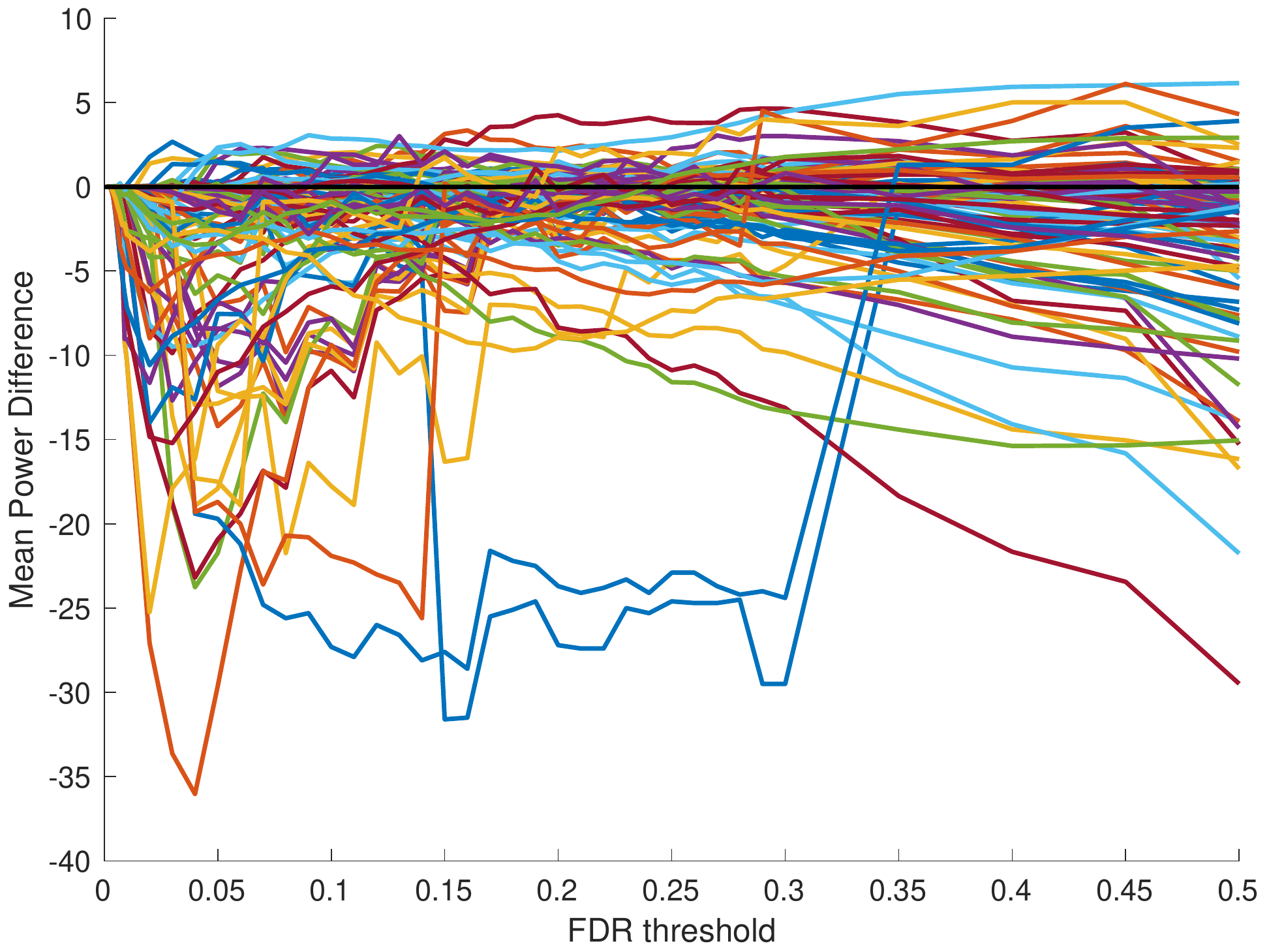} & \protect\includegraphics[width=3in]{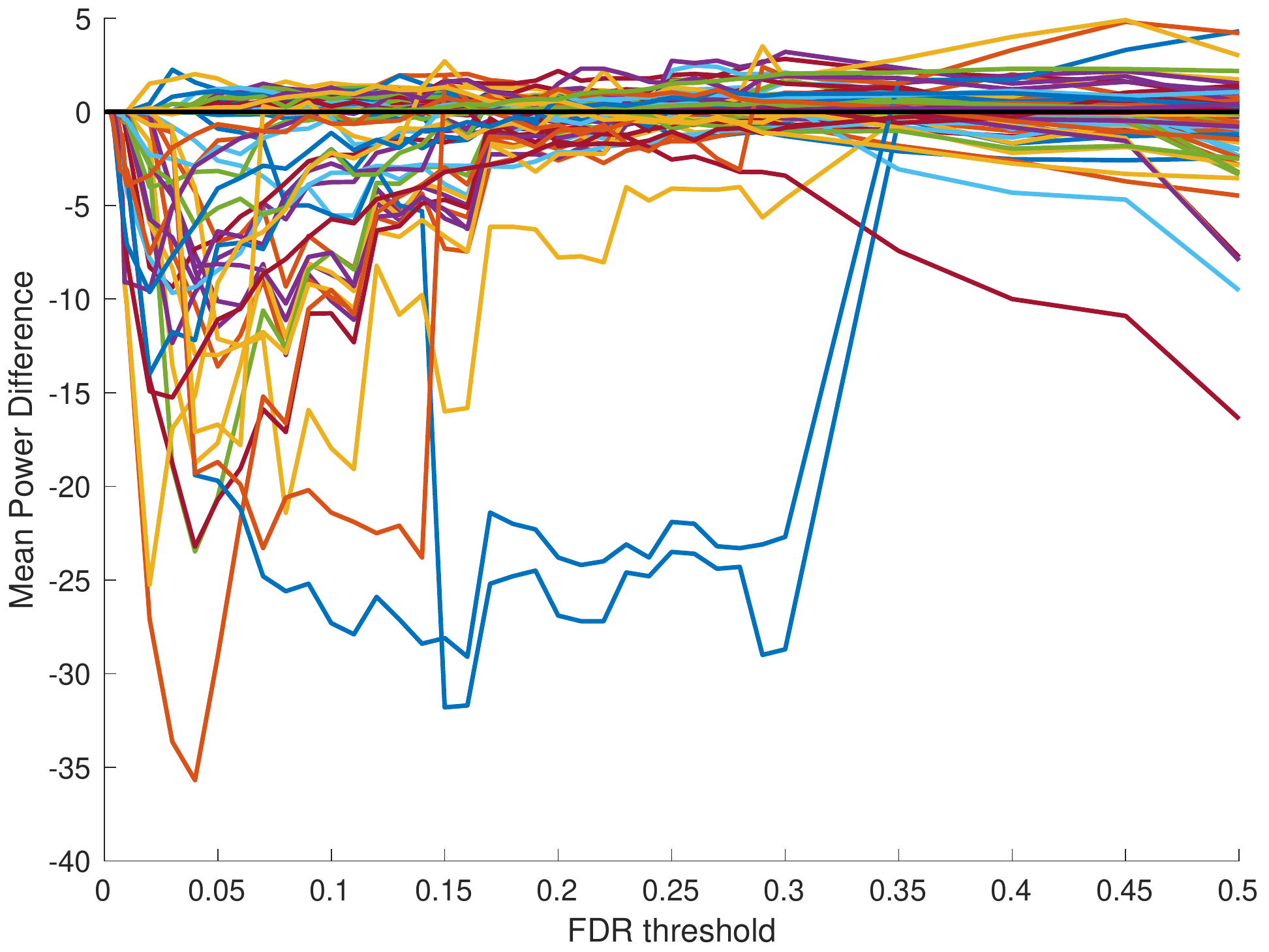}\tabularnewline
	\end{tabular}\protect\caption{\textbf{Power difference vs.~Multi-knockoff-select.} Each panel shows
		the difference in power between one of the methods considered in this
		paper and multi-knockoff-select. All the methods were applied to all the
		datasets that are included the combined set described in \suppsec\ref{subsec:Combined-dataset}.
		Note that the scale of the $y$-axis varies across the panels.\label{fig:all-vs-multiKO-select}}
\end{figure}

\begin{figure}
	\centering %
	\begin{tabular}{ll}
		A. Empirical FDR (multi-knockoff-select) & B. Knockoff+ vs.~multi-knockoff-select\tabularnewline
		\protect\includegraphics[width=3in]{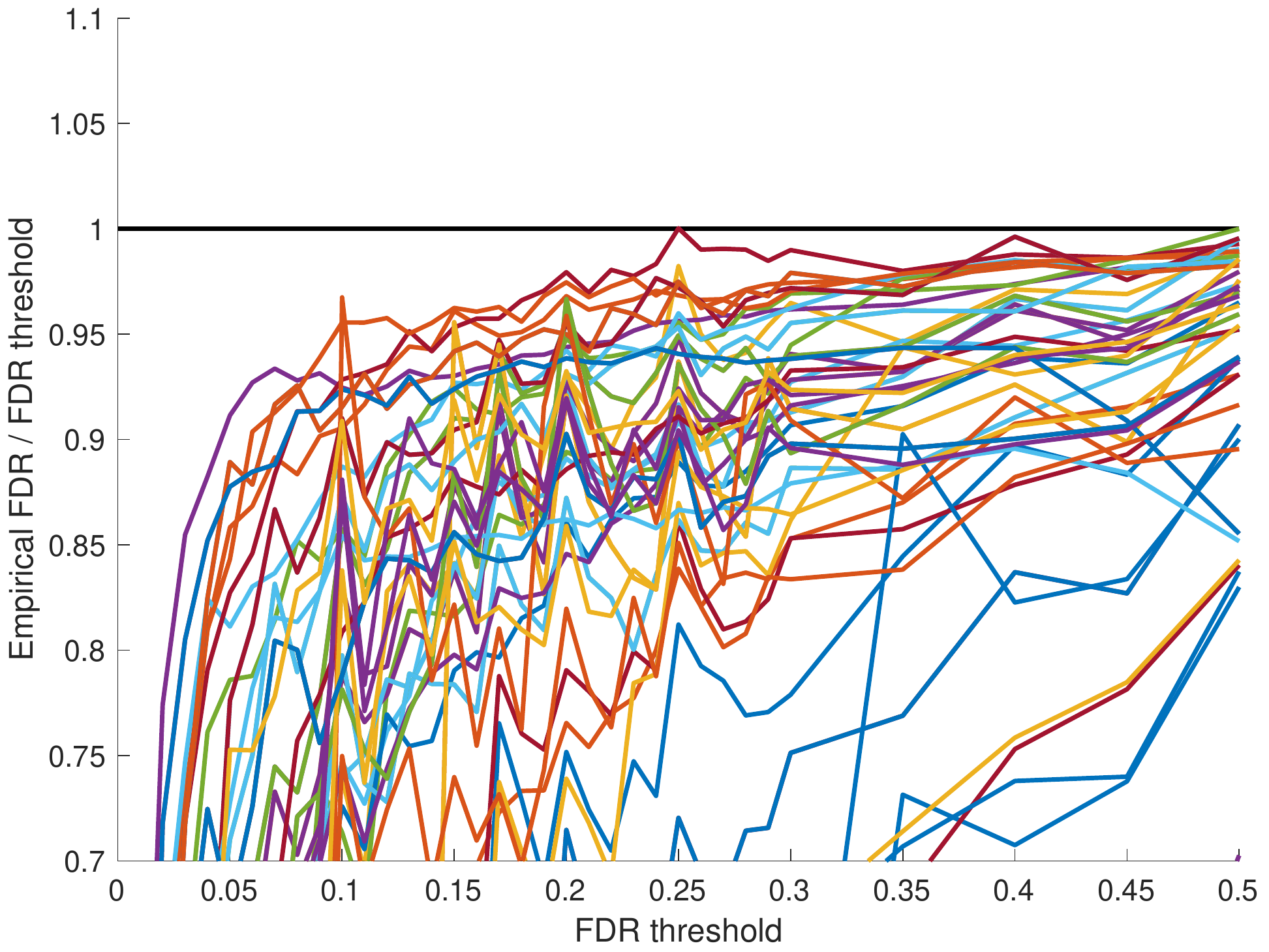} & \protect\includegraphics[width=3in]{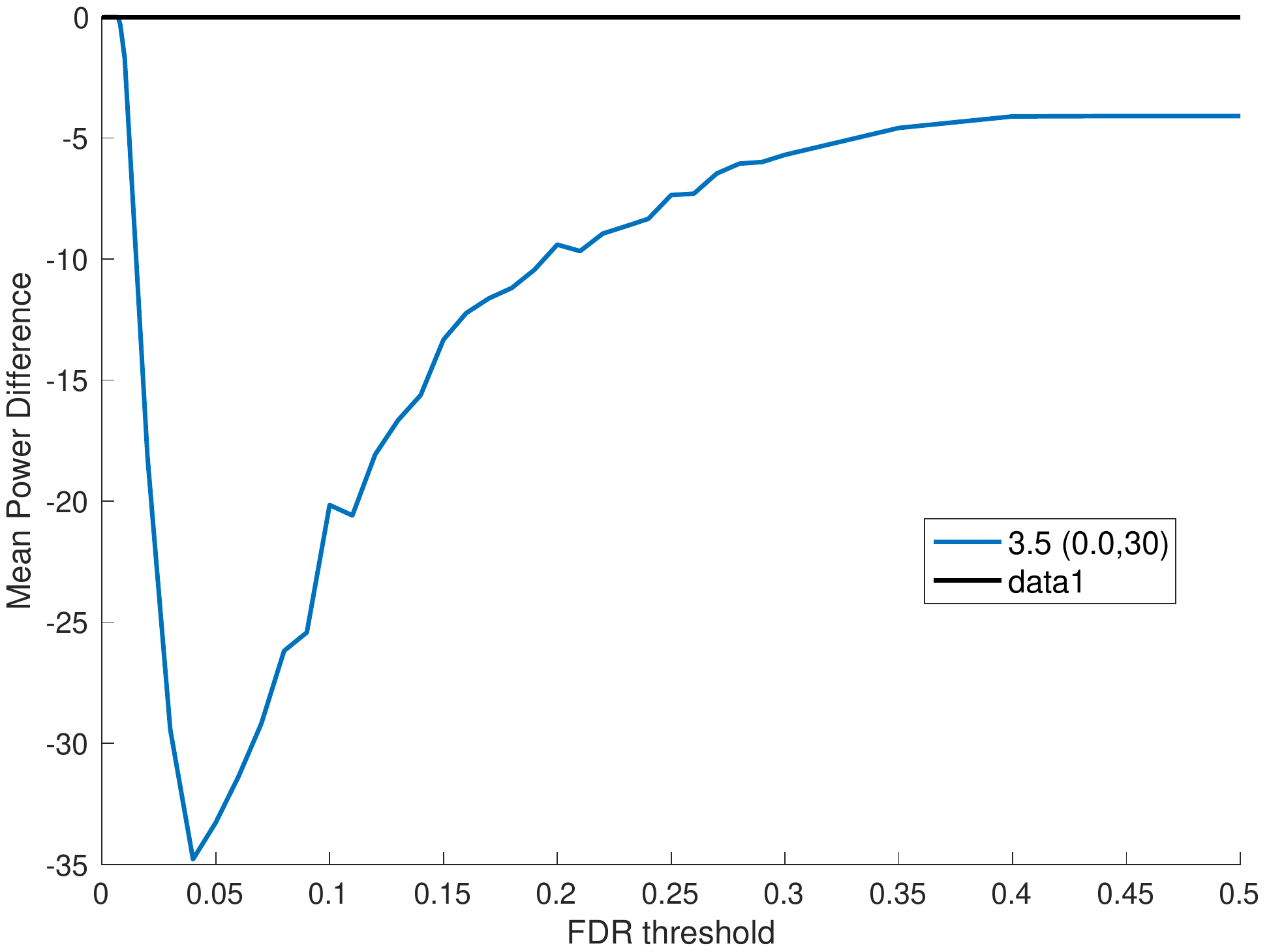}\tabularnewline
		C. Empirical FDR  (multi-knockoff-select) & \tabularnewline
		\protect\includegraphics[width=3in]{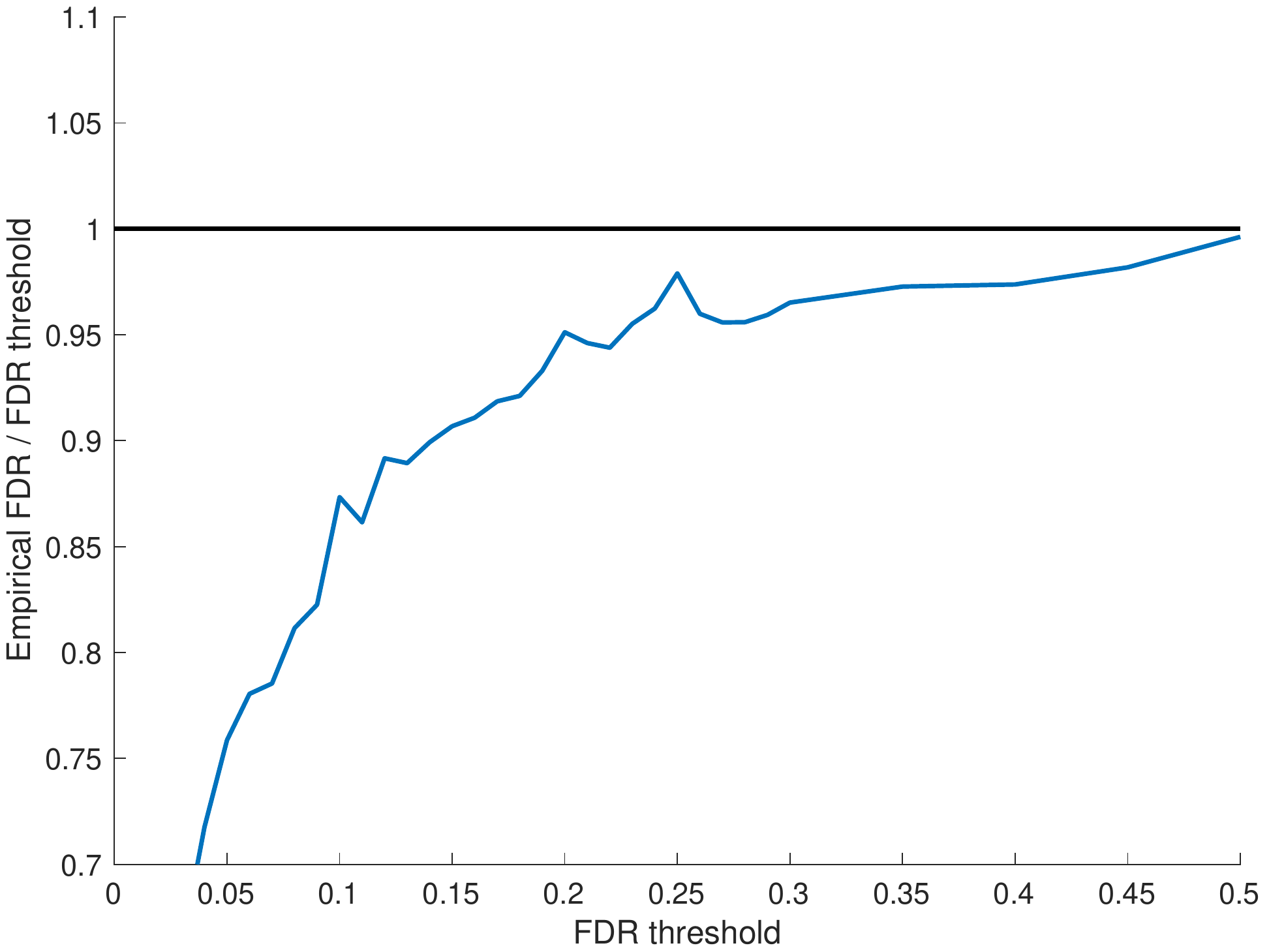} & \tabularnewline
	\end{tabular}\protect\caption{\textbf{More on multi-knockoff-select.} (A) The FDR of multi-knockoff-select
		is empirically controlled on the wide set of parameter values of the
		combined set (\suppsec\ref{subsec:Combined-dataset}). (B) Multi-knockoff-select
		(with $d\in\left\{ 1,3,7\right\} $ and $b=50$ batches) is uniformly
		more powerful than knockoff+ on the $n=3000$, $p=1000$, $K=30$,
		$A=3.5$, $\Theta_{0}=I_{p}$ example of \cite{barber:controlling}.
		(C) Empirical FDR of multi-knockoff-select on the same data used in
		panel B.\label{fig:supp_mKOsel_more}}
\end{figure}


\begin{figure}[h]
	\centering %
	\begin{tabular}{ll}
		A. $\rho=0$ & B. $\rho=0.5$\tabularnewline
		\protect\includegraphics[width=3in]{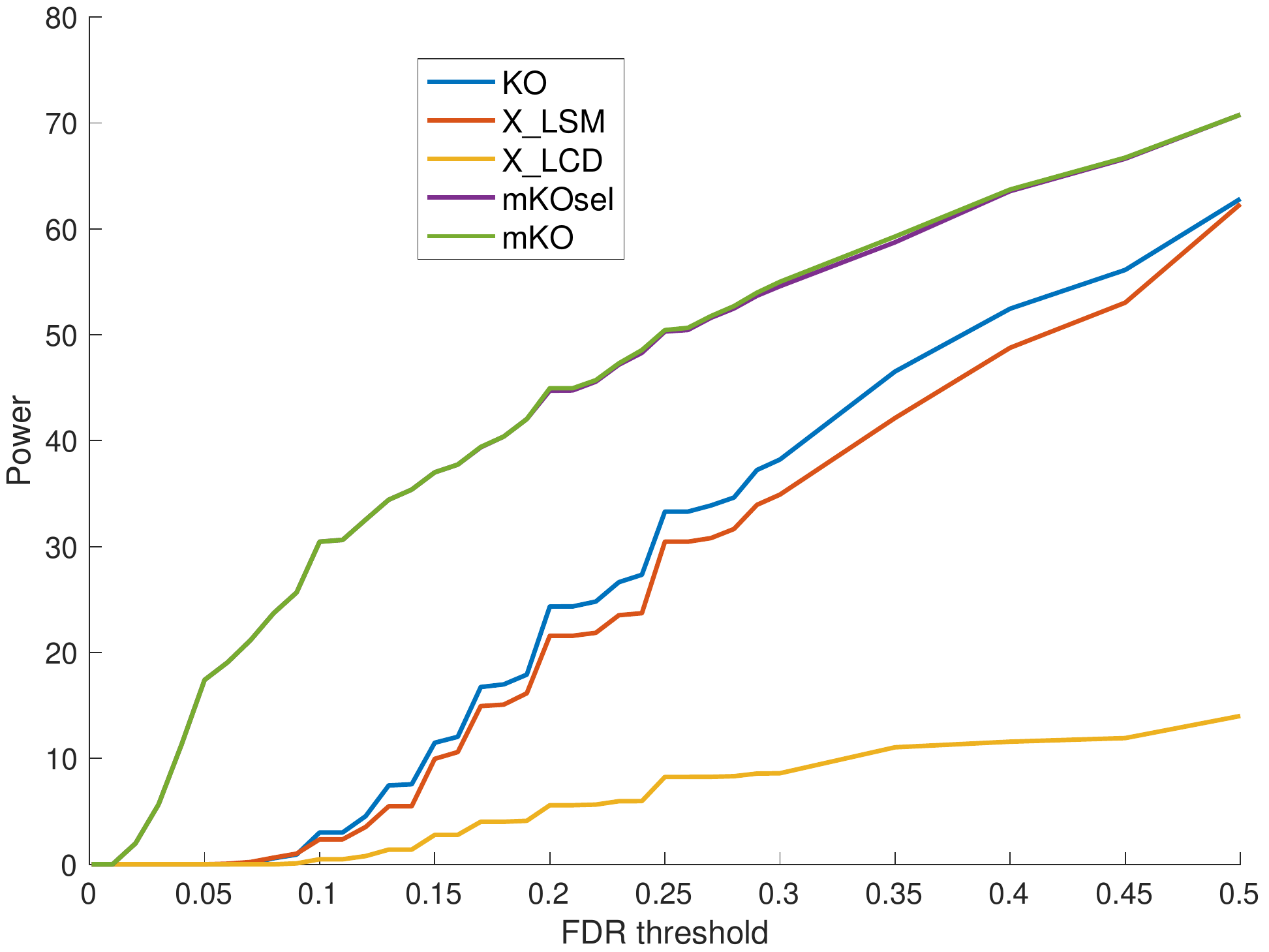} & \protect\includegraphics[width=3in]{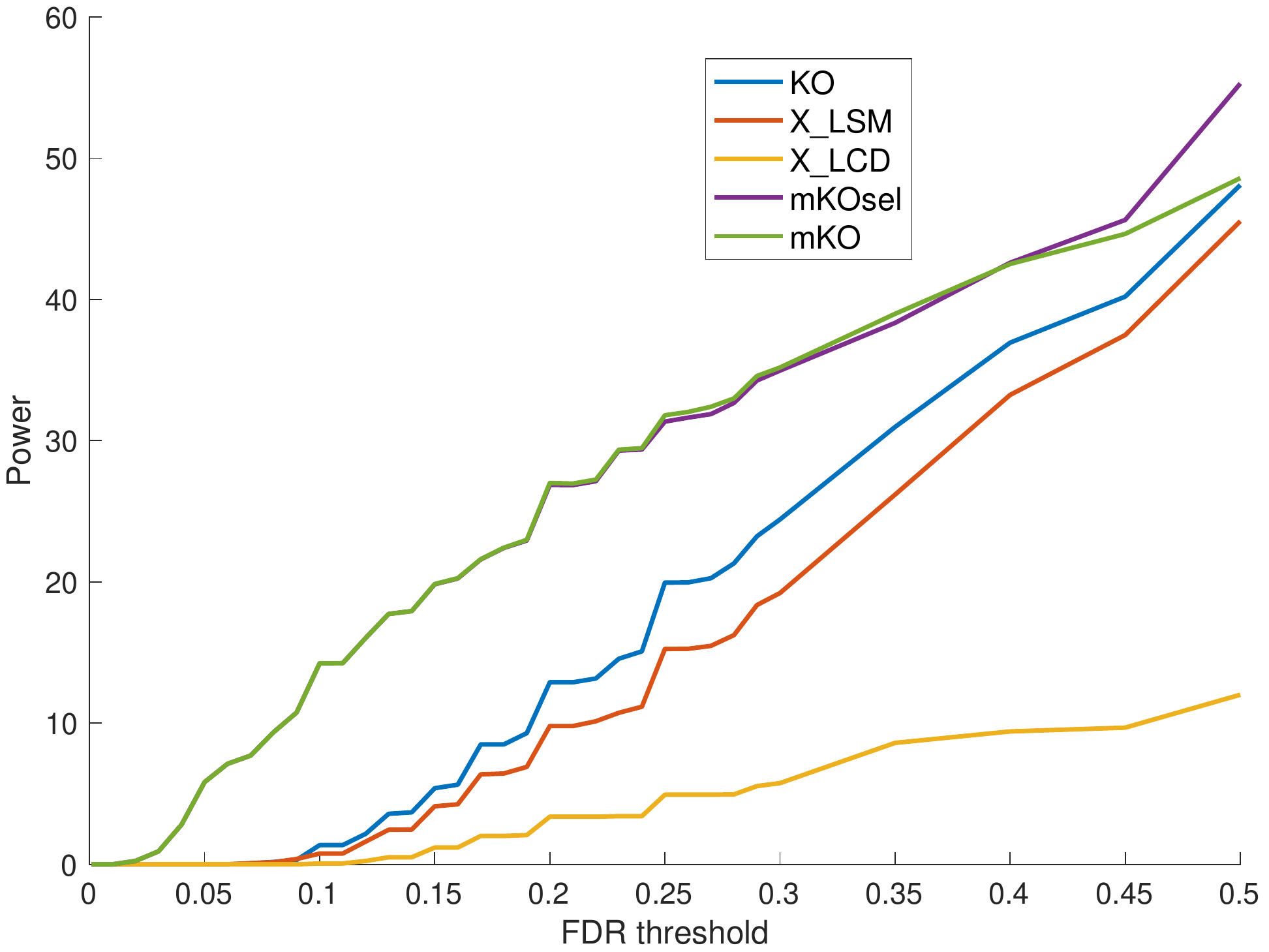}\tabularnewline
		C. $\rho=0.9$ & D. Empirical FDR ($\rho=0$)\tabularnewline
		\protect\includegraphics[width=3in]{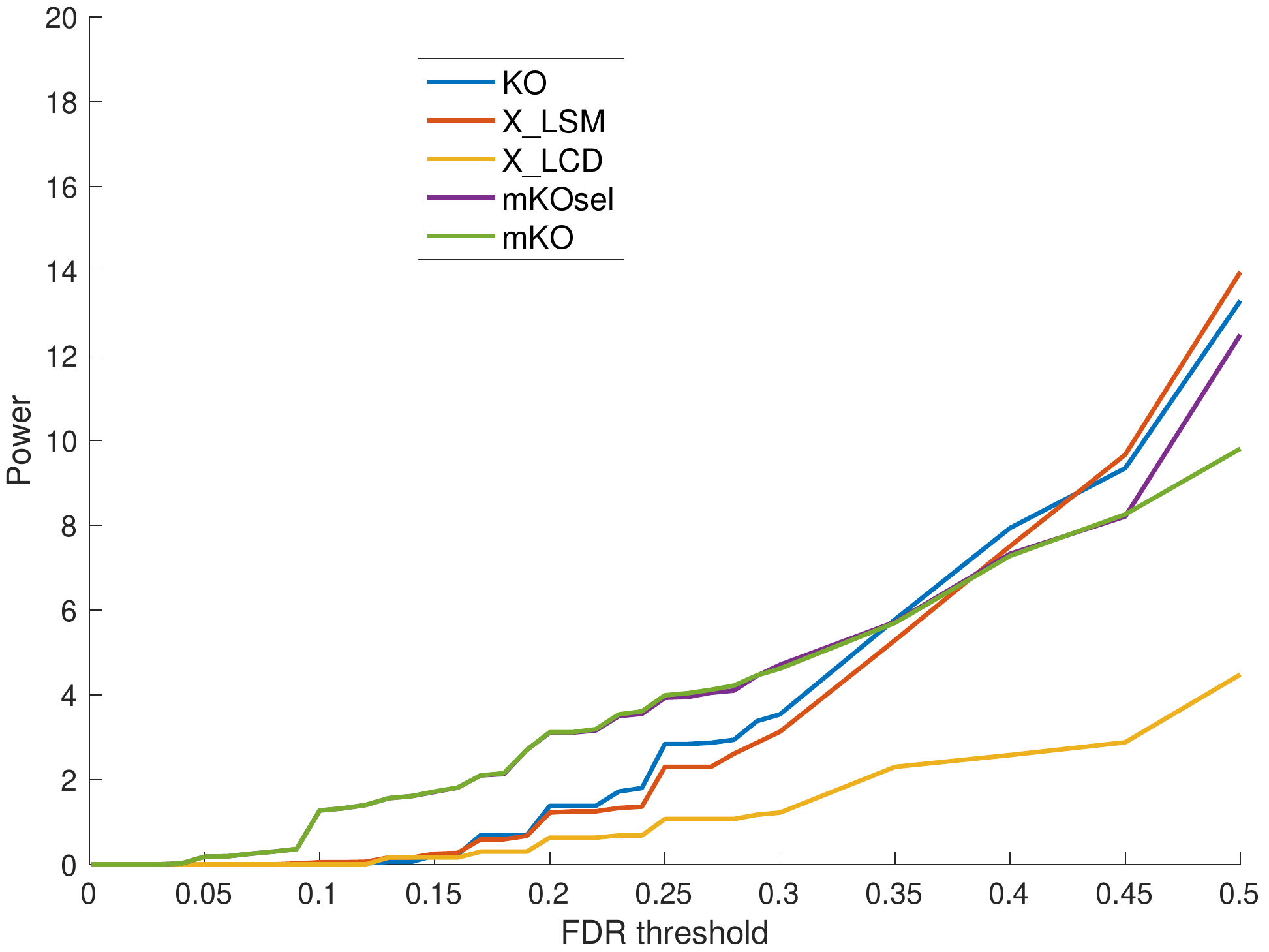} & \protect\includegraphics[width=3in]{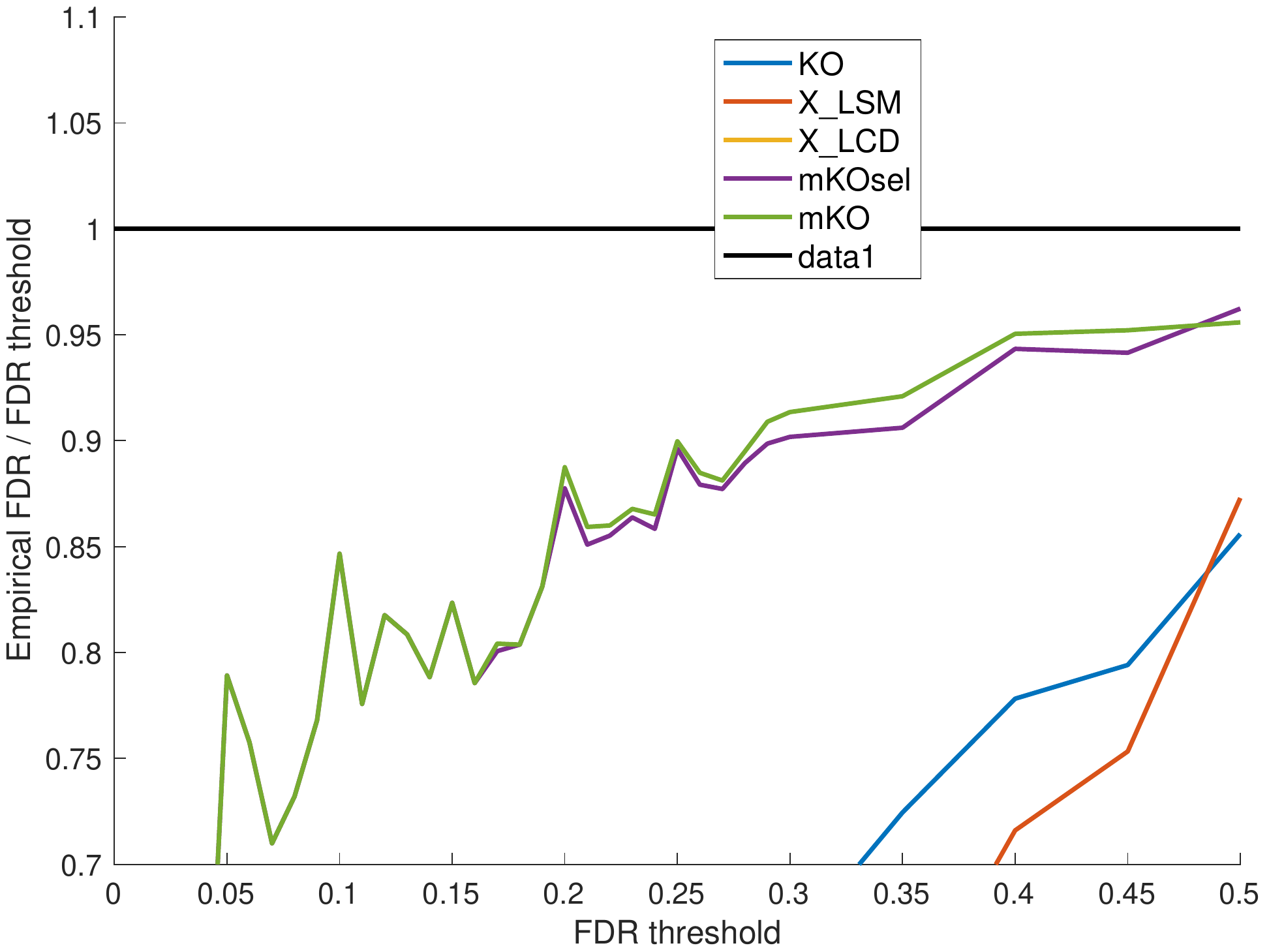}\tabularnewline
	\end{tabular}\protect\caption{\textbf{Comparison with model-X knockoffs. }In this figure we examine
		the performance of the model-X knockoffs. The data consisted of $n=600$,
		$p=200$, $K=10$, $A=2.8$ and a varied feature correlation strength
		as described in \suppsec\ref{subsec:The-model-X-dataset}. (A-C)
		The model-X LSM is on-par or below knockoff+ while the model-X LCD
		is generally significantly behind both. At the same time multi-knockoff
		and multi-knockoff-select dominates the single knockoff methods except
		when the FDR threshold is $\ge0.35$ and the feature correlation is
		very high ($\rho=0.9$). (D) For this case of $\rho=0$ the empirical
		FDR of LCD is always below 70\% of the threshold hence it does not appear in the plot.
		Similarly, for $\rho=0.5,0.9$ the empirical FDR
		of all the methods was consistently significantly below the threshold
		so we omitted those figures.\label{fig:supp_modelX}}
\end{figure}

\clearpage
	
\end{document}